\documentclass[eqsecnum,superscriptaddress,nofootinbib,preprint,11pt]{revtex4-1}

\usepackage{amsmath} 
\usepackage{amssymb} 
\usepackage{amsthm} 
\usepackage{bbm} 
\usepackage{tikz} 
\usepackage{enumerate} 

\usepackage{footnotebackref} 

\newcommand{\abs}[1]{{\left|{#1}\right|}} 
\newcommand{\inner}[2]{{\langle {#1}\vert {#2} \rangle}} 
\newcommand{\ket}[1]{\vert{#1}\rangle} 
\newcommand{\bra}[1]{\langle{#1}\vert} 

\newcommand{\secref}[1]{Sec.~\ref{#1}}
\newcommand{\eqnref}[1]{(\ref{#1})}
\newcommand{\figref}[1]{Fig.~\ref{#1}}
\newcommand{\appref}[1]{Appendix~\ref{#1}}
\newcommand{\footref}[1]{Footnote~\ref{#1}}
\newcommand{\thmref}[1]{Theorem~\ref{#1}}

\theoremstyle{definition}
\newtheorem{theorem}{Theorem}

\begin{document}
\count\footins = 1000 

\title
{An elementary proof and detailed investigation of the bulk-boundary correspondence in the generic two-band model of Chern insulators}\thanks{This paper is a follow-up to \cite{chen2017elementary}, various material of which is repeated here.}


\author{Bo-Hung Chen}
\email{kenny81778189@gmail.com}
\affiliation{Department of Physics and Center for Theoretical Physics, National Taiwan University, Taipei 10617, Taiwan}

\author{Dah-Wei Chiou}
\email{dwchiou@gmail.com}
\affiliation{Department of Physics, National Sun Yat-sen University, Kaohsiung 80424, Taiwan}
\affiliation{Center for Condensed Matter Sciences, National Taiwan University, Taipei 10617, Taiwan}


\begin{abstract}
With the inclusion of arbitrary long-range hopping and (pseudo)spin-orbit coupling amplitudes, we formulate a generic model that can describe \emph{any} two-dimensional two-band bulk insulators, thus providing a simple framework to investigate \emph{arbitrary} adiabatic deformations upon the systems of any \emph{arbitrary} Chern numbers. Without appealing to advanced techniques beyond the standard methods of solving linear difference equations and applying Cauchy's integral formula, we obtain a mathematically elementary yet rigorous proof of the bulk-boundary correspondence on a strip, which is robust against any adiabatic deformations upon the bulk Hamiltonian and any uniform edge perturbation along the edges. The elementary approach not only is more transparent about the underlying physics but also reveals various intriguing nontopological features of Chern insulators that have remained unnoticed or unclear so far. Particularly, if a certain condition is satisfied (as in most renowned models), the loci of edge bands in the energy spectrum and their (pseudo)spin polarizations can be largely inferred from the bulk Hamiltonian alone without invoking any numerical computation for the energy spectrum of a strip.
\end{abstract}


\maketitle


\section{Introduction}\label{sec:introduction}
Topological states of condensed matter have been intensively investigated and become a rapidly developing
area of research in recent years. One of the central concepts of topological matter is the \emph{bulk-boundary correspondence}, which posits that, for a variety of systems, the edge modes on the boundary are characterized by topological invariants of the physics in the bulk --- i.e., the topological phases of matter are manifested in terms of robust edge modes.
The first example of this correspondence is the integer quantum Hall effect, discovered in 1980, of which the explanation was proposed by Laughlin in 1981 \cite{laughlin1981quantized}.
Since then, the bulk-boundary correspondence has been revealed in numerous experiments and numerical simulations (see \cite{hasan2010colloquium,qi2011topological,asboth2016short} for reviews). Meanwhile, theoretical accounts of topological states have been developed in different approaches with various degrees of rigor (see e.g.\ \cite{hatsugai1993chern,hatsugai2009bulk,essin2011bulk,graf2013bulk,rudner2013anomalous,cano2014bulk} and more references in \cite{hasan2010colloquium,qi2011topological,asboth2016short}) and led to a hierarchical classification of topological condensed matter systems \cite{kitaev2009periodic,ryu2010topological}.

A mathematically rigorous proof of the robustness of the bulk-boundary correspondence is generally challenging. For topological systems with band structures, such as Chern insulators and topological insulators,\footnote{Chern insulators are band insulators that exhibit nontrivial Chern numbers and break time-reversal symmetry, whereas topological insulators are topologically nontrivial band insulators that preserve time-reversal symmetry. In the literature, the term ``topological insulator'' is also occasionally used in the broader sense to refer to any topologically nontrivial band insulators (Chern insulators included), regardless of the time-reversal symmetry.} the main difficulty lies in the fact that the notions of edge modes and topological invariants are inherently anchored to two conflicting settings and hence cannot be retained simultaneously in a single setting. Rigorously speaking, only in the explicit presence of boundaries can one make sense of edge modes. On the other hand, the topological invariants of band insulators are defined upon the bulk Brillouin zone, which make sense only if the system is without explicit boundaries and thus respects the full lattice translational symmetry --- i.e., either the system is infinite in all dimensions or it is finite in some dimensions but imposed with the periodic (Born-von Karman) boundary condition.
Secondly, a rigorous proof has to consider \emph{arbitrary} topological invariants with \emph{arbitrary} adiabatic deformations,\footnote{The Hamiltonian of a bulk insulator is said to be \emph{adiabatically deformed}, if its bulk energy bands are continuously deformed while the bulk gap remains open and the essential symmetries (if any) remain respected.} but it is difficult to incorporate all of them into a single framework.

Various advanced approaches have been developed to provide firm mathematical foundation for understanding and proving the bulk-boundary correspondence. Among them, the approach of the $K$-theory is perhaps the most powerful formalism (see \cite{prodan2016bulk} for a review); it is powerful and broad in scope in the sense that it rigorously proves the bulk-boundary correspondence for a wild range of different topological condensed matter systems and provides a systematic scheme to classify them. Advanced approaches, however, heavily involve advanced mathematical technicalities and are often not very transparent about the underlying mechanism and various detailed features.
Contrary and complementary to advanced approaches, considerable effort has been devoted to understanding the bulk-boundary correspondence from a more elementary perspective (see \cite{pershoguba2012shockley,rhim2017bulk,rhim2018unified} for recent examples). Even if less broad or less powerful than the well-established approaches, an elementary approach can be very valuable, as it may still provide a new perspective of known physics and even reveal new physics.
In our previous work \cite{chen2017elementary}, by generalizing the Su-Schrieffer-Heeger (SSH) model \cite{su1979solitons} with arbitrary long-range hopping amplitudes, we provided a simple framework that takes into account any arbitrary adiabatic deformations upon the systems of any arbitrary winding numbers, and offered a mathematically rigorous proof of the bulk-boundary correspondence for the generalized SSH model without appealing to any advance techniques beyond the standard methods of solving linear difference equations and applying Cauchy's integral formula.

In this paper, extending the treatment of \cite{chen2017elementary} from the one-dimensional generalized SSH model to a two-dimensional case, we aim to give a rigorous yet elementary proof of the bulk-boundary correspondence in the generic two-band model of Chern insulators. With the inclusion of arbitrary long-range hopping and (pseudo)spin-orbit coupling amplitudes, the two-band model we construct is broadly generic to the extent that it can describe \emph{any} two-dimensional two-band bulk insulators. Many renowned models, such as the Rice-Mele model \cite{rice1982elementary}, the Haldane model \cite{haldane1988model}, and the Qi-Wu-Zhang model \cite{qi2006topological}, can be viewed as special cases of our generic model. The two-dimensional two-band model is much richer in structure than the one-dimensional generalized SSH model. Nevertheless, the techniques devised in \cite{chen2017elementary} can be carried over essentially under a dimension-reduction scheme that recasts the Chern number of the two-dimensional Brillouin zone into a sum of winding numbers of various one-dimensional loops in the Brillouin zone.

Our elementary approach not only elucidates how the topological nontriviality gives rise to edge modes but also uncovers various intriguing nontopological features that have remained unknown or unclear so far.
(i) As opposed to the bulk states on a strip,\footnote{See \footref{foot:tomography} for how the width of a strip affects the bulk states.} the wavefunctions and energies of the edge states are independent of the width of the strip as long as the width is large enough (so that the finite size effect is negligible).
(ii) If a certain condition (called the ``semi-special'' condition in this paper) is satisfied (as in most renowned models), the loci of edge bands (except for those induced or modified by edge perturbation) in the energy spectrum can be directly inferred from the bulk momentum-space Hamiltonian alone without invoking any full-fledged numerical computation for the energy spectrum of a strip. (iii) As a consequence, the condition for having degenerate edge bands is also found. (iv) We obtain a precise description and a clear understanding of the phenomenon of ``spin-momentum locking'' --- namely, in edge states of a strip, the (pseudo)spin is polarized to a unique direction associated with the edge mode momentum. However, contrary to what many assume, the notion of spin-momentum locking is not a topological feature in the strict sense (i.e., robust under arbitrary adiabatic deformations); it is robust only under deformations within the confines of the semi-special condition. (v) While it is well known that the bulk-boundary correspondence is robust against any edge perturbation that is uniformly imposed along the edges of a strip, we obtain a more detailed picture about how different kinds of edge perturbation deform and induce edge modes differently. (vi) Under certain circumstances (as in many renowned models), the energy spectrum (both bulk and edge modes included) of a strip exhibits a symmetric feature that the energy eigenvalues appears in pairs with opposite signs. We elaborate and compare two different symmetries giving this feature, which become identical in the absence of edge perturbation.

This paper is organized as follows. The generic two-dimensional two-band model is first formulated and elaborated in \secref{sec:two-band model}. Its topology is then analyzed in detail in \secref{sec:topology of the model}, with the emphasis that the Chern number can be cast in terms of winding numbers. In \secref{sec:correspondence on a strip}, we give a detailed description and a rigorous proof of the bulk-boundary correspondence for the Chern insulator on a strip. To demonstrate the ideas and predictions of our approach, the numerical analyses of various concrete examples are presented in \secref{sec:examples}. Finally, the results of this work are summarized and discussed in \secref{sec:summary}.

\section{Generic two-band model of Chern insulators}\label{sec:two-band model}
We first formulate the generic two-band tight-binding model in a two-dimensional lattice.\footnote{The two-dimensional lattice is assumed to be generic, even though it is depicted as a square lattice such as in \figref{fig:strip} for illustrative purpose. The lattice momenta $k_x$ and $k_y$ in accord with the lattice are not necessarily perpendicular to each other, but they too are often depicted so. Furthermore, for convenience, $k_x$ and $k_y$ are conventionally rescaled to have $k_{x,y}\equiv k_{x,y}+2\pi$.} This model can describe any \emph{arbitrary} two-band bulk insulators (called ``Chern insulators'' if the corresponding Chern number is nonzero).

\subsection{Bulk momentum-space Hamiltonian}\label{sec:bulk momentum-space Hamiltonian}
To begin with, we neglect all boundary effects and study only the physics in the bulk. That is, we either consider an infinite system or impose the periodic (Born-von Karman) boundary condition.
In this idealized setting, if the lattice translational invariance is not broken (e.g., by external electromagnetic field), we have the lattice momentum $\mathbf{k}\equiv(k_x,k_y)$ as a good quantum number.
The system is described by the bulk Hamiltonian, which takes the form $\hat{H}_\mathrm{bulk}=\sum_\mathbf{k}\hat{H}(\mathbf{k})\ket{\mathbf{k}}\bra{\mathbf{k}}$ in the momentum space.
For a two-band tight-binding system, the bulk momentum-space Hamiltonian $\hat{H}(\mathbf{k})$ is formally given by
\begin{equation}\label{H(k)}
  \hat{H}(\mathbf{k}) := \bra{\mathbf{k}}\hat{H}_\mathrm{bulk}\ket{\mathbf{k}}
  =\sum_{\alpha,\beta\in\{\uparrow,\downarrow\}} \bra{\mathbf{k},\alpha}\hat{H}_\mathrm{bulk}\ket{\mathbf{k},\beta} \ket{\alpha}\bra{\beta},
\end{equation}
where the two-value variable $\alpha\in\{\uparrow,\downarrow\}$ accounts for the two-band degree of freedom and represents either real spin (i.e., spin-up and spin-down) or pseudospin (e.g., bipartite sublattice sites), depending on the underlying physics of the system.
Whatever the underlying physics is, $\hat{H}(\mathbf{k})$ is always given by a $2\times2$ hermitian matrix and hence takes the generic form:\footnote{We do not consider non-hermitian Hamiltonians, as they are nonphysical in the sense that their eigenvalues are not all real and the eigenstates are no longer orthogonal to one another. Nevertheless, a non-hermitian formalism is useful for describing transient physics --- typically, the dissipative or amplifying physics is neatly represented by the imaginary part of energy eigenvalues. When it comes to the physics of bulk-boundary correspondence in the ordinary sense, the energy spectrum under consideration is all based on the \emph{stationary} physics, not the \emph{transient} physics. Recently, there have been various works \cite{PhysRevLett.102.065703, PhysRevLett.116.133903, PhysRevLett.118.040401, xiong2018does} devoted to the bulk-boundary correspondence for non-hermitian systems, and the new development has broadened the concept of the bulk-boundary correspondence beyond the ordinary sense. In this paper, we restrict our investigation to the ordinary sense.}
\begin{equation}\label{H(kx,ky)}
  \hat{H}(\mathbf{k}) = \mathbf{h}(\mathbf{k})\cdot\boldsymbol{\sigma}
  + h_0(\mathbf{k})\mathbbm{1}_{2\times2}
  = \left(
      \begin{array}{cc}
        h_z(\mathbf{k})+h_0(\mathbf{k}) & h_x(\mathbf{k})-ih_y(\mathbf{k}) \\
        h_x(\mathbf{k})+ih_y(\mathbf{k}) & -h_z(\mathbf{k})+h_0(\mathbf{k})\\
      \end{array}
    \right).
\end{equation}
Obviously, the bulk energy spectrum is given by
\begin{equation}\label{E bulk}
E(\mathbf{k})=h_0(\mathbf{k})\pm\abs{\mathbf{h}(\mathbf{k})}.
\end{equation}
If the minimum of the upper band spectrum $h_0(\mathbf{k})+\abs{\mathbf{h}(\mathbf{k})}$ is larger than the maximum of the lower band spectrum $h_0(\mathbf{k})-\abs{\mathbf{h}(\mathbf{k})}$, the upper and lower bulk bands is gapped and the system is a bulk insulator.
For a given $\mathbf{k}$, $h_0(\mathbf{k})$ only offsets the energy. Therefore, the term $h_0(\mathbf{k}) \mathbbm{1}_{2\times2}$ is usually neglected, provided that $\abs{\nabla_\mathbf{k}h_0(\mathbf{k})}$ is small enough so the upper and lower bulk bands do not overlap.

Since the lattice momentum is periodic, i.e., $k_{x,y}\equiv k_{x,y}+2\pi$, the generic form of $h_{a=x,y,z,0}$ can be cast as
\begin{subequations}\label{h(k)}
\begin{eqnarray}
\label{h(k)a}
&&h_a(k_x,k_y) := \sum_{n_x,n_y=-\infty}^\infty w^a_{n_x,n_y} e^{i(n_xk_x+n_yk_y)},\\
\label{h(k)b}
&&\quad\text{where}\ w^{a*}_{n_x,n_y}=w^a_{-n_x,-n_y}\in\mathbb{C}.
\end{eqnarray}
\end{subequations}
As the Fourier series of \eqnref{h(k)} can represent any \emph{arbitrary} function mapping from $T^2\cong[0,2\pi]\times[0,2\pi]$ to $\mathbb{C}$ with $h_a(k_x+2m\pi,k_y+2n\pi)=h_a(k_x,k_y)$ for any $m,n\in\mathbb{Z}$, the form of \eqnref{h(k)} provides a starting point to study \emph{arbitrary} adiabatic deformations upon a system of an \emph{arbitrary} Chern number.\footnote{Our main goal is to obtain a \emph{mathematically rigorous} proof of the bulk-boundary correspondence. We are obliged to take into consideration \emph{all} arbitrary adiabatic deformations and arbitrary Chern numbers, even if the corresponding $\hat{H}(\mathbf{k})$ with arbitrary $h_a(\mathbf{k})$ is purely artificial and cannot be realized in a realistic system.}
Many renowned two-band models, such as the Rice-Mele model \cite{rice1982elementary}, the Haldane model \cite{haldane1988model}, and the Qi-Wu-Zhang model \cite{qi2006topological}, can be treated as special cases in this generic setting.
Particularly, if we set $w^x_{1,0}=w^{x*}_{-1,0}=-i/2$, $w^y_{0,1}=w^{y*}_{0,-1}=-i/2$, $w^z_{1,0}=w^{z*}_{-1,0}=1/2$, $w^z_{0,1}=w^{z*}_{0,-1}=1/2$, $w^z_{0,0}=u$, and $w^a_{n_x,n_y}=0$ otherwise, we have $h_0(\mathbf{k})=0$, $h_x(\mathbf{k})=\sin k_x$, $h_y(\mathbf{k})=\sin k_y$, and $h_z(\mathbf{k})=u+\cos k_x+\cos k_y$, which gives the Qi-Wu-Zhang model.

If we deal with a finite system with $N_x$ and $N_y$ unit cells in the $x$ and $y$ directions, respectively, $\mathbf{k}$ takes the discrete values $\mathbf{k}\in\{(m\delta_x,n\delta_y)|m=0,1,\dots,N_x-1,n=0,1,\dots,N_y-1\}$ with $\delta_{x,y}=2\pi/N_{x,y}$. It is only an approximation to treat $h_a(\mathbf{k})$ as a continuous map when $N_x$ and $N_y$ are large but finite. To make this approximation sensible, the map $h_a(\mathbf{k})$ has to be ``smooth'' enough, or more precisely, $\frac{\abs{\partial_{k_{x,y}}h_a(\mathbf{k})}}{\abs{h_a(\mathbf{k})}} \ll\frac{1}{\delta_{x,y}}$. This requires $\sum_{n_x,n_y=-\infty}^\infty$ to be truncated to $\sum_{n_x=-\bar{n}_x}^{\bar{n}_x}\sum_{n_y=-\bar{n}_y}^{\bar{n}_y}$ by two delimiting integers $\bar{n}_x,\bar{n}_y$ that satisfy $\bar{n}_x\ll N_x$ and $\bar{n}_y\ll N_y$.
In other words, to make sense of the bulk-boundary correspondence, the macroscopic length scales (indicated by $N_x$ and $N_y$) of a given sample has to be much larger than the upper bound (specified by $\bar{n}_x$ and $\bar{n}_y$) for the distance (specified by $n_x$ and $n_y$) of long-range interaction.\footnote{It will be made clear in the next subsection, \secref{sec:bulk real-space Hamiltonian}, that the coefficients $w^a_{n_x,n_y}$ are interpreted as long-range hopping and (pseudo)spin-orbit coupling amplitudes over the distance $(n_x,n_y)$.}

\subsection{Bulk real-space Hamiltonian}\label{sec:bulk real-space Hamiltonian}
To study the physics in the bulk for a finite system while neglecting the physics on the boundary, we impose the periodic boundary condition: i.e., $\ket{m_x+N_x,m_y,\alpha}\equiv\ket{m_x,m_y+N_y,\alpha}\equiv\ket{m_x,m_y,\alpha}$, where $\{(m_x,m_y)|,m_x,m_y\in\mathbb{Z}\}$ represents the two-dimensional lattice sites. As the periodic boundary condition respects the lattice translational invariance, Bloch's theorem applies. The Bloch's theorem allows us to introduce the plane wave basis states
\begin{eqnarray}\label{ket k}
  \ket{k_x}=\frac{1}{\sqrt{N_x}}\sum_{m_x =1}^{N_x} e^{im_x k_x} \ket{m_x}, \nonumber\\
  \ket{k_y}=\frac{1}{\sqrt{N_y}}\sum_{m_y =1}^{N_y} e^{im_y k_y} \ket{m_y},
\end{eqnarray}
so that the Bloch eigenstates, labeled by $\epsilon=\pm$ (for upper and lower bands) and $\mathbf{k}$, read as
\begin{eqnarray}
  \ket{\Psi_\epsilon(k_x, k_y)} &=& \ket{k_x}\otimes\ket{k_y}\otimes\ket{u_\epsilon(k_x,k_y)},\nonumber\\
  \ket{u_\epsilon(k_x,k_y)} &=& a_\epsilon(k_x,k_y)\ket{\uparrow}+b_\epsilon(k_x,k_y)\ket{\downarrow}.
\end{eqnarray}
The vectors $\ket{u_\epsilon(\mathbf{k})}$ are eigenstates of $\hat{H}(\mathbf{k})$ defined in \eqnref{H(k)}; i.e., $\hat{H}(\mathbf{k})\ket{u_\epsilon(\mathbf{k})}=E(\epsilon,\mathbf{k})\ket{u_\epsilon(\mathbf{k})}$.

Substituting \eqnref{ket k} into \eqnref{H(k)} with \eqnref{H(kx,ky)} and \eqnref{h(k)}, we obtain the bulk real-space Hamiltonian:
\begin{eqnarray}\label{H bulk}
  \hat{H}_\mathrm{bulk}
  &=& \sum_{m_x=1}^{N_x} \sum_{m_y=1}^{N_y}
  \sum_{n_x=-\bar{n}_x}^{\bar{n}_x}\sum_{n_y=-\bar{n}_y}^{\bar{n}_y}  \ket{m_x-n_x,m_y-n_y}\bra{m_x,m_y}\nonumber\\
  && \qquad \otimes
  \left(
  \begin{array}{cc}
        w^z_{n_x,n_y}+w^0_{n_x,n_y} & w^x_{n_x,n_y}-iw^y_{n_x,n_y} \\
        w^x_{n_x,n_y}+iw^y_{n_x,n_y} & -w^z_{n_x,n_y}+w^0_{n_x,n_y} \\
      \end{array}
  \right).
\end{eqnarray}
Therefore, the coefficients $w^{a}_{n_x,n_y}$ represent the coupling constants between the lattice sites $\ket{m_x,m_y}$ and $\ket{m_x-n_x,m_y-n_y}$. For $(n_x,n_y)=(0,0)$, $w^z_{0,0}+w^0_{0,0}$ and $w^z_{0,0}-w^0_{0,0}$ are the on-site potentials for the $\uparrow$ and $\downarrow$ states, respectively, whereas $w^x_{0,0}\pm iw^y_{0,0}$ corresponds to the on-site interaction between $\uparrow$ and $\downarrow$.
For $(n_x,n_y)\neq(0,0)$, $w^z_{n_x,n_y}+w^0_{n_x,n_y}$ and $w^z_{n_x,n_y}-w^0_{n_x,n_y}$ give the hopping amplitudes from $\ket{m_x,m_y}$ to $\ket{m_x-n_x,m_y-n_y}$ for $\uparrow$ and $\downarrow$, respectively. Meanwhile, $w^x_{n_x,n_y}+iw^y_{n_x,n_y}$ and $w^x_{n_x,n_y}-iw^y_{n_x,n_y}$ correspond to the (pseudo)spin precession from $\uparrow$ to $\downarrow$ and from $\downarrow$ ro $\uparrow$, respectively, when the spinor hops from $\ket{m_x,m_y}$ to $\ket{m_x-n_x,m_y-n_y}$.
The fact that the (pseudo)spin precession depends on the hopping variables $n_x,n_y$ is called the ``(pseudo)spin-orbit coupling''.
As we have taken into account all arbitrary amplitudes $w^{a}_{n_x,n_y}$ with arbitrary interaction distances (specified by $n_x$ and $n_y$), our generic two-band model of Chern insulators is said to include arbitrary long-range hopping and (pseudo)spin-orbit coupling amplitudes.\footnote{By contrast, most renowned models consider hopping and (pseudo)spin-orbit coupling only up to the nearest neighbor or next-nearest neighbor extent. For example, in the Qi-Wu-Zhang model, as $w^{a}_{n_x,n_y}=0$ for $\abs{n_x},\abs{n_y}\geq2$, it does not include any long-range amplitudes beyond the nearest neighbor interaction.}

As commented in the end of \secref{sec:bulk momentum-space Hamiltonian}, to make sense of the smooth approximation of $h_a(\mathbf{k})$ for a finite system, we have to introduce two positive integers $\bar{n}_x$ and $\bar{n}_y$ as the upper bounds for the long-range amplitudes in the $x$ and $y$ directions, respectively. That is, we need to assume $w^{a}_{n_x,n_y}\approx0$ as long as $\abs{n_x}>\bar{n}_x$ or $\abs{n_y}>\bar{n}_y$. This is a reasonable assumption, because long-range amplitudes should become inappreciable when the distance of interaction becomes very large. Smoothness of $h_a(\mathbf{k})$ demands that the macroscopic length scales, $N_x$ and $N_y$, of a finite system has to be much larger than the longest distance of appreciable long-range amplitudes, i.e., $\bar{n}_x\ll N_x$ and $\bar{n}_y\ll N_y$.

\subsection{$h_0=0$ symmetry}\label{sec:h0=0 symmetry}
In the case that $h_0(\mathbf{k})=0$, the energy spectrum is given by $E(\mathbf{k}) =h_0(\mathbf{k})\pm\abs{\mathbf{h}(\mathbf{k})} =\pm\abs{\mathbf{h}(\mathbf{k})}$, and consequently the upper and lower bands exhibit the symmetry of opposite eigenvalues of energy.
More precisely, if $\ket{u(\mathbf{k})} \equiv\inner{\mathbf{k}}{\Psi} =(a(\mathbf{k}),b(\mathbf{k}))^\mathrm{T}$ is an eigenstate with the eigenvalue $E(\mathbf{k})$ of $\hat{H}(\mathbf{k})$ given by \eqnref{H(kx,ky)} with $h_0=0$, we have
\begin{equation}
  \hat{H}(\mathbf{k})\left( \begin{array}{c}
                              a(\mathbf{k}) \\
                              b(\mathbf{k})
                            \end{array} \right)
  \equiv\left( \begin{array}{cc}
             h_z   & h_x-ih_y \\
             h_x+ih_y & -h_z
          \end{array} \right)
   \left( \begin{array}{c}
             a \\
             b
          \end{array} \right)
  =\left( \begin{array}{c}
             h_z a  + (h_x-ih_y)b \\
             (h_x+ih_y)a - h_z b
          \end{array} \right)
  =E(\mathbf{k})\left( \begin{array}{c}
                            a(\mathbf{k}) \\
                            b(\mathbf{k})
                          \end{array} \right),
\end{equation}
which follows
\begin{eqnarray}
  \hat{H}(\mathbf{k})\left( \begin{array}{c}
                              b(\mathbf{k})^* \\
                              -a(\mathbf{k})^*
                            \end{array} \right)
  &\equiv&\left( \begin{array}{cc}
             h_z   & h_x-ih_y \\
             h_x+ih_y & -h_z
          \end{array} \right)
   \left( \begin{array}{c}
             b^* \\
             -a^*
          \end{array} \right)
  =\left( \begin{array}{c}
             h_z b^*  - (h_x-ih_y)a^* \\
             (h_x+ih_y)b^* + h_z a^*
          \end{array} \right) \nonumber\\
  &=&\left( \begin{array}{c}
             \left(h_z b  - (h_x+ih_y)a\right)^* \\
             \left((h_x-ih_y)b + h_z a\right)^*
          \end{array} \right)
  =-E(\mathbf{k})\left( \begin{array}{c}
                            b(\mathbf{k})^* \\
                            -a(\mathbf{k})^*
                          \end{array} \right).
\end{eqnarray}
Consequently, $\ket{\tilde{u}(\mathbf{k})} \equiv\inner{\mathbf{k}}{\tilde{\Psi}} =p(\mathbf{k})(-b(\mathbf{k})^*,a(\mathbf{k})^*)^\mathrm{T}$ is an eigenstate of $\hat{H}(\mathbf{k})$ with the opposite eigenvalue $-E(\mathbf{k})$. Here, for generality, we also include an arbitrary phase factor $p(\mathbf{k})$ that satisfies $\abs{p(\mathbf{k})}=1$ and $p(\mathbf{k}+(2\pi,0)) = p(\mathbf{k}+(0,2\pi)) = p(\mathbf{k})$.
Particularly, we can choose $p(\mathbf{k})=e^{i\mathbf{k}\cdot \bar{\mathbf{m}}}\equiv e^{i(k_x\bar{m}_x+k_y\bar{m}_y)}$ with two arbitrary integers $\bar{m}_x, \bar{m}_y$.

Correspondingly, in the real-space representation, if $\inner{\mathbf{m}}{\Psi} \equiv(a_{m_x,m_y},b_{m_x,m_y})^\mathrm{T}$, then we have
\begin{eqnarray}\label{tilde a b}
\inner{\mathbf{m}}{\tilde{\Psi}}
&\equiv&
\left(\begin{array}{c}
          \tilde{a}_{m_x,m_y} \\
          \tilde{b}_{m_x,m_y}
      \end{array} \right)
=\sum_\mathbf{k} \inner{\mathbf{m}}{\mathbf{k}} \inner{\mathbf{k}}{\tilde{\Psi}}
=\sum_\mathbf{k} \frac{e^{i(\mathbf{m}+\bar{\mathbf{m}})\cdot\mathbf{k}}}{\sqrt{N_xN_y}}
\left( \begin{array}{c}
         b(\mathbf{k})^* \\
         -a(\mathbf{k})^*
       \end{array} \right) \nonumber\\
&=& \left(
\sum_\mathbf{k} \frac{e^{-i(\mathbf{m}+\bar{\mathbf{m}})\cdot\mathbf{k}}}{\sqrt{N_xN_y}}
(i\sigma_y)
\left( \begin{array}{c}
         a(\mathbf{k}) \\
         b(\mathbf{k})
       \end{array} \right)
\right)^*
=i\sigma_y
\left(
\sum_\mathbf{k} \frac{e^{-i(\mathbf{m}+\bar{\mathbf{m}})\cdot\mathbf{k}}}{\sqrt{N_xN_y}}
\inner{\mathbf{k}}{\Psi}
\right)^* \nonumber\\
&=& i\sigma_y \inner{-\mathbf{m}-\bar{\mathbf{m}}}{\Psi}^*
=
\left(\begin{array}{c}
          b_{-m_x-\bar{m}_x,-m_y-\bar{m}_y}^* \\
          -a_{-m_x-\bar{m}_x,-m_y-\bar{m}_y}^*
      \end{array} \right).
\end{eqnarray}
It is because of the lattice translational invariance that the relation between $(a_\mathbf{m},b_\mathbf{m})$ and $(\tilde{a}_\mathbf{m},\tilde{b}_\mathbf{m})$ is up to an arbitrary lattice vector $\bar{\mathbf{m}}\equiv(\bar{m}_x,\bar{m}_y)$.

Although the condition $h_0(\mathbf{k})=0$ is artificial and at best an approximation in reality, imposing $h_0=0$ is very helpful for finding various qualitative features of the two-band system.

\section{Topology of the two-band model}\label{sec:topology of the model}
The bulk momentum-space Hamiltonian \eqnref{H(kx,ky)} is specified by the map $h_a: \mathbf{k}\equiv(k_x,k_y)\in T^2\equiv [0,2\pi]\times[0,2\pi] \mapsto (\mathbf{h}(\mathbf{k}),h_0(\mathbf{k}))\in \left(\mathbb{R}^3\setminus\{0\}\right)\times\mathbb{R}$, where $\mathbf{h}(\mathbf{k})=0$ is excluded to ensure an open gap of the bulk spectrum. The topology (more precisely, homotopy type) of the map $h_a:T^2\rightarrow (\mathbb{R}^3\setminus\{0\})\times\mathbb{R}$ can be characterized by the map $\mathbf{h}:T^2\rightarrow \mathbb{R}^3\setminus\{0\}$. Therefore, without losing generality, we assume $h_0(\mathbf{k})=0$ and focus on $\mathbf{h}(\mathbf{k})$. In this section, we will prove that the topology of $\mathbf{h}:T^2\rightarrow \mathbb{R}^3\setminus\{0\}$ is classified by the Chern number and elaborate on its geometrical meaning. We adopt some of the techniques used in \cite{sun2013lecture}.

\subsection{Chern number and winding number}
In the case of $h_0=0$, the eigenvalues of $\hat{H}(\mathbf{k})$ are given by $E_\pm(\mathbf{k}) =\pm\abs{\mathbf{h}(\mathbf{k})}$, and the two eigenstates $u_{\pm}(\mathbf{k})$ are related with each other by the $h_0=0$ symmetry as discussed in \secref{sec:h0=0 symmetry}.
In the following, without losing generality, we focus only on the lower band $u_-(\mathbf{k})$.

The eigenstate $u_-(\mathbf{k})$ of $E_-(\mathbf{k})$ can be expressed as
\begin{equation}\label{uS}
u^{(S)}_-(\mathbf{k}) = \frac{1}{\mathcal{N}^{(S)}}
\left(
  \begin{array}{c}
    h_z-\abs{\mathbf{h}} \\
    h_x+ih_y \\
  \end{array}
\right),
\quad \text{singular when}\ h_x=h_y=0\ \text{and}\ h_z>0,
\end{equation}
or alternatively as
\begin{equation}\label{uN}
u^{(N)}_-(\mathbf{k}) = \frac{1}{\mathcal{N}^{(N)}}
\left(
  \begin{array}{c}
    -h_x+ih_y \\
    h_z+\abs{\mathbf{h}} \\
  \end{array}
\right),
\quad \text{singular when}\ h_x=h_y=0\ \text{and}\ h_z<0.
\end{equation}
Note that the wavefunction $u^{(S)}_-(\mathbf{k})$ is everywhere well defined except that $\mathbf{h}(\mathbf{k})$ is in the ``north-pole'' direction, whereas the wavefunction $u^{(N)}_-(\mathbf{k})$ is everywhere well defined except that $\mathbf{h}(\mathbf{k})$ in the ``south-pole'' direction. In the region where both wavefunctions are well defined, they differ with each other by a phase:
\begin{equation}\label{differ by a gauge}
u^{(N)}_-(\mathbf{k}) = u^{(S)}_-(\mathbf{k})\, e^{i\phi(\mathbf{k})},
\end{equation}
where
\begin{equation}
e^{i\phi(\mathbf{k})} =
\frac{\frac{h_z+\abs{\mathbf{h}}}{h_x+ih_y}}
     {\abs{\frac{h_z+\abs{\mathbf{h}}}{h_x+ih_y}}}
=\frac{\abs{h_x+ih_y}}{h_x+ih_y}.
\end{equation}
This corresponds to a gauge transformation for the Berry connection given by
\begin{equation}
{A}^{(N)}(\mathbf{k}) = {A}^{(S)}(\mathbf{k}) + d\mathbf{k}\,\partial_\mathbf{k}\phi(\mathbf{k}),
\end{equation}
where the Berry connection is defined as
\begin{equation}\label{Berry connection}
{A}(\mathbf{k})
:=-i\bra{u_-(\mathbf{k})}\partial_\mathbf{k}\ket{u_-(\mathbf{k})}
d\mathbf{k}.
\end{equation}
The corresponding Berry curvature, which is gauge independent, is defined as the exterior derivative of the Berry connection:
\begin{equation}\label{Berry curvature}
{\Omega}(\mathbf{k}) := d {A}(\mathbf{k}).
\end{equation}
Let the Brillouin zone $T^2$ be covered by $D_S\subset T^2$ and $D_N\subset T^2$: $D_S$ denotes a region (maybe disconnected) where $u^{(S)}_-$ is regular everywhere, $D_{N}\subset T^2$ denotes a region where $u^{(N)}_-$ is regular everywhere, and $\partial D_S=-\partial D_N$ denotes the boundary between $D_S$ and $D_N$.
The Berry curvature integrated over the Brillouin zone is then given by
\begin{eqnarray}\label{Berry flux}
\int_{T^2}{\Omega}
&=& \int_{D_S}d{A}^{(S)}
+ \int_{D_{N}}d{A}^{(N)}
=\int_{\partial D_S} {A}^{(S)}
+ \int_{\partial D_{N}} {A}^{(N)} \nonumber\\
&=&\int_{\partial D_S} \left({A}^{(S)}-{A}^{(N)}\right)
=-\int_{\partial D_S} d\mathbf{k}\,\partial_\mathbf{k}\phi(\mathbf{k}),
\end{eqnarray}
where we have used Stokes' theorem. We will use \eqnref{Berry flux} to prove that the total flux is quantized.

Pictorially, the noth- and south-pole singularities take place whenever the map $\mathbf{h}: T^2 \rightarrow \mathbb{R}^3\setminus\{0\}$ touches the positive or negative $z$ axis. If we suppose north-pole singularities do not coincide with south-pole singularities in positions of the same $k_y$, the Brillouin zone can be covered by $D_S$ and $D_N$ in such a way that each of them consists of ``horizontal strips'' as shown in \figref{fig:singularities}.\footnote{Generally, it is possible that a north- or south-pole singularity occupies a continues curve or even a continuous region, instead of an isolated point, in the Brillouin zone. The following argument for the quantization of the flux remains the same even if this occurs. Moreover, it is also possible that a noth pole and a south pole coincide in the same $k_y$. In this case, instead of a horizontal straight line as shown in \figref{fig:singularities}, we have to choose a deformed boundary line to detour around the coincident poles. For the special cases that $h_z(k_x,k_y)=h_z(k_y)$ is independent of $k_x$, which are the topic of \secref{sec:special cases}, the assumption that noth poles and south poles do not coincide in the same $k_y$ is completely valid.} As $\partial D_S$ consists of horizontal oriented lines as shown in \figref{fig:singularities}, it follows from \eqnref{Berry flux} that
\begin{eqnarray}
\int_{T^2}{\Omega} &=& -\int_{\partial D_S} d\mathbf{k}\,\partial_\mathbf{k}\phi(\mathbf{k})
=\sum_{k_y=k_y^{(1)},k_y^{(2)},\dots} \pm\, i\int_0^{2\pi} dk_x \frac{\partial}{\partial k_x}
\log \frac{\abs{h_x+ih_y}}{h_x+ih_y} \nonumber\\
&=&-\sum_{k_y=k_y^{(1)},k_y^{(2)},\dots} \pm\, i\int_0^{2\pi} dk_x \frac{\partial}{\partial k_x}
\log \left(h_x(\mathbf{k})+ih_y(\mathbf{k})\right),
\end{eqnarray}
where $k_y^{(1)},k_y^{(2)},\dots$ are the horizontal positions of the horizontal lines and $\pm$ corresponds to their orientations. For a fixed $k_y^{(i)}$, each summand is associated with the \emph{winding number} of the corresponding ``constant-$k_y$ loop'' winding around the $z$ axis. A constant-$k_y$ loop is defined as an $S^1$ loop given by $\{\mathbf{h}(k_x,k_y)|k_x\in[0,2\pi]\}$  and orientated in the increasing $k_x$ direction (see \figref{fig:tori}). More precisely, the winding number of the $k_y=k_y^{(i)}$ loop is given by the integral of the complex logarithm function as \cite{rudin1976principles}
\begin{equation}\label{winding number}
w(k_y)=\frac{1}{2\pi i}\int_0^{2\pi} dk_x \frac{\partial}{\partial k_x}
\log \left(h_x(k_x,k_y)+ih_y(k_x,k_y)\right).
\end{equation}
Consequently, we have
\begin{equation}\label{Chern number}
\int_{T^2}{\Omega} = 2\pi \sum_{k_y=k_y^{(1)},k_y^{(2)},\dots} \pm w(k_y)
\equiv 2\pi \mathtt{C},
\qquad \mathtt{C}\in\mathbb{Z}.
\end{equation}
That is, the total flux of the Berry curvature is quantized and characterized by the integer $\mathtt{C}$ known as the \emph{Chern number}. Furthermore, the Chern number is equal to the sum of the winding numbers around the $z$ axis of the constant-$k_y$ loops that separate the north-pole regions from the south-pole regions and inherit the orientations of the south-pole regions.
Alternatively, \eqnref{Chern number} can be rewritten as
\begin{equation}\label{Chern number 2}
\mathtt{C}\equiv\frac{1}{2\pi}\int_{T^2}{\Omega} = w(k_y^{(1)}) -  w(k_y^{(2)}) + w(k_y^{(3)}) - w(k_y^{(4)}) + \dots,
\qquad \mathtt{C}\in\mathbb{Z},
\end{equation}
where the orientation of the constant-$k_y$ loop is taken to be the increasing-$k_x$ direction.

\begin{figure}
\begin{tikzpicture}


\begin{scope}[shift={(0,0)},scale=1]

 \draw [help lines,->] (-0.5,0) -- (4.5,0);
 \node at (4.8,0) {$k_x$};
 \node [below] at (4,0) {$2\pi$};
 \draw [help lines,->] (0,-0.5) -- (0,4.5);
 \node at (-0.1,4.8) {$k_y$};
 \node [left] at (0,4) {$2\pi$};

 \path [fill=lightgray] (0,0) rectangle (4,1.25);
 \path [fill=lightgray] (0,1.8) rectangle (4,2.35);
 \path [fill=lightgray] (0,3.3) rectangle (4,4);

 \draw [blue] (0,0) rectangle (4,4);

 \draw [fill] (1.6,0.3) circle [radius=0.05];
 \draw [fill] (1.2,0.6) circle [radius=0.05];
 \draw [fill] (1,1) circle [radius=0.05];
 \draw [] (2,1.5) circle [radius=0.05];
 \draw [fill] (1.5,2.1) circle [radius=0.05];
 \draw [] (3,2.6) circle [radius=0.05];
 \draw [] (2.5,3) circle [radius=0.05];
 \draw [fill] (1.8,3.6) circle [radius=0.05];

\end{scope}


\begin{scope}[shift={(7,0)},scale=1]

 \draw [help lines,->] (-0.5,0) -- (4.5,0);
 \node at (4.8,0) {$k_x$};
 \node [below] at (4,0) {$2\pi$};
 \draw [help lines,->] (0,-0.5) -- (0,4.5);
 \node at (-0.1,4.8) {$k_y$};
 \node [left] at (0,4) {$2\pi$};

 \draw [magenta,thick,->] (0,1.25) -- (2,1.25);
 \draw [magenta,thick,-] (2,1.25) -- (4,1.25);
 \node [left] at (0,1.25) {$k_y^{(1)}$};

 \draw [magenta,thick,-] (0,1.8) -- (2,1.8);
 \draw [magenta,thick,<-] (2,1.8) -- (4,1.8);
 \node [left] at (0,1.8) {$k_y^{(2)}$};

 \draw [magenta,thick,->] (0,2.35) -- (2,2.35);
 \draw [magenta,thick,-] (2,2.35) -- (4,2.35);
 \node [left] at (0,2.35) {$k_y^{(3)}$};

 \draw [magenta,thick,-] (0,3.3) -- (2,3.3);
 \draw [magenta,thick,<-] (2,3.3) -- (4,3.3);
 \node [left] at (0,3.3) {$k_y^{(4)}$};

 \draw [blue] (0,0) rectangle (4,4);

 \draw [fill] (1.6,0.3) circle [radius=0.05];
 \draw [fill] (1.2,0.6) circle [radius=0.05];
 \draw [fill] (1,1) circle [radius=0.05];
 \draw [] (2,1.5) circle [radius=0.05];
 \draw [fill] (1.5,2.1) circle [radius=0.05];
 \draw [] (3,2.6) circle [radius=0.05];
 \draw [] (2.5,3) circle [radius=0.05];
 \draw [fill] (1.8,3.6) circle [radius=0.05];

\end{scope}

\end{tikzpicture}
\caption{An example configuration of noth-pole singularities (solid dots) and south-pole singularities (hollow dots) in the Brillouin zone $T^2\equiv[0,2\pi]\times[0,2\pi]$. \textit{Left}: $D_N$ is covered by the shaded strips, and $D_S$ is by the unshaded strips. \textit{Right}: $\partial D_S\equiv-\partial D_N$ is given by the horizontal oriented lines.}\label{fig:singularities}
\end{figure}
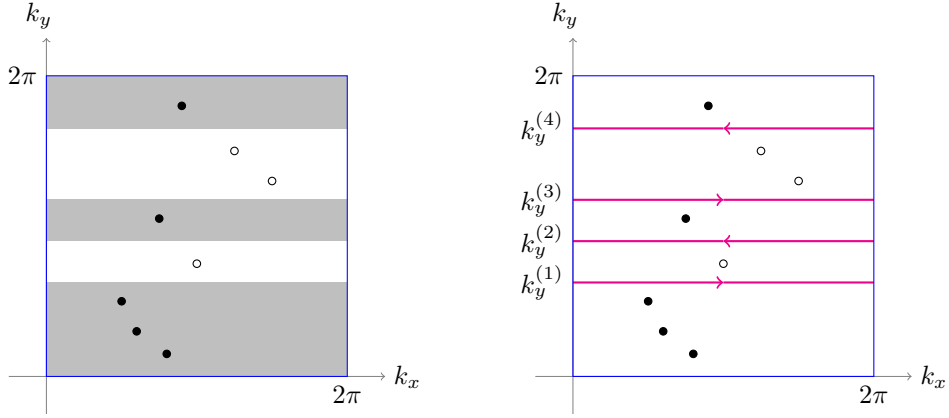

\subsection{Magnetic monopole}

The Berry connection \eqnref{Berry connection} as a one-form admits the equivalent expression in terms of $d\mathbf{h}$:
\begin{equation}\label{Berry connection in h}
  A(\mathbf{h}):= -i\bra{u_-} d\ket{u_-} \equiv -i\bra{u_-} \partial_{\mathbf{h}}\ket{u_-} d\mathbf{h}.
\end{equation}
Substituting \eqnref{uN} into \eqnref{Berry connection in h} then gives the explicit form of the Berry connection in terms of the components in the $\mathbf{h}$ space:
\begin{equation}
  A_{h_x}^{(N)}= \frac{-h_y}{2\abs{\mathbf{h}}(\abs{\mathbf{h}}+h_z)}, \qquad A_{h_y}^{(N)}=\frac{h_x}{2\abs{\mathbf{h}}(\abs{\mathbf{h}}+h_z)}, \qquad A_{h_z}^{(N)}=0.
\end{equation}
By \eqnref{Berry curvature}, the corresponding Berry curvature is then given by the $\mathbf{h}$-space components:
\begin{equation}\label{berry curvature}
 \Omega_{h_x h_y}=\frac{h_z}{2\abs{\mathbf{h}}^3},\quad \Omega_{h_y h_z}=\frac{h_x}{2\abs{\mathbf{h}}^3},\quad  \Omega_{h_z h_x}=\frac{h_y}{2\abs{\mathbf{h}}^3}.
\end{equation}
The magnetic filed strength $B^{h_\lambda}:=\epsilon^{h_\lambda h_\mu h_\nu}\Omega_{h_\mu h_\nu}$ is precisely the magnetic field produced by a unit magnetic monopole located at the origin in the $\mathbf{h}$ space. Therefore, the Chern number, as a measure of the total magnetic flux passing through the Brillouin zone, can be pictorially understood as how many times the torus $T^2$ encloses the origin in the $\mathbf{h}$ space for the map $\mathbf{h}:T^2\rightarrow \mathbb{R}^3\backslash\{0\}$.

To visualize how the map $\mathbf{h}:T^2\rightarrow \mathbb{R}^3\backslash\{0\}$ coils around the origin, see \figref{fig:tori} for examples of $\mathtt{C}=1$ and the figures in \secref{sec:examples} for examples of $\mathtt{C}=0,1,2$. As $k_y$ increases, if a constant-$k_y$ loop winding around the positive $z$ axis becomes winding around the negative $z$ axis or the other way around, the winding number will change by an integer. It is intuitive to see that the integer of this change counts up the number of times the torus encloses the origin. This gives a pictorial explanation for \eqnref{Chern number 2}.

\subsection{More about the winding number}\label{sec:more on winding number}
In order to prove the bulk-boundary correspondence in the next section, we need to elaborate on the winding number given in \eqnref{winding number}.

First, let us define
\begin{equation}\label{hxy}
h_{xy}(k_x,k_y) := h_x(k_x,k_y) + i h_y(k_x,k_y).
\end{equation}
By rewriting $z=e^{ik_x}$ and $dz=ie^{ik_x}dk_x$, the winding number given in \eqnref{winding number} can be recast as a contour integral along the unit circle on the complex plane:
\begin{equation}\label{nu 2}
  w(k_y) = \frac{1}{2\pi i}\oint_{\abs{z}=1} dz \frac{h'_{xy}(z;k_y)}{h_{xy}(z;k_y)}.
\end{equation}
Note that, according to \eqnref{h(k)}, $h_{xy}(z;k_y)$ is a Laurent polynomial of $z$ over $\mathbb{C}$ given by
\begin{equation}\label{hxy(z)}
h_{xy}(z;k_y) \equiv h_x(k_x,k_y) + i h_y(k_x,k_y) \Big|_{e^{ik_x}\rightarrow z}= \sum_{n=-\bar{n}_x}^{\bar{n}_x} w_n(k_y) z^n,
\end{equation}
where the coefficients are given by
\begin{equation}\label{wn}
w_{n_x}(y_k) := \sum_{n_y=-\bar{n}_y}^{\bar{n}_y}
e^{in_yk_y} \left(w^x_{n_x,n_y}+iw^y_{n_x,n_y}\right).
\end{equation}
Thus, $z^{\bar{n}_x}h_{xy}(z;k_y)$ is polynomial of $z$ and can be formally factorized as
\begin{equation}\label{factorization}
  z^{\bar{n}_x}h_{xy}(z;k_y)=\sum_{n=-\bar{n}_x}^{\bar{n}_x} w_n(k_y) z^{n+\bar{n}_x} = w_{\bar{n}_x}(k_y) \prod_j (z-\xi_j)^{\nu_j},
\end{equation}
where $\xi_j$ are the roots of $z^{\bar{n}_x}h_{xy}(z;k_y)$ and $\nu_j\in\mathbb{N}$ are the corresponding multiplicities. Substituting \eqnref{factorization} for $h_{xy}(z;k_y)$ into \eqnref{nu 2} leads to
\begin{equation}
  w(k_y) = \sum_{j} \frac{1}{2\pi i}\oint_{\abs{z}=1} dz \frac{\nu_j}{z-\xi_j}
  - \frac{1}{2\pi i}\oint_{\abs{z}=1} dz \frac{\bar{n}_x}{z}.
\end{equation}
Cauchy's integral formula then implies
\begin{equation}\label{key eq 1}
  w(k_y) = -\bar{n}_x + \sum_{j=1,\dots \atop \abs{\xi_j}<1} \nu_j,
  \quad \text{where}\ z^{\bar{n}_x}h_{xy}(z;k_y)\equiv\sum_{n=-\bar{n}_x}^{\bar{n}_x} w_n(k_y) z^{n+\bar{n}_x}
  \propto \prod_j (z-\xi_j)^{\nu_j}.
\end{equation}
That is, the winding number is the sum of the multiplicities of those roots of $\sum_{n=-\bar{n}_x}^{\bar{n}_x} w_n(k_y) z^{n+\bar{n}_x}$ that are located inside the unit circle on the complex plane.\footnote{Note that we assume $\abs{\xi_j}\neq1$ for all $\xi_j$ in \eqnref{factorization}. If $\abs{\xi_j}=1$, we would have $\xi_j=e^{i\theta}$ for some $\theta\in[0,2\pi]$ and therefore $h_{xy}(k_x=\theta;k_y)=0$, which is the case that the constant-$k_y$ loop hits the $z$ axis and the winding number cannot be defined.}

Similarly, repeating the above calculation with $z=e^{-ik_x}$, $dz=-ie^{-ik_x}dk_x$ and $h_{xy}(z^{-1};k_y)=\sum_{n=-\bar{n}_x}^{\bar{n}_x} w_n(k_y) z^{-n}\equiv\sum_{n=-\bar{n}_x}^{\bar{n}_x} w_{-n}(k_y) z^n$, we obtain a different expression:
\begin{equation}\label{key eq 2}
w(k_y) = \bar{n}_x - \sum_{j=1,\dots \atop \abs{\xi_j}<1} \nu_j,
\quad \text{where}\ z^{\bar{n}_x}h_{xy}(z^{-1};k_y)\equiv\sum_{n=-\bar{n}_x}^{\bar{n}_x} w_{-n}(k_y) z^{n+\bar{n}_x}
\propto \prod_j (z-\xi_j)^{\nu_j}.
\end{equation}

Equivalently, the winding number can also be expressed in terms of
\begin{equation}
h_{xy}^*(k_x,k_y) := h_x(k_x,k_y) -i h_y(k_x,k_y)
\end{equation}
as
\begin{equation}
  w(k_y) = -\frac{1}{2\pi i}\int_{-\pi}^{\pi} dk_x \frac{d}{dk_x}\log h_{xy}(k_x,k_y)^*,
\end{equation}
because the complex conjugation gives the opposite polar angle.
Consequently, we have
\begin{equation}\label{key eq 3}
w(k_y) = -\bar{n}_x + \sum_{j=1,\dots \atop \abs{\xi_j}<1} \nu_j,
\quad \text{where}\ z^{\bar{n}_x}h^*_{xy}(z^{-1};k_y)\equiv\sum_{n=-\bar{n}_x}^{\bar{n}_x} w^*_n(k_y) z^{n+\bar{n}_x}
\propto \prod_j (z-\xi_j)^{\nu_j},
\end{equation}
and
\begin{equation}\label{key eq 4}
w(k_y) = \bar{n}_x - \sum_{j=1,\dots \atop \abs{\xi_j}<1} \nu_j,
\quad \text{where}\ z^{\bar{n}_x}h^*_{xy}(z;k_y)\equiv\sum_{n=-\bar{n}_x}^{\bar{n}_x} w^*_{-n}(k_y) z^{n+\bar{n}_x}
\propto \prod_j (z-\xi_j)^{\nu_j},
\end{equation}
where $h_{xy}^*(z;k_y)$ is defined as
\begin{equation}\label{hxy(z)*}
h_{xy}^*(z;k_y)
\equiv h_x(k_x,k_y) - i h_y(k_x,k_y) \Big|_{e^{ik_x}\rightarrow z}
=\sum_{n=-\bar{n}_x}^{\bar{n}_x} w^*_n(k_y) z^{-n}
\end{equation}
in accord with \eqnref{hxy(z)}.
Equations \eqnref{key eq 1}, \eqnref{key eq 2}, \eqnref{key eq 3} and \eqnref{key eq 4} are the key identities that will be used to relate the winding number to the multiplicity of the edge states.

\section{Bulk-Boundary correspondence on a strip}\label{sec:correspondence on a strip}
To study the physics not only for the bulk but also for the boundaries, we have to remove the periodic (Born-von Karman) boundary condition. The standard treatment is to consider a strip of a lattice with a large but finite width as depicted in \figref{fig:strip}. That is, along the $y$ direction, we still impose the periodic  boundary condition or simply treat the strip as infinitely long (i.e., take the limit $N_y\rightarrow\infty$); along the $x$ direction, however, we keep $N_x$ large but finite and impose the open boundary condition on both edges (i.e., any out-of-edge hopping is set to vanish).
More precisely, the Hamiltonian of the strip, $\hat{H}_\mathrm{strip}$, is given exactly as the same as $\hat{H}_\mathrm{bulk}$ in \eqnref{H bulk} except that the summands involving $\ket{m_x-n_x,m_y-n_y}\bra{m_x,m_y}$ are dropped out as long as $m_x-n_x$ is out of range (i.e., $m_x-n_x<1$ or $m_x-n_x>N_x$).
In this section, we will make a precise statement of the bulk-boundary correspondence on a strip in terms of the edge modes (i.e., energy eigenstates localized on the right or left edge region) and provide a rigorous proof of it.

\begin{figure}
\begin{tikzpicture}

\begin{scope}[scale=0.4]

 \draw[-,olive] (1,0.05) grid (10,15.95);
 \draw [->,thick] (4,-1) -- (7,-1);
 \node [right] at (7,-1) {$x$};
 \node [below] at (0.7,0.3) {\tiny $m_x\!=\!1$};
 \node [below] at (10.6,0.3) {\tiny $m_x\!=\!N_x$};
 \draw [->, thick] (-0.5,6.5) -- (-0.5,9.5);
 \node [above] at (-0.5,9.5) {$y$};

 \foreach \i in {1,2,...,10}{
  \foreach \j in {1,2,...,15}{
   \draw [blue,fill=blue] (\i,\j) circle [radius=0.1];
  }
 }

\end{scope}


\begin{scope}[shift={(7,0)},scale=0.4]

 \draw[-,olive] (1,0.05) grid (10,15.95);
 \draw [->,thick] (4,-1) -- (7,-1);
 \node [right] at (7,-1) {$x$};
 \node [below] at (0.7,0.3) {\tiny $m_x\!=\!1$};
 \node [below] at (10.6,0.3) {\tiny $m_x\!=\!N_x$};
 \draw [->, thick] (-0.5,6.5) -- (-0.5,9.5);
 \node [above] at (-0.5,9.5) {$k_y$};

 \foreach \i in {1,2,...,10}{
  \foreach \j in {1,2,...,15}{
   \draw [blue,fill=blue] (\i,\j) circle [radius=0.1];
  }
 }

 \draw [red, rounded corners] (0.5,0.5) rectangle (10.5,1.5);

\end{scope}

\end{tikzpicture}
\caption{\textit{Left}: A strip of a two-dimensional lattice. \textit{Right}: The strip can be viewed as a one-dimensional lattice where each unit cell (shown as the enclosed region) consists of $N_x$ sublattice sites.}\label{fig:strip}
\end{figure}
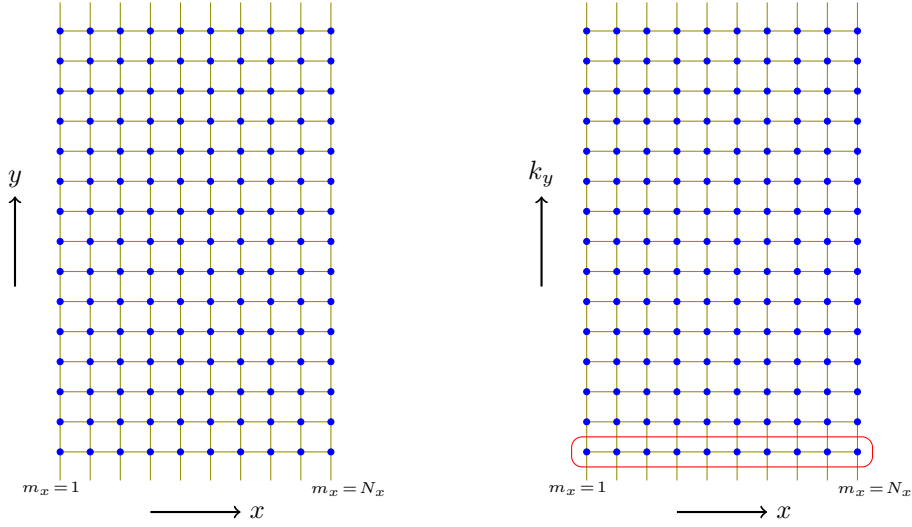

\subsection{Casting the eigenvalue problem}\label{sec:eigenvalue problem}
In the presence of the boundary edges, the lattice translational invariance is broken along the $x$ direction, but remains respected along the $y$ direction. Therefore, the lattice momentum $k_y$ is still a good quantum number. In other words, the system of a strip can be viewed as a one-dimensional lattice with each unit cell consisting of $N_x$ sublattice sites, as shwon in \figref{fig:strip}.
Under the Fourier transformation with respect to $k_y$, the Hamiltonian of the strip, $\hat{H}_\mathrm{strip}$, is reduced into a $k_y$-indexed Hamiltonian, denoted as $\hat{H}_{N_x}(k_y)$ and given by
\begin{eqnarray}
  &&\hat{H}_{N_x}(k_y) :=
  \bra{k_y}\hat{H}_{\text{strip}}\ket{k_y} \\
  && = \sum_{n_y=-\bar{n}_y}^{\bar{n}_y}  \sum_{n_x=0}^{\bar{n}_x} \sum_{m_x=n_x +1}^{N_x}
  e^{i n_y k_y} \ket{m_x-n_x}\bra{m_x}
  \otimes
  \left(
  \begin{array}{cc}
        w^0_{n_x,n_y}+w^z_{n_x,n_y} & w^x_{n_x,n_y}-iw^y_{n_x,n_y} \\
        w^x_{n_x,n_y}+iw^y_{n_x,n_y} & w^0_{n_x,n_y}-w^z_{n_x,n_y} \\
      \end{array}
  \right) \nonumber\\
  && \mbox{} + \sum_{n_y=-\bar{n}_y}^{\bar{n}_y} \sum_{n_x=1}^{\bar{n}_x} \sum_{m_x=1}^{N_x -n_x}
  e^{-i n_y k_y}  \ket{m_x+n_x}\bra{m_x}
  \otimes
  \left(
  \begin{array}{cc}
        w^0_{-n_x,-n_y}+w^z_{-n_x,-n_y} & w^x_{-n_x,-n_y}-iw^y_{-n_x,-n_y} \\
        w^x_{-n_x,-n_y}+iw^y_{-n_x,-n_y} & w^0_{-n_x,-n_y}-w^z_{n_x,n_y} \\
      \end{array}
  \right). \nonumber
\end{eqnarray}

The eigenvalue problem of $\hat{H}_{N_x}(k_y)\ket{\Psi(k_y)}=E(k_y)\ket{\Psi(k_y)}$ with
\begin{equation}\label{Psi ky}
  \ket{\Psi(k_y)}=\sum_{m_x =1}^{N_x} \ket{m_x}\otimes \left( \begin{array}{c}
                                                 a_{m_x} \\
                                                 b_{m_x}
                                               \end{array} \right),
\end{equation}
then gives $2N_x$ equations for $2N_x$ variables $a_m$ and $b_n$ as
\begin{subequations}\label{eqs}
\begin{eqnarray}
  \label{eq a left}
   E(k_y) a_{m_x} &=& \sum_{n_y=-\bar{n}_y}^{\bar{n}_y}e^{i n_y k_y}\sum_{n_x=-m_x+1}^{\bar{n}_x}  (w^0_{n_x,n_y}+w^z_{n_x,n_y}) a_{m_x+n_x} + (w^x_{n_x,n_y}-iw^y_{n_x,n_y}) b_{m_x+n_x},\nonumber\\
  && \qquad
  \text{for}\ m_x=1,\dots,\bar{n}_x, \\
  \label{eq a bulk}
  E(k_y) a_{m_x} &=& \sum_{n_y=-\bar{n}_y}^{\bar{n}_y}e^{i n_y k_y}\sum_{n_x=-\bar{n}_x}^{\bar{n}_x} (w^0_{n_x,n_y}+w^z_{n_x,n_y}) a_{m_x +n_x} + (w^x_{n_x,n_y}-iw^y_{n_x,n_y}) b_{m_x+n_x}, \nonumber\\
  &&\qquad
  \text{for}\ {m_x}=\bar{n}_x+1,\dots,N_x-\bar{n}_x, \\
  \label{eq a right}
  E(k_y) a_{m_x} &=& \sum_{n_y=-\bar{n}_y}^{\bar{n}_y}e^{i n_y k_y}\sum_{n_x=-\bar{n}_x}^{N_x-m_x}   (w^0_{n_x,n_y}+w^z_{n_x,n_y}) a_{m_x +n_x} + (w^x_{n_x,n_y}-iw^y_{n_x,n_y}) b_{m_x+n_x},\nonumber\\
  && \qquad
  \text{for}\ m_x=N_x-\bar{n}_x+1,\dots,N_x, \\
  \label{eq b left}
  E(k_y) b_{m_x}&=&  \sum_{n_y=-\bar{n}_y}^{\bar{n}_y}e^{i n_y k_y}\sum_{n_x=-\bar{n}_x}^{m_x-1} (w^x_{-n_x,n_y}+iw^y_{-n_x,n_y}) a_{m_x -n_x} + (w^0_{-n_x,n_y}-w^z_{-n_x,n_y}) b_{m_x-n_x},\nonumber\\
  && \qquad
  \text{for}\ m_x=1,\dots,\bar{n}_x, \\
  \label{eq b bulk}
  E(k_y) b_{m_x} &=& \sum_{n_y=-\bar{n}_y}^{\bar{n}_y}e^{i n_y k_y}\sum_{n_x=-\bar{n}_x}^{\bar{n}_x} (w^x_{-n_x,n_y}+iw^y_{-n_x,n_y}) a_{m_x-n_x} + (w^0_{-n_x,n_y}-w^z_{-n_x,n_y}) b_{m_x -n_x}, \nonumber\\
  &&\qquad
  \text{for}\ {m_x}=\bar{n}_x+1,\dots,N_x -\bar{n}_x, \\
  \label{eq b right}
  E(k_y) b_{m_x}&=& \sum_{n_y=-\bar{n}_y}^{\bar{n}_y}e^{i n_y k_y}\sum_{n_x=m_x-N_x}^{\bar{n}_x}    (w^x_{-n_x,n_y}+iw^y_{-n_x,n_y})a_{m_x -n_x} + (w^0_{-n_x,n_y}-w^z_{-n_x,n_y}) b_{m_x-n_x},\nonumber\\
  && \qquad
  \text{for}\ {m_x}=N_x-\bar{n}_x+1,\dots,N_x.
\end{eqnarray}
\end{subequations}
Here, each of \eqnref{eq a bulk} and \eqnref{eq b bulk} gives $N-2\bar{n}_x$ equations for the lattice points in the bulk region (i.e., away from the edges); each of \eqnref{eq a left} and \eqnref{eq b left} gives $\bar{n}_x$ equations for the points close to the left edge; each of \eqnref{eq a right} and \eqnref{eq b right} gives $\bar{n}_x$ equations for the points close to the right edge.
Particularly, \eqnref{eq a bulk} and \eqnref{eq b bulk} for the bulk region can be combined into
\begin{eqnarray}\label{eigenvalue eq}
E(k_y)
\left(
\begin{array}{c}
a_{m_x} \\
b_{m_x}
\end{array}
\right)
&=& \sum_{n_y=-\bar{n}_y}^{\bar{n}_y}\sum_{n_x=-\bar{n}_x}^{\bar{n}_x}e^{i n_y k_y}
\left(
  \begin{array}{cc}
     w^0_{n_x,n_y}+w^z_{n_x,n_y} & w^x_{n_x,n_y}-iw^y_{n_x,n_y} \\
     w^x_{n_x,n_y}+iw^y_{n_x,n_y} & w^0_{n_x,n_y}-w^z_{n_x,n_y}
  \end{array}
\right)
\left(
\begin{array}{c}
a_{m_x+n_x} \\
b_{m_x+n_x}
\end{array}
\right), \nonumber\\
&&\qquad \text{for}\ {m_x}=\bar{n}_x+1,\dots,N_x-\bar{n}_x.
\end{eqnarray}

To solve the difference equation \eqnref{eigenvalue eq}, the standard strategy is to make the ansatz
\begin{equation}\label{ansatz}
\left(
\begin{array}{c} a_{m_x} \\
                b_{m_x}
\end{array}
\right)
=\left(
\begin{array}{c} a \\
                 b
\end{array}
\right)
\xi^{m_x}
\end{equation}
for some complex numbers $a$, $b$, and $\xi$ to be solved. Substituting this ansatz into \eqnref{eigenvalue eq}, we have
\begin{eqnarray}\label{solve xi}
&&E(k_y)
\left(
\begin{array}{c}
a \\
b
\end{array}
\right)
= \sum_{n_y=-\bar{n}_y}^{\bar{n}_y}\sum_{n_x=-\bar{n}_x}^{\bar{n}_x}
\xi^{n_x} e^{i n_y k_y}
\left(
  \begin{array}{cc}
     w^0_{n_x,n_y}+w^z_{n_x,n_y} & w^x_{n_x,n_y}-iw^y_{n_x,n_y} \\
     w^x_{n_x,n_y}+iw^y_{n_x,n_y} & w^0_{n_x,n_y}-w^z_{n_x,n_y}
  \end{array}
\right)
\left(
\begin{array}{c}
a \\
b
\end{array}
\right) \nonumber\\
&\implies&
\left.
\left(
  \begin{array}{cc}
     h_0(k_x,k_y)+h_z(k_x,k_y) -E(k_y) & h_x(k_x,k_y)-ih_y(k_x,k_y) \\
     h_x(k_x,k_y)+ih_y(k_x,k_y) & h_0(k_x,k_y)-h_z(k_x,k_y) - E(k_y)
  \end{array}
\right)
\right|_{e^{ik_x}\rightarrow \xi}
\left(
\begin{array}{c}
a \\
b
\end{array}
\right) \nonumber\\
&\equiv&
\left(
  \begin{array}{cc}
     h_0(\xi;k_y)+h_z(\xi;k_y) -E & h^*_{xy}(\xi;k_y) \\
     h_{xy}(\xi;k_y) & h_0(\xi;k_y)-h_z(\xi;k_y) - E
  \end{array}
\right)
\left(
\begin{array}{c}
a \\
b
\end{array}
\right)
=0,
\end{eqnarray}
where $h_{xy}(z;k_y)$ and $h_{xy}^*(z,k_y)$ are defined in \eqnref{hxy(z)} and \eqnref{hxy(z)*}, and $h_a(z;k_y)$ is similarly defined as
\begin{equation}
h_a(z;k_y):=h_a(k_x,k_y)\Big|_{e^{ik_x}\rightarrow z}
\equiv
\sum_{n_y=-\bar{n}_y}^{\bar{n}_y}\sum_{n_x=-\bar{n}_x}^{\bar{n}_x}
e^{in_yk_y}w^a_{n,n_y}z^n.
\end{equation}
This admits nontrivial solutions if and only if the determinant of the $2\times2$ matrix vanishes, i.e.,
\begin{equation}\label{h square}
h_x^2(\xi;k_y) + h_y^2(\xi;k_y) + h_z^2(\xi;k_y) - h_0^2(\xi;k_y) - E(k_y)^2
=0.
\end{equation}
Obviously, $\xi=e^{ik_x}$ for any $k_x\in\mathbb{R}$ gives a solution to \eqnref{h square} corresponding to $E(k_y)=h_0(k_x,k_y)\pm\abs{\mathbf{h}(k_x,k_y)}$, because we have \eqnref{E bulk}. If a linear superposition of \eqnref{ansatz} with various values of $\xi=e^{ik_x}$ that correspond to a specific $E(k_y)$ satisfies the boundary conditions \eqnref{eq a left}, \eqnref{eq a right}, \eqnref{eq b left}, and \eqnref{eq b right}, it is then an eigenstate of $\hat{H}_{N_x}(k_y)$. Subject to these boundary conditions, the permitted values of $E(k_y)$ for a given $k_y$ are reduced to a discrete spectrum. The eigenstates of this kind are said to be inherited from the spectrum of $\hat{H}(\mathbf{k})$ and thus referred to as \emph{bulk states}. The spectrum lines, i.e.\ $E(k_y)$ against $k_y$, of the bulk states of $\hat{H}_{N_x}(k_y)$ are clustered into two groups corresponding to $h_0(\mathbf{k})+\abs{\mathbf{h}(\mathbf{k})}$ and $h_0(\mathbf{k})-\abs{\mathbf{h}(\mathbf{k})}$, respectively;\footnote{\label{foot:tomography}When $N_x$ is large enough, the effect of imposing the open boundary condition upon the bulk states of $\hat{H}_{N_x}(k_y)$ is negligible, and consequently the bulk states of $\hat{H}_{N_x}(k_y)$ are approximately given by \eqnref{ket k} with $k_x=m\delta_x$ for $m=1,\dots,N_x$ and $\delta_x=2\pi/N_x$ as if the periodic boundary condition was still imposed in the $x$ direction. Approximately, the bulk spectrum lines of $\hat{H}_{N_x}(k_y)$ against $k_y$ can be understood as the superimposition of the tomographic cross sections of the spectrum of $\hat{H}(\mathbf{k})$ scanned along the $k_x$ direction.} we will simply refer to these bulk spectrum lines as upper or lower \emph{bulk bands}, respectively.

Meanwhile, it is possible that \eqnref{h square} admits solutions with $\abs{\xi}\neq1$. These solutions, if any, are referred to as \emph{edge states}, which are defined as energy eigenstates of $\hat{H}_{N_x}(k_y)$ that are localized on the left or right edge region and exponentially decayed towards the opposite edge. The localization on the left or right edge entails the condition $\abs{\xi}<1$ or $\abs{\xi^{-1}}<1$, respectively.
We will refer to the spectrum lines against $k_y$ of the edge states as left or right \emph{edge bands}, respectively.
As long as $N_x$ is large enough (so that the finite size effect is negligible), the wavefunctions and energies of edge states are independent of $N_x$, as opposed to bulk states (see \footref{foot:tomography}). The proof is given as follows.

For given $k_y$ and $N_x$, suppose we have a left edge mode solution with energy $E(k_y)$. This solution $\ket{\Psi(k_y)}$ must be a linear superposition of \eqnref{ansatz}, i.e., $\inner{m_x}{\Psi(k_y)}=(a_{m_x},b_{m_x})^\mathrm{T}=\sum_{i}(a_i,b_i)^\mathrm{T}\xi_i^{m_x}$ with particular values of $\xi=\xi_i$ and coefficients $(a_i,b_i)$. All $\xi_i$ satisfy \eqnref{h square} and $\abs{\xi_i}<1$. Furthermore, the linear superposition with the particular values of $\xi_i$ and $(a_i,b_i)$ satisfies the left boundary conditions \eqnref{eq a left} and \eqnref{eq b left}. The right boundary conditions \eqnref{eq a right} and \eqnref{eq b right} on the other hand become superfluous, since $\abs{\xi_i}<0$ implies $\xi_i^{m_x}\rightarrow0$ as $m_x$ approaches the right edge in the large $N_x$ limit. Because the right boundary conditions become irrelevant and all the other relevant conditions --- \eqnref{h square} for the bulk and \eqnref{eq a left} and \eqnref{eq b left} for the left boundary --- are independent of $N_x$, the linear superposition $(a_{m_x},b_{m_x})^\mathrm{T}=\sum_{i}(a_i,b_i)^\mathrm{T}\xi_i^{m_x}$ remains an eigenstate with the same energy $E(k_y)$, even if we change the value of $N_x$. The same argument can be repeated for right edge modes. Therefore, we have proven that the edge states are independent of $N_x$.

We will study other features of edge states in depth shortly.

\subsection{$h_0=0$ symmetry}\label{sec:h0=0 symmetry'}
In case of $h_0(\mathbf{k})=0$ (i.e., $w^0_{n_x,n_y}=0$ for all $n_x,n_y$), if $(a_{m_x},b_{m_x})$ is a solution to \eqnref{eqs}, it is straightforward to show that \eqnref{eq a left}, \eqnref{eq a bulk}, and \eqnref{eq a right} are interchanged with \eqnref{eq b right}, \eqnref{eq b bulk}, and \eqnref{eq b left}, respectively, under the replacement $(a_{m_x},b_{m_x}) \rightarrow (b_{N_x+1-m_x}^*,-a_{N_x+1-m_x}^*)$ and $E(k_y)\rightarrow-E(k_y)$.
That is, if $\ket{\Psi(k_y)}$ given in \eqnref{Psi ky} is an eigenstate of $\hat{H}_{N_x}(k_y)$ with the eigenvalue $E(k_y)$, then $\ket{\tilde{\Psi}(k_y)}$ give by
\begin{equation}\label{tilde Psi ky}
  \ket{\tilde{\Psi}(k_y)}=\sum_{m_x =1}^{N_x} \ket{m_x}\otimes \left( \begin{array}{c}
                                                 \tilde{a}_{m_x} \\
                                                 \tilde{b}_{m_x}
                                               \end{array} \right)
                        =\sum_{m_x =1}^{N_x} \ket{m_x}\otimes \left( \begin{array}{c}
                                                 b_{N_x+1-m_x}^* \\
                                                 -a_{N_x+1-m_x}^*
                                               \end{array} \right),
\end{equation}
is an eigenstate of $\hat{H}_{N_x}(k_y)$ with the opposite eigenvalue $-E(k_y)$.

Compared with \eqnref{tilde a b}, the relation between $\ket{\Psi(k_y)}$ and $\ket{\tilde{\Psi}(k_y)}$ can be understood as inherited from the $h_0=0$ symmetry of $\hat{H}_\mathrm{bulk}$. However, because the lattice translational invariance in $x$ direction is broken in the presence of the left and right boundaries, $\bar{m}_x$ in \eqnref{tilde a b} is now fixed to $N_x+1$.
Consequently, the $h_0=0$ symmetry relates a \emph{left} edge mode with a \emph{right} edge mode.

\subsection{Special cases}\label{sec:special cases}
The analysis in \secref{sec:eigenvalue problem} shows that the bulk states of $\hat{H}_{N_x}(k_y)$ is inherited from the spectrum of $\hat{H}_\mathrm{bulk}$ but says little about the edge modes. To learn more about the edge modes, it is advantageous to first conduct a close investigation on the special case that $h_0(\mathbf{k})=0$ and $h_z(\mathbf{k})$ is independent of $k_x$. That is, before studying generic cases, we focus on the case that $h_0(\mathbf{k})$ vanishes identically and $h_z(\mathbf{k})$ given in \eqnref{h(k)} takes the special form:
\begin{equation}\label{hz special}
h_z(k_x,k_y) \equiv h_z(k_y) = \sum_{n_y=-\infty}^\infty w^z_{0,n_y} e^{in_yk_y},
\quad \text{with}\ w^z_{n_x\neq0,n_y}=0.
\end{equation}
See \figref{fig:tori} for examples of special and generic cases and their difference.
Substituting $w^0_{n_x,n_y}=0$ (i.e., $h_0(\mathbf{k})=0$) and \eqnref{hz special} into \eqnref{eqs}, we obtain\footnote{Also see \appref{appendix} for the corresponding matrix form.}
\begin{subequations}\label{eqs'}
\begin{eqnarray}
  \label{eq' a left}
  E(k_y) a_{m_x} &=& h_z(k_y)a_{m_x} + \sum_{n_y=-\bar{n}_y}^{\bar{n}_y}e^{i n_y k_y}\sum_{n_x=-m_x+1}^{\bar{n}_x} (w^x_{n_x,n_y}-iw^y_{n_x,n_y}) b_{m_x+n_x},\nonumber\\
  && \qquad
  \text{for}\ m_x=1,\dots,\bar{n}_x, \\
  \label{eq' a bulk}
  E(k_y) a_{m_x} &=& h_z(k_y)a_{m_x} + \sum_{n_y=-\bar{n}_y}^{\bar{n}_y}e^{i n_y k_y}\sum_{n_x=-\bar{n}_x}^{\bar{n}_x} (w^x_{n_x,n_y}-iw^y_{n_x,n_y}) b_{m_x+n_x}, \nonumber\\
  &&\qquad
  \text{for}\ {m_x}=\bar{n}_x+1,\dots,N_x-\bar{n}_x, \\
  \label{eq' a right}
  E(k_y) a_{m_x} &=& h_z(k_y)a_{m_x} + \sum_{n_y=-\bar{n}_y}^{\bar{n}_y}e^{i n_y k_y}\sum_{n_x=-\bar{n}_x}^{N_x-m_x} (w^x_{n_x,n_y}-iw^y_{n_x,n_y}) b_{m_x+n_x},\nonumber\\
  && \qquad
  \text{for}\ m_x=N_x-\bar{n}_x+1,\dots,N_x, \\
  \label{eq' b left}
  E(k_y) b_{m_x}&=& -h_z(k_y)b_{m_x} + \sum_{n_y=-\bar{n}_y}^{\bar{n}_y}e^{i n_y k_y}\sum_{n_x=-\bar{n}_x}^{m_x-1} (w^x_{-n_x,n_y}+iw^y_{-n_x,n_y}) a_{m_x -n_x},\nonumber\\
  && \qquad
  \text{for}\ m_x=1,\dots,\bar{n}_x, \\
  \label{eq' b bulk}
  E(k_y) b_{m_x} &=& -h_z(k_y)b_{m_x} + \sum_{n_y=-\bar{n}_y}^{\bar{n}_y}e^{i n_y k_y}\sum_{n_x=-\bar{n}_x}^{\bar{n}_x} (w^x_{-n_x,n_y}+iw^y_{-n_x,n_y}) a_{m_x-n_x}\nonumber\\
  &&\qquad
  \text{for}\ {m_x}=\bar{n}_x+1,\dots,N_x -\bar{n}_x, \\
  \label{eq' b right}
  E(k_y) b_{m_x}&=& -h_z(k_y)b_{m_x} + \sum_{n_y=-\bar{n}_y}^{\bar{n}_y}e^{i n_y k_y}\sum_{n_x=m_x-N_x}^{\bar{n}_x}    (w^x_{-n_x,n_y}+iw^y_{-n_x,n_y})a_{m_x -n_x},\nonumber\\
  && \qquad
  \text{for}\ {m_x}=N_x-\bar{n}_x+1,\dots,N_x.
\end{eqnarray}
\end{subequations}

The coupled equations \eqnref{eqs'} may admit edge modes that are purely ``$a$-type'' (i.e., with $b_{m_x}=0$) or ``$b$-type'' (i.e., with $a_{m_x}=0$). Imposing $b_{m_x}=0$ on \eqnref{eq' a left}, \eqnref{eq' a bulk}, and \eqnref{eq' a right}, and substituting the ansatz \eqnref{ansatz} into \eqnref{eq' b bulk}, we obtain the condition for $a$-type solutions:
\begin{subequations}\label{a part eq}
\begin{eqnarray}
\label{a part eq, part 1}
a\text{-type:}\qquad E(k_y) &=& h_z(k_y),\\
\label{a part eq, part 2}
0 &=& h_{xy}(\xi;k_y),
\end{eqnarray}
\end{subequations}
where $h_{xy}(z;k_y)$ is defined in \eqnref{hxy(z)}.
The solutions of $\xi$ in \eqnref{a part eq, part 2} are exactly those roots $\xi_j$ of the polynomial $z^{\bar{n}_x}h_{xy}(z;k_y)$. If $\xi_j$ has multiplicity $\nu_j$, any linear superpositions of
\begin{equation}\label{am solutions}
a_{m_x} = m_x^\ell\xi_j^{m_x}, \qquad \ell=0,1,\dots,\nu-1,
\end{equation}
are also solutions to \eqnref{a part eq, part 2}.\footnote{If $F(z):=z^{\bar{n}_x}h_{xy}(z;k_y)\equiv \sum_{n=-\bar{n}_x}^{\bar{n}_x} w_n(k_y) z^{n+\bar{n}_x}$ can be factorized as $F(z)=(z-\xi_j)^{\nu_j}f(z)$, where $f(z)$ is a polynomial of $z$ and $f(\xi_j)\neq0$, we have $\left.\frac{\partial^\ell F(z)}{\partial z^\ell}\right|_{z=\xi_j}= \left.\sum_{n=-\bar{n}_x}^{\bar{n}_x} w_n(k_y) \frac{d^\ell}{dz^\ell}z^{n+\bar{n}_x}\right|_{z=\xi_j}=0$ for $\ell=1,\dots,\nu_j-1$, implying that $a_{m_x}=\frac{d^\ell}{d\xi_j^\ell}\xi_j^{m_x+\bar{n}_x}$ for $\ell=1,\dots,v_i-1$ are all solutions to $\sum_{m_x=-\bar{n}_x}^{\bar{n}_x} w_{m_x}(k_y) a_{m_x}=0$. These solutions of $a_{m_x}$ can be rearranged into \eqnref{am solutions}.}
In case $\xi_j=0$, all the solutions in \eqnref{am solutions} collapse to $a_{m_x}=0$, which is problematic and has to be subjected to closer scrutiny. The fact that $\xi_j=0$ is a root of $z^{\bar{n}_x}h_{xy}(z;k_y)$ with multiplicity $\nu_j$ means that $F(z):=z^{\bar{n}_x}h_{xy}(z;k_y) \equiv\sum_{n=-\bar{n}_x}^{\bar{n}_x} w_n(k_y) z^{n+\bar{n}_x} =z^{\nu_j}f(z)$, where $f(z)$ is a polynomial of $z$ and $f(0)\neq0$. Consequently, in $F(z)$, the coefficients of $z^0,z^1,\dots,z^{\nu_j-1}$ all vanish, namely, $w_n(k_y)=0$ for $n=-\bar{n}_x,-\bar{n}_x+1,\dots,-\bar{n}_x+\nu_j-1$. This implies that \eqnref{eq' b bulk} in fact does not involve $a_1,a_2,\dots,a_{\nu_j}$ at all, as the index $m_x$ thereof runs for $m_x=\bar{n}_x+1,\dots,N_x-\bar{n}_x$. Therefore, in case of $\xi_j=0$, \eqnref{eq' b bulk} leaves $a_1,a_2,\dots,a_{\nu_j}$ untouched, and thus we still have $\nu_j$ linearly independent solutions to \eqnref{a part eq, part 2}, or more precisely to \eqnref{eq' b bulk}, given as
\begin{equation}\label{am solutions'}
\quad a_{m_x} = \delta_{m_x,\ell}, \qquad \ell=1,\dots,\nu_j.
\end{equation}
To sum up, whether $\xi_j$ vanishes or not, the root $\xi_j$ with multiplicity $\nu_j$ corresponds to $\nu_j$ linearly independent solutions to \eqnref{eq' b bulk}.

To be left edge modes, we must have $\abs{\xi_j}<1$ and consequently we have $\sum_{j=1,\dots;\,\abs{\xi_j}<1} \nu_j$ independent candidate left edge solutions. Furthermore, the solutions must satisfy the additional equations that have not been considered yet, namely \eqnref{eq' b left} and \eqnref{eq' b right} (with $b_{m_x}=0$). Whereas \eqnref{eq' b right} becomes superfluous in the $N_x\rightarrow\infty$ limit as $\xi^{N_x}\rightarrow0$, \eqnref{eq' b left} gives additional $\bar{n}_x$ constraints.\footnote{\label{foot:degenerate constraints}Accidentally, \eqnref{eq' b left} might become degenerate (i.e., with rank smaller than $\bar{n}_x$) for a special value of $k_y$, thus imposing fewer than $\bar{n}_x$ constraints. This occurs when a different edge band touches or intercepts the trajectory of $h_z(k_y)$ (or, equivalent, $-h_z(k_y)$) at the special point of $k_y$, thus increasing the number of purely $a$-type ($b$-type) solutions at the interception point. This \emph{accidental} situation can be avoided by an adiabatic deformation.} Therefore, we have $-\bar{n}_x+\sum_{j=1,\dots;\,\abs{\xi_j}<1} \nu_j$ left edge modes.\footnote{\label{foot:extra constraints}Note that the role of \eqnref{eq' b left} is to impose extra $\bar{n}_x$ constraints and its exact form is not essential here. \secref{sec:uniform edge perturbation} will elaborate on this point.}
By comparison with \eqnref{key eq 1}, it turns out that, for a given $k_y$, the number of left edge modes that are purely $a$-type is given by the winding number $w(k_y)$, provided $w(k_y)\geq0$.
Meanwhile, it follows from \eqnref{a part eq, part 1} that, in the region of $k_y$ with $w(k_y)>0$, the spectrum line of $a$-type left edge modes follows the function of $h_z(k_y)$ (more precisely, in the $N_x\rightarrow\infty$ limit).
As a consequence, the $a$-type edge band branches out from the upper bulk band cluster at the positions of $k_y$ where the constant-$k_y$ loop touches noth poles or from the lower bulk band cluster at the positions where the constant-$k_y$ loop touches south poles, because the upper and lower bulk bands are given by $E=\pm\abs{\mathbf{h}(k_x,k_y)}$, respectively, and north and south poles correspond to $\abs{\mathbf{h}}=\pm h_z$, respectively. See \figref{fig:band scheme} for a schematic illustration. In case there are no north or south poles (i.e., $\abs{\mathbf{h}}\neq\pm h_z$ for all $\mathbf{k}\in T^2$), the $a$-type edge band is a ``standalone'' edge band, which does not touch the bulk bands anywhere.

On the other hand, imposing $a_{m_x}=0$ gives rise to the condition for $b$-type solutions:
\begin{subequations}\label{b part eq}
\begin{eqnarray}
\label{b part eq, part 1}
b\text{-type:}\qquad E(k_y) &=& -h_z(k_y),\\
\label{b part eq, part 2}
0 &=& h_{xy}^*(\xi;k_y).
\end{eqnarray}
\end{subequations}
Following the same argument as above leads to the conclusion that, according to \eqnref{key eq 4}, the number of left edge modes that are purely $b$-type is given by the $\abs{w(k_y)}$, provided $w(k_y)\leq0$.
Furthermore, the edge band of $b$-type modes follows the trajectory of $-h_z(k_y)$ and branches out from the upper bulk band cluster at the positions of south poles or from the lower bulk band cluster at the positions of north poles. See \figref{fig:band scheme} again.
Also note that purely $a$-type and $b$-type modes cannot appear at the same time, because they correspond to opposite signs of $w(k_y)$.

Similarly, recasting the ansatz \eqnref{ansatz} as
\begin{equation}
\left(
\begin{array}{c} a_{m_x} \\
                b_{m_x}
\end{array}
\right)
=\left(
\begin{array}{c} a \\
                 b
\end{array}
\right)
\xi^{N_x-m_x},
\end{equation}
we can obtain the relation between $w(k_y)$ and the number of right edge modes that are purely $a$-type or $b$-type, by virtue of \eqnref{key eq 2} and \eqnref{key eq 3}. This relation is completely ``symmetric'' to that for left edge modes under the $h_0=0$ symmetry as discussed in \secref{sec:h0=0 symmetry'}. In the following, we will focus solely on the left edge modes, as the right edge modes can be readily obtained via the symmetry.

It must be noted that we have not yet considered all possible edge modes, as it is well possible that an edge mode may have both nonzero $a_{m_x}$ and $b_{m_x}$.\footnote{As we will see in \secref{sec:edge perturbation I} and \secref{sec:edge perturbation II}, these edge bands are most likely induced by ``edge perturbation'', which is the topic of \secref{sec:uniform edge perturbation}.} However, for our main purpose (which will become clear later), we are interested primarily in the edge bands that connect the upper and lower bulk band clusters across the bulk gap. These edge bands, it turns out, must be purely $a$-type or $b$-type. The proof by contradiction is given as follows.

Suppose there exists a ``nonconforming'' edge band that connects the upper and lower bulk band clusters but does not follow the trajectory of $\pm h_z(k_y)$. The nonconforming band nevertheless unavoidably intercepts the trajectory of $\pm h_z(k_y)$ somewhere. Imposing $E(k_y)=h_z(k_y)$ or $E(k_y)=-h_z(k_y)$ upon \eqnref{eqs'} entails $b_{m_x}=0$ or $a_{m_x}=0$, respectively, implying that any edge mode must become purely $a$-type or $b$-type at the interception point with $E(k_y)=\pm h_z(k_y)$.
However, as we have just proved, the multiplicity of purely $a$-type or $b$-type edge modes at $k_y$ is given by $\abs{w(k_y)}$, except that $k_y$ happens to be an \emph{accidental} point (as mentioned in \footref{foot:degenerate constraints}).
Since the interception of a nonconforming edge band with $\pm h_z(k_y)$ is \emph{nonaccidental} as it is unavoidable, the assumption of existing a nonconforming edge band therefore leads to the contradiction against what we have just proved.
As a conclusion, the edge bands connecting the upper and lower bulk band clusters must exactly follow the trajectory of $\pm h_z(k_y)$ and be purely $a$-type or $b$-type.

Apart from the ``upper-to-lower'' edge bands (i.e., those connecting the upper bulk band cluster to the lower bulk band cluster in the increasing-$k_y$ direction; e.g., see (a) and (b) of \figref{fig:deformations}) and the ``lower-to-upper'' edge bands, there are other possibilities of edge bands. First, there might be ``upper-to-upper'' edge bands, which appear from the upper bulk band cluster and merge back into the upper bulk band cluster without touching the lower bulk band cluster. An upper-to-upper edge band may occur within the bulk gap (e.g., see (d) and (e) of \figref{fig:deformations}), above the top contour of the upper bulk band cluster, or inside the upper bulk cluster. Likewise, there might be ``lower-to-lower'' edge bands with both ends anchored to the lower bulk band cluster. As opposed to upper-to-lower and lower-to-upper edge bands, an upper-to-upper or lower-to-lower edge band may follow the trajectories of $\pm h(k_y)$, intercept $\pm h_z(k_y)$ somewhere, or remain untouched with $\pm h_z(k_y)$, since the interception with $\pm h_z(k_y)$ is no longer unavoidable. If it follows the trajectories of $\pm h(k_y)$, it is purely $a$-type or $b$-type. If it intercepts $\pm h_z(k_y)$ somewhere, it becomes locally $a$-type or $b$-type at the interception point.\footnote{\label{foot:interception point}This gives the accidental situation mentioned in \footref{foot:degenerate constraints}.}
Finally, there might be ``standalone'' edge bands, which stretch over the whole domain of $k_y$ and touch neither the upper bulk band cluster nor the lower one (see (o) of \figref{fig:deformations2}). The standalone edge band is possible only if the trajectory of $h_z(k_y)$ does not connect the upper and lower bulk band clusters, because otherwise the standalone edge band would unavoidably intercept $\pm h_x(k_y)$ somewhere, thus leading to the same contradiction as mentioned above for the case of upper-to-lower and lower-to-upper edge bands. Particularly, the standalone edge band and the upper-to-lower or lower-to-upper edge band cannot coexist in the special case. In case $\mathbf{h}(\mathbf{k})$ has no north or south poles,  all constant-$k_y$ loops are of the same winding number and the trajectory of $\pm h_z(k_y)$ will give rise to a standalone edge band with the multiplicity given by the winding number.\footnote{This is demonstrated in \secref{sec:special case IV}.}

\begin{figure}

\centering
    \includegraphics[width=0.85\textwidth]{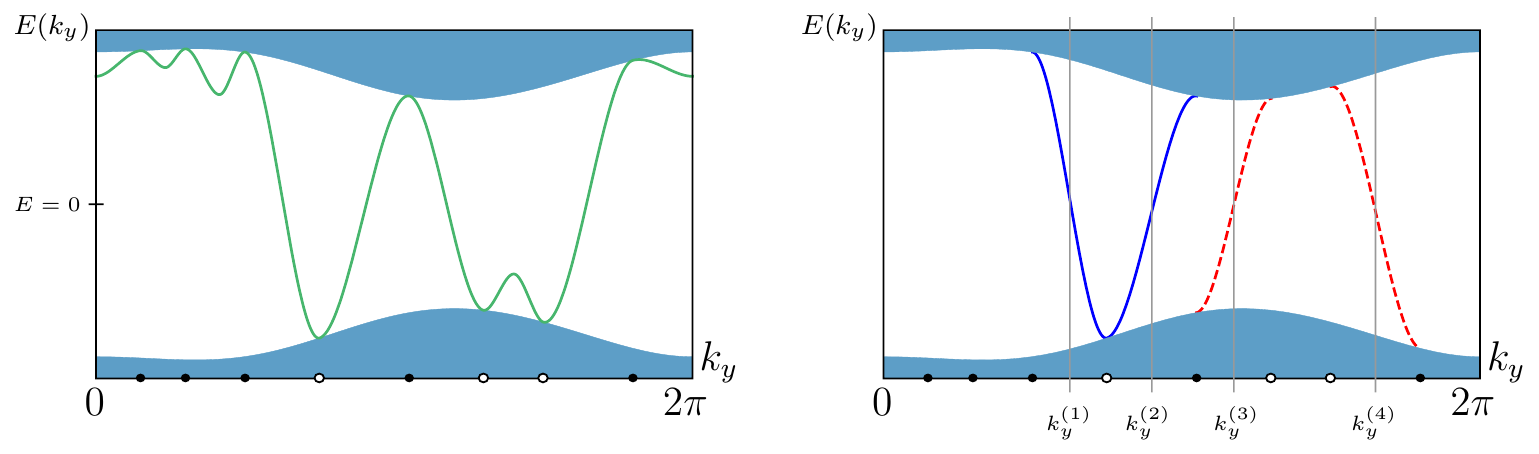}

\caption{\textit{Left}: A schematic plot of the bulk band clusters (upper and lower shaded areas) and the function of $h_z(k_y)$ (solid curve). The positions of north and south poles (assumed to be the same as given in \figref{fig:singularities}) are indicated by solid and hollow dots, at witch the trajectory of $h_z(k_y)$ branches out from the upper and lower bulk band clusters, respectively. \textit{Right}: An example of the left edge bands connecting the upper and lower bulk bands. Those of $a$-type are indicated by solid curves, which follow the trajectory of $h_z(k_y)$; those of $b$-type are indicated by dashed curves, which follow the trajectory of $-h(k_y)$. Here, we assume $w(k_y^{(1)}),w(k_y^{(2)})>0$ and $w(k_y^{(3)}),w(k_y^{(4)})<0$.}\label{fig:band scheme}
\end{figure}

\subsection{Counting the edge bands}\label{sec:counting edge bands}
According to the analysis in the previous subsection, the left edge bands connecting the upper and lower bulk band clusters must either be $a$-type and follow the trajectory of $+h_z(k_y)$ if $w(k_y)>0$ or be $b$-type and follow the trajectory of $-h_z(k_y)$ if $w(k_y)<0$. Furthermore, the multiplicity of the edge band is given by $\abs{w(k_y)}$.
Take \figref{fig:singularities} as an example of the north and south poles configuration for $\mathbf{h}(\mathbf{k})$. The trajectory of $h_z(k_y)$ against the bulk band clusters is schematically depicted on the left panel of \figref{fig:band scheme}.
The left edge bands connecting the upper and lower bulk band clusters take place in the intervals of $k_y$ that are delimited by a pair of north and south poles as depicted on the right panel of \figref{fig:band scheme}.

Let $N_1$ be the number (counted with multiplicities of degeneracy) of upper-to-lower left edge bands and $N_2$ be that of lower-to-upper left edge bands. Taking the right panel of \figref{fig:band scheme} as an example, where the depiction assumes $w(k_y^{(1)}),w(k_y^{(2)})>0$ and $w(k_y^{(3)}),w(k_y^{(4)})<0$, we have
\begin{eqnarray}\label{N1-N2}
N_1-N_2 &=& \abs{w(k_y^{(1)})} - \abs{w(k_y^{(2)})} - \abs{w(k_y^{(3)})} + \abs{w(k_y^{(4)})} \nonumber\\
&=& w(k_y^{(1)}) - w(k_y^{(2)}) + w(k_y^{(3)}) - w(k_y^{(4)}).
\end{eqnarray}
A moment of reflection tells that $N_1-N_2$ remains the same as given in the final line of \eqnref{N1-N2} even if we assume different signs of $w(k_y^{(i)})$, because flipping the sign of $w(k_y^{(i)})$ also exchanges an $a$-type edge band for an $b$-type edge band and \textit{vice versa}, thus flipping an upper-to-lower edge band into a lower-to-upper edge band or the other way around.
The final line of \eqnref{N1-N2} is exactly the Chern number as given in \eqnref{Chern number 2}. Therefore, for the special case, we have proved the bulk-boundary correspondence, which is explicitly stated as follows:
\begin{theorem}[Bulk-boundary correspondence]\label{theorem 1} Let $N_1$ be the number (counted with multiplicities of degeneracy) of the left edge bands connecting the upper bulk band cluster to the lower bulk cluster (in the increasing-$k_y$ direction) and $N_2$ be that of the left edge bands connecting the lower bulk band bluster to the upper bulk band cluster. We have
\begin{equation}
N_1-N_2 = \mathtt{C},
\end{equation}
where $\mathtt{C}$ is the Chern number of the bulk momentum-space Hamiltonian $\hat{H}(\mathbf{k})$.
Alternatively, $N_1$ is also the number of the right edge bands connecting the lower bulk band cluster to the upper bulk band cluster, and $N_2$ is also the number of the right edge bands connecting the upper bulk band cluster to the lower bulk band cluster.\footnote{This alternative description can be easily inferred under the $h_0=0$ symmetry, but it remains true even if the $h_0=0$ symmetry is broken (see \secref{sec:generic cases}).}
\end{theorem}


\begin{figure}

\centering
    \includegraphics[width=0.8\textwidth]{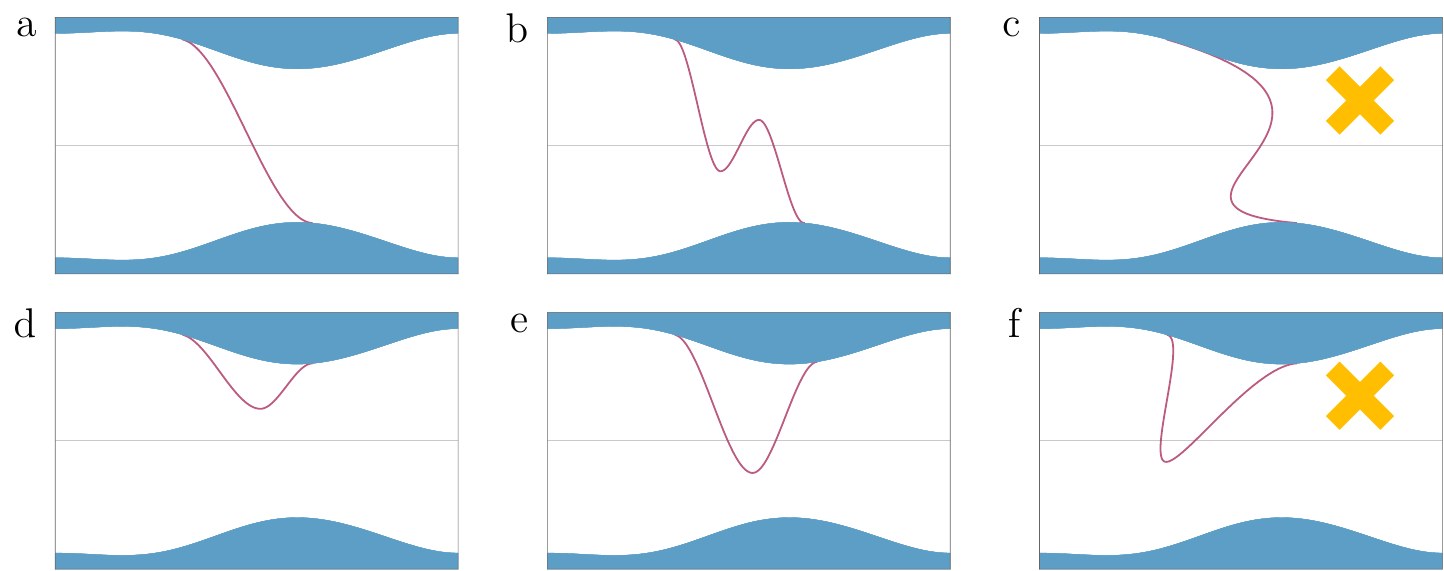}

\caption{More possible configurations of edge bands and impossible ones (marked with no-go signs.)}\label{fig:deformations}
\end{figure}

There are more possibilities and detailed features of edge bands not shown in \figref{fig:band scheme}. Let us discuss them with the illustration of \figref{fig:deformations}. An edge band connecting the upper and lower bulk bands as depicted in (a) of \figref{fig:deformations} might be continuously deformed into (b), thus giving rise to more interception points with $E=0$ in pairs with opposite signs of $dE(k_y)/dk_y$. However, it cannot be deformed into (c), where the edge band is not a single-valued function of $k_y$. This is because for any given $k_y$, the total number of energy eigenstates (both bulk and edge modes included) must equal $2N_x$, which would be violated if an edge band were not a single-valued function of $k_y$.
Furthermore, it is possible to have an upper-to-upper (or lower-to-lower) edge band as depicted in (d), which may or may not follow the trajectory of $\pm h_z(k_y)$. Similarly, (d) might be deformed into (e) and gives rise to interception points with $E=0$ in pairs with opposite signs of $dE(k_y)/dk_y$, but it cannot be deformed into (f) for the same reason as that for (c).
Finally, if the trajectory of $h_z(k_y)$ does not connect the upper and lower bulk bands, it is possible to have a standalone edge band, which may or may not intercept $E=0$ but anyway must be a single-valued function of $k_y$, for which the interception points with $E=0$ are in pairs with opposite signs of $dE(k_y)/dk_y$. Taking all the possible and forbidden cases into account, we can paraphrase \thmref{theorem 1} into the equivalent statement in terms of interception points with $E=0$ of \emph{all} edge bands (upper-to-lower, lower-to-upper, upper-to-upper, lower-to-lower and standalone all included) as follows:
\begin{theorem}[Bulk-boundary correspondence, equivalent description]\label{theorem 2}
Let $N'_1$ be the number (counted with multiplicities of degeneracy) of the interception points with $E=0$ of all left edge bands at which the derivative $dE(k_y)/dk_y$ of the edge band is negative, and $N'_2$ be the number of those at which the derivative $dE(k_y)/dk_y$ is positive. We have
\begin{equation}
N'_1-N'_2=\mathtt{C}.
\end{equation}
Alternatively, let $N''_1$ be the number of the interception points with $E=0$ of all right edge bands at which the derivative $dE(k_y)/dk_y$ of the edge band is positive, and $N''_2$ be the number of those at which the derivative $dE(k_y)/dk_y$ is negative.\footnote{With the $h_0=0$ symmetry, we have $N'_1=N''_1$ and $N'_2=N''_2$.} We have
\begin{equation}
N''_1-N''_2=\mathtt{C}.
\end{equation}
In case $h_0(\mathbf{k})\neq0$, the energy level is offset and the $E=0$ level is replaced by the middle line between the bulk energy gap.
\end{theorem}

\subsection{Uniform edge perturbation}\label{sec:uniform edge perturbation}
Recall our proof of the bulk-boundary correspondence starting from \eqnref{eqs'}. For purely $a$-type edge modes, the obtained candidate edge solutions as linear superpositions of \eqnref{am solutions} or \eqnref{am solutions'} are further subject to the ``edge equations'' \eqnref{eq' b left} and \eqnref{eq' b right}, of which one becomes superfluous (due to edge-mode exponential decay) and the other imposes additional $\bar{n}_x$ constraints upon the candidate solutions. As the role of these edge equations is mainly to place extra $\bar{n}_x$ constraints, their exact forms seem inessential as far as edge mode counting is concerned (also recall \footref{foot:extra constraints}).
Consequently, we should expect that many results in \secref{sec:special cases} remain true, even if the edge equations are deviated from the open boundary condition.

Let us study in more detail what will happen if, on top of the open boundary condition, we introduce ``uniform edge perturbation'', which perturbs only the edge regions and is translationally invariant along the $y$ direction. That is, the edge perturbation leaves the bulk equations \eqnref{eq' a bulk} and \eqnref{eq' b bulk} unchanged and modifies the edge equations \eqnref{eq' a left}, \eqnref{eq' a right}, \eqnref{eq' b left} and \eqnref{eq' b right} into
\begin{subequations}\label{eqs''}
\begin{eqnarray}
  \label{eq'' a left}
  E(k_y) a_{m_x} &=& h_z(k_y)a_{m_x} + \sum_{n_y=-\bar{n}_y}^{\bar{n}_y}e^{i n_y k_y}\sum_{n_x=-m_x+1}^{\bar{n}_x} (w^x_{n_x,n_y}-iw^y_{n_x,n_y}) b_{m_x+n_x}\nonumber\\
  && \mbox{} + \sum_{n_y=-\bar{n}_y}^{\bar{n}_y}e^{i n_y k_y}\sum_{m'_x=1}^{\bar{n}_x} (\zeta^0_{m_x,m'_x;n_y}+\zeta^z_{m_x,m'_x;n_y}) a_{m'_x} + (\zeta^x_{m_x,m'_x;n_y}-i\zeta^y_{m_x,m'_x;n_y}) b_{m'_x},\nonumber\\
  && \qquad\qquad
  \text{for}\ m_x=1,\dots,\bar{n}_x, \\
  \label{eq'' a right}
  E(k_y) a_{m_x} &=& h_z(k_y)a_{m_x} + \sum_{n_y=-\bar{n}_y}^{\bar{n}_y}e^{i n_y k_y}\sum_{n_x=-\bar{n}_x}^{N_x-m_x}   (w^x_{n_x,n_y}-iw^y_{n_x,n_y}) b_{m_x+n_x}\nonumber\\
  && \mbox{} + \sum_{n_y=-\bar{n}_y}^{\bar{n}_y}e^{i n_y k_y}\sum_{m'_x=N_x}^{N_x-\bar{n}_x+1} (\zeta^0_{m_x,m'_x;n_y}+\zeta^z_{m_x,m'_x;n_y}) a_{m'_x} + (\zeta^x_{m_x,m'_x;n_y}-i\zeta^y_{m_x,m'_x;n_y}) b_{m'_x},\nonumber\\
  && \qquad\qquad
  \text{for}\ m_x=N_x-\bar{n}_x+1,\dots,N_x, \\
  \label{eq'' b left}
  E(k_y) b_{m_x}&=& -h_z(k_y)b_{m_x} + \sum_{n_y=-\bar{n}_y}^{\bar{n}_y}e^{i n_y k_y}\sum_{n_x=-\bar{n}_x}^{m_x-1} (w^x_{-n_x,n_y}+iw^y_{-n_x,n_y}) a_{m_x -n_x}\nonumber\\
  && \mbox{} + \sum_{n_y=-\bar{n}_y}^{\bar{n}_y}e^{i n_y k_y}\sum_{m'_x=1}^{\bar{n}_x} (\zeta^x_{m_x,m'_x;n_y}+i\zeta^y_{m_x,m'_x;n_y}) a_{m'_x} + (\zeta^0_{m_x,m'_x;n_y}-\zeta^z_{m_x,m'_x;n_y}) b_{m'_x},\nonumber\\
  && \qquad\qquad
  \text{for}\ m_x=1,\dots,\bar{n}_x, \\
  \label{eq'' b right}
  E(k_y) b_{m_x}&=&  -h_z(k_y)b_{m_x} + \sum_{n_y=-\bar{n}_y}^{\bar{n}_y}e^{i n_y k_y}\sum_{n_x=m_x-N_x}^{\bar{n}_x}    (w^x_{-n_x,n_y}+iw^y_{-n_x,n_y})a_{m_x -n_x}\nonumber\\
  && \mbox{} + \sum_{n_y=-\bar{n}_y}^{\bar{n}_y}e^{i n_y k_y}\sum_{m'_x=N_x}^{N_x-\bar{n}_x+} (\zeta^x_{m_x,m'_x;n_y}+i\zeta^y_{m_x,m'_x;n_y}) a_{m'_x} +(\zeta^0_{m_x,m'_x;n_y}-\zeta^z_{m_x,m'_x;n_y}) b_{m'_x},\nonumber\\
  && \qquad\qquad
  \text{for}\ {m_x}=N_x-\bar{n}_x+1,\dots,N_x,
\end{eqnarray}
\end{subequations}
where $\zeta^a_{m_x,m'_x;n_y}$ for $m_x=1,\dots,\bar{n_x},N_x-\bar{n}_x+1,\dots,N_x$ are arbitrary constant parameters that give edge perturbation upon the coefficients $w^a_{n_x,n_y}$ (for $n_x=m'_x-m_x$) at different edge sites indexed by $m_x$,\footnote{Recall that in a special case we have $w^0_{n_x,n_y}=0$ and $w^z_{n_x\neq0,n_y}=0$ as shown in \eqnref{hz special}.} and the condition of hermiticity demands\footnote{See \appref{appendix} for a matrix form of \eqnref{eqs'} with \eqnref{eqs''} and the condition of hermiticity.}
\begin{equation}\label{zeta}
\zeta^{a*}_{m_x,m'_x;n_y}=\zeta^a_{m'_x,m_x;-n_y}.
\end{equation}

Firstly, consider the case that $\zeta^0_{m_x,m'_x;n_y}$ and $\zeta^z_{m_x,m'_x;n_y}$ all vanish but some of $\zeta^x_{m_x,m'_x;n_y}$ and $\zeta^y_{m_x,m'_x;n_y}$ are nonzero. For a purely $a$-type (i.e., $b_{m_x}=0$) left edge mode, \eqnref{eq'' a left} and \eqnref{eq'' a right} still yield $E(k_y)=h_z(k_y)$ and \eqnref{eq'' b right} still becomes superfluous due to the exponential decay. The only relevant change is in \eqnref{eq'' b left}, which imposes different extra $\bar{n}_x$ constraints upon the linear superposition of the candidate solutions. Consequently, an $a$-type left edge band remains purely $a$-type and following the trajectory of $h_z(k_y)$ with the same multiplicity given by $w(k_y)$, but the corresponding wavefunctions are altered as the modified extra $\bar{n}_x$ constraints permit different coefficients for the superposition of the candidate solutions.
For $a$-type right, $b$-type left, and $b$-type right edge modes, the similar conclusion can also be drawn correspondingly. The edge bands that follow the trajectories of $\pm h(k_y)$ are said to be robust under the edge perturbation with vanishing $\zeta^0_{m_x,m'_x;n_y}$ and $\zeta^z_{m_x,m'_x;n_y}$.

Meanwhile, the edge perturbation with nonzero $\zeta^x_{m_x,m'_x;n_y}$ or $\zeta^y_{m_x,m'_x;n_y}$ might give rise to additional edge bands that are neither purely $a$-type nor purely $b$-type and do not follow the trajectories of $\pm h_z(k_y)$. We refer to them as ``edge-perturbation-induced'' edge modes. If an edge-perturbation-induced edge band intercepts the trajectories of $\pm h_z(k_y)$, it becomes purely $a$-type or $b$-type at the interception point, as a consequence of substituting $E=\pm h_z$ into \eqnref{eqs'} and \eqnref{eqs''} with $\zeta^0_{m,m';n_y}=0$ and $\zeta^z_{m,m';n_y}=0$.\footnote{This is what might happen as mentioned in \footref{foot:degenerate constraints}.}

Secondly, consider the case that some of $\zeta^0_{m_x,m'_x;n_y}$ and $\zeta^z_{m_x,m'_x;n_y}$ are nonzero. Under the edge perturbation of this kind, imposing $b_{m_x}=0$ upon \eqnref{eq'' a left} or \eqnref{eq'' a right} yields $E(k_y)a_{m_x}=E(k_y)a_{m_x}+\dots$, which is in direct conflict with $E(k_y)a_{m_x}=E(k_y)a_{m_x}$ implied by the bulk equation \eqnref{eq' a bulk} with $b_{m_x}=0$. Imposing $a_{m_x}=0$ leads to a similar contradiction. Therefore, edge modes are no longer purely $a$-type or $b$-type (unless accidentally at some points) even at the points where edge bands intercept the trajectories of $\pm h_z(k_y)$. The original $a$-type and $b$-type edge bands are deformed from the trajectories of $\pm h_z(k_y)$ and become mixed in $a_{m_x}$ and $b_{m_x}$ under the edge perturbation.
Meanwhile, just like the previous case of vanishing $\zeta^0_{m_x,m'_x;n_y}$ and $\zeta^z_{m_x,m'_x;n_y}$, the edge perturbation might give rise to more edge bands.

With the inclusion of edge perturbation, the $h_0=0$ symmetry discussed in \secref{sec:h0=0 symmetry'} is broken in general.
Nevertheless, in the case of $\zeta^0_{m_x,m'_x;n_y}=0$ and $\zeta^z_{m_x,m'_x;n_y}=0$, the energy spectrum still exhibits the symmetric feature that the energy eigenvalues always appear in pairs with opposite signs as predicted in \thmref{theorem 3}. For the edge modes following the trajectories of $\pm h_z(k_y)$, this symmetry associates a \emph{left} edge mode of $E=\pm h_z$ with a \emph{right} edge mode of $E=\mp h_z$, but the corresponding wavefunctions are no longer related with each other via \eqnref{tilde Psi ky}. On the other hand, for edge-perturbation-induced edge modes, this symmetry associates a \emph{left} (right) edge mode of $E=E_0$ with a different \emph{left} (right) edge mode of $E=-E_0$.\footnote{Edge-perturbation-induced edge modes arise as a result of edge perturbation. As the left (right) induced edge modes know nothing about the edge perturbation upon the right (left) region (because of edge-mode exponential decay), the symmetry of \thmref{theorem 3} must be between two left (right) induced edge modes.} Also see the comment in the last paragraph of \appref{appendix} for the comparison of the symmetry of \thmref{theorem 3} and the $h_0=0$ symmetry.

Under uniform edge perturbation, the original edge bands may or may not deformed and more edge bands may or may not arise, but in any case \thmref{theorem 1} and \thmref{theorem 2} still hold true as will be explained in the next subsection.
We will also demonstrate concrete examples of edge perturbation in \secref{sec:edge perturbation I} and \secref{sec:edge perturbation II}

\subsection{Generic cases}\label{sec:generic cases}
We have proved the bulk-boundary correspondence for the special case with $h_0(\mathbf{k})=0$ and $h_z(\mathbf{k})=h_z(k_y)$ and without edge perturbation. In this subsection, we will show that \thmref{theorem 1} and \thmref{theorem 2} in fact remain valid for any generic cases of $h_a(\mathbf{k})$ as well as under any uniform edge perturbation, mainly because any given $h_0(\mathbf{k})$ and $\mathbf{h}(\mathbf{k})$ can always be adiabatically deformed into a special case.

For any arbitrary map $h_a: \mathbf{k}\in T^2 \mapsto (\mathbf{h}(\mathbf{k}),h_0(\mathbf{k}))\in \left(\mathbb{R}^3\backslash\{0\}\right)\times\mathbb{R}$, it is trivial to see that the part $h_0(\mathbf{k})$ can be continuously deformed into $h_0(\mathbf{k})=0$ without deforming the part of $\mathbf{h}(\mathbf{k})$. On the other hand, at first thought, it seems impossible that an arbitrary $\mathbf{h}(\mathbf{k})$ can always be continuously deformed into a special case with $h_z(k_x,k_y)=h_z(k_y)$ while maintaining gapped in the bulk spectrum (i.e., $\mathbf{h}(\mathbf{k})\neq0$ for any $\mathbf{k}$). Take \figref{fig:tori} as an example, where $\mathbf{h}(\mathbf{k})$ as a map from $\mathbf{k}\in T^2\equiv[0,2\pi]\times[0,2\pi]$ to $\mathbf{h}\in\mathbb{R}^3\backslash\{0\}$ is illustrated as a torus embedded in the $\mathbf{h}$ space. As the special case depicted in (a) has constant-$k_y$ loops all in the same orientation while the generic case depicted in (b) has constant-$k_y$ loops paired in opposite orientations, it looks dubious that they can be adiabatically deformed into each other although they give the same Chern number. However, one can adiabatically deform (a) into (c) by ``pinching'' the tube of the torus in a way that the part not close to the origin is shrunk into a one-dimensional line (i.e., constant-$k_y$ loops are shrunk into single points) while the part close to the origin remain bulged. Likewise, (b) can be adiabatically deformed into (d). It is then obvious that (c) and (d) can be adiabatically deformed into each other. Therefore, via $\text{(a)} \leftrightarrow \text{(c)} \leftrightarrow \text{(d)} \leftrightarrow \text{(b)}$, it turns out the special case (a) and the generic case (b) can be adiabatically deformed into each other. Generically, $\mathbf{h}(\mathbf{k})$ might be more complicated than depicted in \figref{fig:tori} and the torus embedded in the $\mathbf{h}$ space may coil around the origin several times, depending on its Chern number. In any case, we can always adiabatically deform it into a special case with the same coiling structure by the same strategy illustrated in \figref{fig:tori}.

\begin{figure}

\centering
    \includegraphics[width=0.7\textwidth]{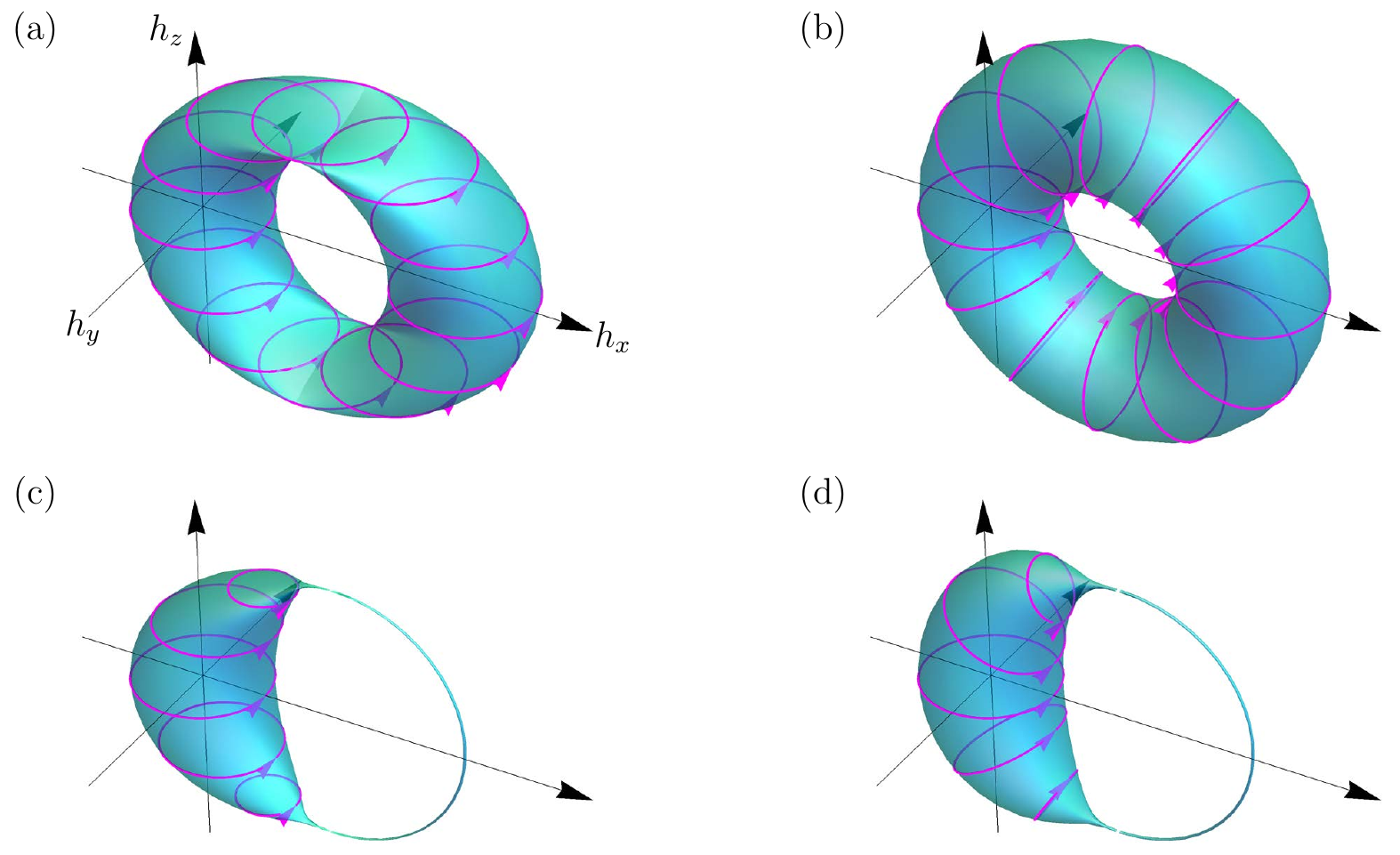}

\caption{$\mathbf{h}: \mathbf{k}\mapsto \mathbf{h}(\mathbf{k})$ is illustrated as a torus embedded in the $\mathbf{h}$ space. Trajectories $\{\mathbf{h}(k_x,k_y)|k_x\in[0,2\pi]\}$ of constant $k_y$ are depicted as oriented loops (called ``constant-$k_y$ loops''). (a): An example of a special case with $h_z(\mathbf{k})=h_z(k_y)$. (b): An example of a generic case. (c): An example as adiabatically deformed from (a). (d): An example as adiabatically deformed from (b). All four examples are of $\mathtt{C}=1$.}\label{fig:tori}
\end{figure}

Once we realize that any generic $h_a(\mathbf{k})$ can be adiabatically deformed from a special case, the energy spectrum of both bulk and edge modes for a generic $h_a(\mathbf{k})$ can be viewed as continuously deformed from that of a corresponding special case. Furthermore, if edge perturbation is introduced, the energy spectrum can also be viewed as continuously deformed from that without edge perturbation.
When a special case is adiabatically deformed into a generic case or continuously deformed to include uniform edge perturbation, the upper and lower bulk bands remain gapped but a few qualitative changes can happen for the edge bands as discussed in the following.

Firstly, if $h_0(\mathbf{k})=0$ is continuously deformed into $h_0(\mathbf{k})\neq0$ in the $\mathbf{k}$ space, for a given $k_y$, each discrete eigenvalue $E(k_y)$ (whether it corresponds to a bulk mode or an edge mode) in the spectrum of $\hat{H}_{N_x}(k_y)$ will continuously shift accordingly to \eqnref{h square}. The energy shift of each spectrum line is also continuous with respect to the change of $k_y$. As a result, the whole band structure of spectrum lines (both bulk and edge bands included) against $k_y$ is continuously deformed into a new one. Provided that $\abs{\nabla_\mathbf{k}h_0(\mathbf{k})}$ is small enough so that energies of the upper and lower bulk bands do not overlap, the system remains a bulk insulator. In this case, the numbers of upper-to-lower and lower-to-upper edge bands remain the same under the continuous deformation upon the whole band structure. Therefore, \thmref{theorem 1} and \thmref{theorem 2} still hold up, except that the energy level is offset and accordingly the $E=0$ level is replaced by the middle line between the bulk energy gap.

Secondly, unlike the special case, $h_z$ is no longer a function of $k_y$ alone and it makes no sense to talk about the trajectory of $\pm h_z(k_y)$.
Consequently, the upper-to-lower and lower-to-upper edge modes in general are no longer purely $a$-type or $b$-type. The positions of $k_y$ where the edge bands branch out from the bulk band clusters are no longer identified with north or south poles, either. Furthermore, the fact that the left (right) edge bands connecting the upper and lower bulk band clusters cannot intercept one another in a special case no longer holds true in a generic case. Two left (right) edge bands that do not intercept each other can be deformed into two left (right) edge bands that intercept each other and \textit{vice versa}.

Thirdly, deformations between (a) and (b) and between (d) and (e) in \figref{fig:deformations} can also take place when a special case is deformed into a generic case.

Finally, more different kinds of change are possible as depicted in \figref{fig:deformations2}.\footnote{\label{foot:only schematic}\figref{fig:deformations} and \figref{fig:deformations2} are only schematic. Under deformation, both the bulk bands and the edge bands are deformed, but only the deformation of the latter is depicted. Also note that edge bands in general can appear anywhere --- within the bulk gap, above the upper bulk band cluster, below the lower bulk band cluster, or even inside the bulk band clusters (see (a) and (b) of \figref{fig:edge perturbation I}). \figref{fig:deformations} and \figref{fig:deformations2} only depict the case of edge bands within the bulk gap, and accordingly the upper (lower) contour of the upper (lower) bulk band cluster is not shown.} We focus solely on left edge bands in the following discussion, as the situation for right edge bands is similar.
For (a)--(c): In a special case, an upper-to-lower (or lower-to-upper) edge band may have degeneracy (if the multiplicity $\abs{w(k_y)}>0$). When the special case is deformed into a generic case, the degeneracy is in general lifted and the edge band is split into many bands as shown as the change from (a) to (b). However, the change from (b) to (c) is impossible for the same reason for the no-go patterns in \figref{fig:deformations}.
For (d)--(f): A new upper-to-upper (or lower-to-lower) edge band may appear as splitting from the bulk band cluster as shown in (d);\footnote{A new upper-to-upper edge band can appear above the bulk band cluster or inside the cluster as well as within the bulk gap, although only the case within the bulk gap is depicted.} conversely, an upper-to-upper edge band may disappear as submerging into the bulk band cluster. The upper-to-upper edge band can be stretched downwards to the extent that it touches the lower (upper) bulk band cluster as shown in (e). The configuration in (e) can be further morphed into two edge bands connecting the upper and lower bulk band clusters as shown in (f). The converse changes from (f) to (e) and (e) to (d) are also possible.
For (g)--(i): An upper-to-lower edge band and a lower-to-upper edge band as shown in (g) can be deformed into an upper-to-upper edge band and a lower-to-lower left edge band as shown in (h). The converse from (h) to (g) is also possible. However, (g) cannot be morphed into (i) for the same reason for the no-go patterns in \figref{fig:deformations}. It is noteworthy that the pattern of (g) is impossible in a special case, where left edge bands connecting the upper and lower bulk band clusters follow the trajectory of $\pm h_z(k_y)$ and do not intercept one another. Nevertheless, (g) is possible in a generic case, as the pattern of (e) or (f) in a special case can be morphed into (g).
For (j)--(l): An upper-to-upper (or lower-to-lower) edge band and an upper-to-lower (or lower-to-upper) edge band as shown in (j) can be deformed to touch each other as shown in (k). They can be further morphed into a single edge band connecting the upper and lower bulk band clusters as shown in (l). The converse changes from (l) to (k) and (k) to (j) are also possible.
For (m)--(o): Two upper-to-upper (lower-to-lower) edge bands as shown in (m) can be merged into a single upper-to-upper (lower-to-lower) edge band as shown in (n). An edge band as in (n) can be deformed into a standalone edge mode as shown in (o). Conversely, a standalone edge band as in (o) can be deformed to touche the bulk band and become a nonstandalone edge band as in (n); an edge band as in (n) can split into two edge bands as in (m).

It should be emphasized that in \figref{fig:deformations2} two edge bands depicted in a same plot are both left (or, equivalently, right) edge bands. A left edge band and a right edge band cannot merge with each other as two left edge bands do, because the left edge band remains localized on the left edge while the right edge band on the right edge.
Therefore, under deformations, the left edge modes and right edge modes can be treated as independent of each other. If $h_0(\mathbf{k})=0$ remains true, the left edge modes and the right edge modes are related via \eqnref{tilde Psi ky}.

\begin{figure}

\centering
    \includegraphics[width=0.8\textwidth]{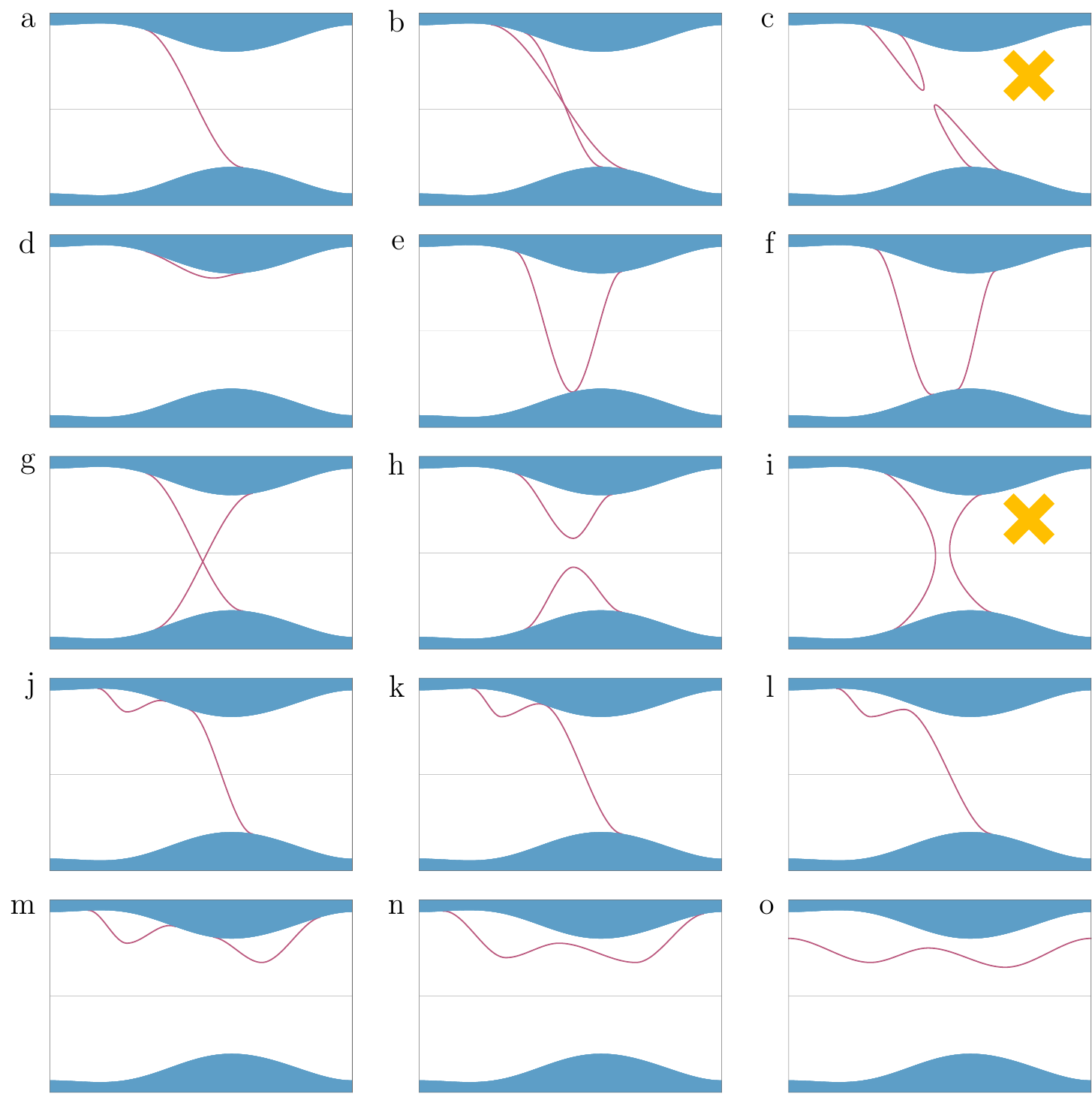}

\caption{Various examples of possible configurations of edge bands and impossible ones (marked with no-go signs) in generic cases.}\label{fig:deformations2}
\end{figure}

Taking into account all possible and forbidden deformations discussed above and the fact that the left and right edge modes do not mix with each other, we can easily draw the conclusion that, even though $N_{1,2}$, $N'_{1,2}$, and $N''_{1,2}$ as defined in \thmref{theorem 1} and \thmref{theorem 2} are in general altered under the deformation from a special case to a generic case and under arbitrary uniform edge perturbation, $N_1-N_2$, $N'_1-N'_2$, and $N''_1-N''_2$ nevertheless remain fixed.
Therefore, we have proved \thmref{theorem 1} and \thmref{theorem 2} for any arbitrary cases including those with uniform edge perturbation.

\subsection{Semi-special cases}\label{sec:semi-special cases}
For special cases, in addition to the topological features in relation to edge mode numbers, we also know a great deal about nontopological traits of the energy spectrum --- particularly, the trajectories of $\pm h_z(k_y)$ give rise to edge bands with the multiplicity given by $\abs{w(k_y)}$ and these edge bands are purely $a$-type or $b$ type. Therefore, even before conducting numerical computation for the energy spectrum of a strip, the energy spectrum (for both bulk and edge modes included) can be largely anticipated from $\mathbf{h}(\mathbf{k})$ alone, except for the finite size effect (which is negligible when $N_x$ is large enough) and possibly extra edge bands that do not follow the trajectories of $\pm h_z(k_y)$ and are most likely induced by edge perturbation.

For generic cases, although the bulk-boundary correspondence still holds up, we are unable to foresee nontopological features in advance as for special cases. However, if a generic case happens to be ``semi-special'', the nontopological features can be anticipated again. The semi-special condition is satisfied if $h_0(\mathbf{k})=0$ and every constant-$k_y$ loop lies on a planar surface in the $\mathbf{k}$ space. See (c) of \figref{fig:semi-special case I} and (f) of \figref{fig:semi-special case II} for examples of the semi-special case. Also note that many renowned models, such as the Rice-Mele model \cite{rice1982elementary}, the Haldane model \cite{haldane1988model}, and the Qi-Wu-Zhang model \cite{qi2006topological}, satisfy the semi-special condition.

Let $\hat{\mathbf{n}}(k_y)$ be the normal unit vector of the plane where the constant-$k_y$ loop lies (and the orientation of the plane is defined as that of the constant-$k_y$ loop). For a given $k_y$, we can always rotate the three axes  $(h_x,h_y,h_z)$ in the $\mathbf{k}$ space into the new ones $(h'_x,h'_y,h'_z)$ such that the direction of $h'_z$ is aligned with $\hat{\mathbf{n}}(k_y)$. Then, in the vicinity of $k_y$, we have $h'_z(k_x,k_y)\approx h'_z(k_y)$ and the system can be temporarily viewed as a special case by treating $h'_z$ as $h_z$. The results obtained  in \secref{sec:special cases} for the special case can then be straightforwardly carried over for the semi-special case in the vicinity of $k_y$, except that, in accordance with the rotation from $(h_x,h_y,h_z)$ to $(h'_x,h'_y,h'_z)$, the components $(a_{m_x},b_{m_x})$ are replaced by $(a^\bot_{m_x},b^\bot_{m_x})$, where the spinor $a_{m_x}\ket{\uparrow}+b_{m_x}\ket{\downarrow}$ is adaptively recast as $a^\bot_{m_x}\ket{\uparrow_\bot}+b^\bot_{m_x}\ket{\downarrow_\bot}$ with $\ket{\uparrow_\bot}$ and $\ket{\downarrow_\bot}$ being the two eigenstates of $\hat{\mathbf{n}}(k_y)\cdot\boldsymbol{\sigma}$.

Let $h_\bot(k_y)$ be a function of $k_y$ defined as $h_\bot:k_y\mapsto h'_z(k_y)$. We arrive at the conclusion for the semi-special case that the trajectories of $\pm h_\bot(k_y)$ against $k_y$ give rise to edge bands with the multiplicity given by $\abs{w_\bot(k_y)}$, which is the winding number of the constant-$k_y$ loop winding around the $h'_z(k_y)$ axis. These edge bands are no longer purely $a$-type or $b$-type; instead, they become purely $a^\bot$-type or $b^\bot$-type (i.e., either $b^\bot_{m_x}=0$ or $a^\bot_{m_x}=0$ for all $m_x$). That is, the semi-special condition entails ``spin-momentum locking'': the (pseudo)spin is either parallel or antiparallel to the direction of $\hat{\mathbf{n}}(k_y)$ in edge states that follow the trajectories of $\pm h_\bot(k_y)$.

It should be noted that in general one cannot associate any given constant-$k_y$ loop with a unique direction like $\hat{\mathbf{n}}(k_y)$, unless the semi-special condition is satisfied. Therefore, the notion of spin-momentum locking does not make sense beyond semi-special cases. The spin-momentum locking is not a topological trait in the strict sense that it is robust under any arbitrary adiabatic deformations; rather, it is robust only under deformations within the confines of the semi-special condition.

Finally, the semi-special condition might be satisfied only in a local open neighborhood of $k_y$. If this happens, all the consequences discussed above hold true locally inside the neighborhood.

\section{Examples}\label{sec:examples}
In order to illustrate various concepts and results we have obtained, we investigate a few concrete examples of special, semi-special, and generic cases, and perform numerical computation for the energy spectrum of a strip.\footnote{The exact numerical values of various parameters and $N_x$ are carefully chosen to give optimal illustration. Particularly, $N_x$ is chosen adaptively to be big enough so that the finite-size effect is negligible and small enough so that the spectrum lines are distinctively visible.} All the examples confirm what we have discussed, especially \thmref{theorem 1} and \thmref{theorem 2}. For simplicity, we set $h_0(\mathbf{k})=0$ and focus only on left edge modes as the right counterparts are simply related via the $h_0=0$ symmetry.
In \secref{sec:edge perturbation I} and \secref{sec:edge perturbation II}, we also take into account edge perturbation upon a special case.

\subsection{Special case I}
We first study a special case with the bulk momentum-space Hamiltonian $\hat{H}(\mathbf{k})=\mathbf{h}\cdot\boldsymbol{\sigma}$ given by
\begin{subequations}\label{special case I}
\begin{eqnarray}
  h_x(k_x,k_y) &=& 1-0.2 \sin k_y + \cos 2k_y +0.4 \cos k_x, \\
  h_y(k_x,k_y) &=& 0.4 \sin k_x, \\
  h_z(k_x,k_y) &=& h_z(k_y) = 0.16\cos k_y + 0.8 \sin 2k_y.
\end{eqnarray}
\end{subequations}
The map $\mathbf{h}:T^2\rightarrow \mathbb{R}^3\setminus\{0\}$ can be visualized as a torus embedded in the $\mathbf{h}$ space as shown in (c) of \figref{fig:special case I}. Each constant-$k_y$ loop is a regular circle, lying on a plane perpendicular to the $h_z$ axis (i.e., $h_z(\mathbf{k})=h_z(k_y)$). The torus as a tube of constant-$k_y$ loops coils around the origin $\mathbf{h}=0$ twice, thus yielding the Chern number $\mathtt{C}=2$.
The regions of the bulk band clusters for a strip are delineated by $\pm\min_{k_x}\abs{\mathbf{h}(k_x,k_y)}$ and $\pm\max_{k_x}\abs{\mathbf{h}(k_x,k_y)}$, provided $N_x$ is large enough. The regions of the bulk band clusters and the trajectory of $h_z(k_y)$ are depicted in (a) of \figref{fig:special case I}.

To see the energy spectrum of a strip, we numerically solve the eigenvalue problem \eqnref{eqs} --- or equivalently \eqnref{eqs'} --- associated wih \eqnref{special case I}. The energy spectrum with $N_x$ chosen to be $N_x=15$ is depicted in (b) of \figref{fig:special case I}. The schematic illustration of \figref{fig:band scheme} is affirmed: The edge bands follow the trajectories of $\pm h_z(k_y)$ with the multiplicity given by $\abs{w(k_y)}$. The wavefunctions for two particular points in the edge bands are depicted in terms of $\abs{a_{m_x}},\abs{b_{m_x}}$ in (d) and (c). They are localized at the left edge and exponentially decayed towards the right edge. Both of them are purely $a$-type.


\begin{figure}

\centering
    \includegraphics[width=0.95\textwidth]{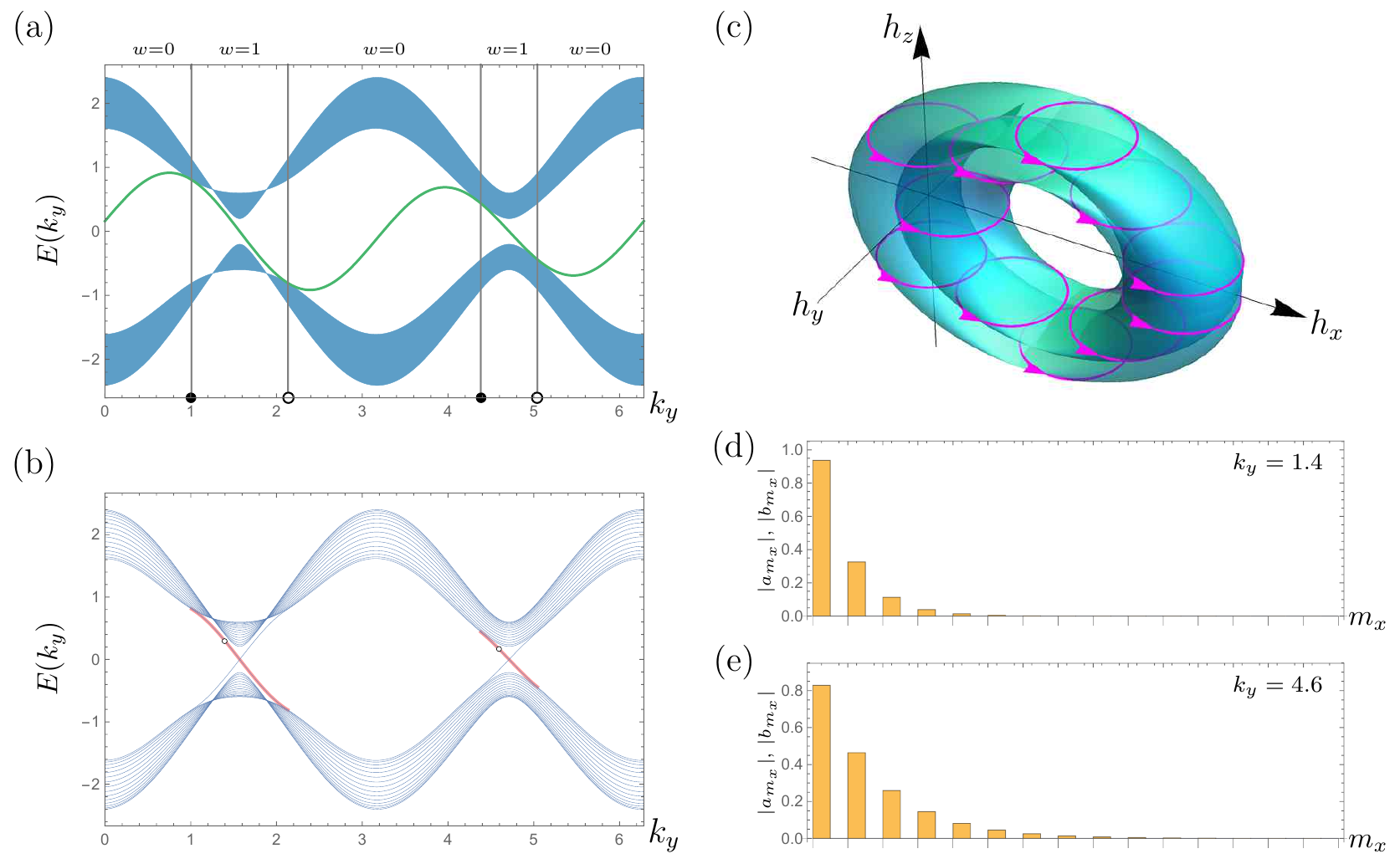}

\caption{\textbf{Special case I} with \eqnref{special case I}. (a): The anticipated regions of the bulk bands (shaded areas) and the trajectory of $h_z(k_y)$ (solid curve). The positions of the north and south poles are indicated by solid and hollow dots, respectively. The winding number $w(k_y)$ for different portions of $k_y$ is indicated on the top. (b): The energy spectrum of a strip with $N_x=15$. The edge bands appear within the bulk gap, following the trajectories of $\pm h_z(k_y)$. The left edge bands are highlighted by thick curves. (c): $\mathbf{h}: \mathbf{k}\mapsto \mathbf{h}(\mathbf{k})$ illustrated as a torus embedded in the $\mathbf{h}$ space. The oriented loops represent various constant-$k_y$ loops. (d) and (e): The edge mode wavefunctions presented as bar charts of $\abs{a_1},\abs{b_1},\dots,\abs{a_{N_x}},\abs{b_{N_x}}$ for the two dotted points as indicated in (b). They are localized at the left edge and purely $a$-type ($b_{m_x}=0$).}\label{fig:special case I}
\end{figure}

\subsection{Semi-special case I}
Next, we study a semi-special case given by
\begin{subequations}\label{semi-special case I}
\begin{eqnarray}
  h_x(k_x,k_y) &=& 1-0.2 \sin k_y + (1-0.4 \cos k_x ) \cos 2k_y, \\
  h_y(k_x,k_y) &=& 0.4 \sin k_x, \\
  h_z(k_x,k_y) &=& 0.16\cos k_y + (1-0.4 \cos k_x ) \sin 2k_y.
\end{eqnarray}
\end{subequations}
Each constant-$k_y$ loop is a regular circle lying on a planar surface in the $\mathbf{k}$ space. Unlike the special case, however, we no longer have $h_z(\mathbf{k})=h_z(k_y)$.
As depicted in (c) of \figref{fig:semi-special case I}, \eqnref{semi-special case I} can be viewed as deformed from \eqnref{special case I} in the similar manner that (a) and (b) in \figref{fig:tori} are related to each other.

The energy spectrum of a strip with $N_x=15$ is depicted in (b) of \figref{fig:semi-special case I}. The edge bands follow the trajectories of $\pm h_\bot(k_y)$. The wavefunctions of edge modes are no longer purely $a$-type or $b$-type but have both nonzero $a_{m_x}$ and nonzero $b_{m_x}$ as shown in (d) and (e). However, if we transform $a_{m_x}$ and $b_{m_x}$ into $a^\bot_{m_x}$ and $b^\bot_{m_x}$, the edge modes become purely $a^\bot$-type or $b^\bot$-type as shown in (f) and (g).


\begin{figure}

\centering
    \includegraphics[width=0.95\textwidth]{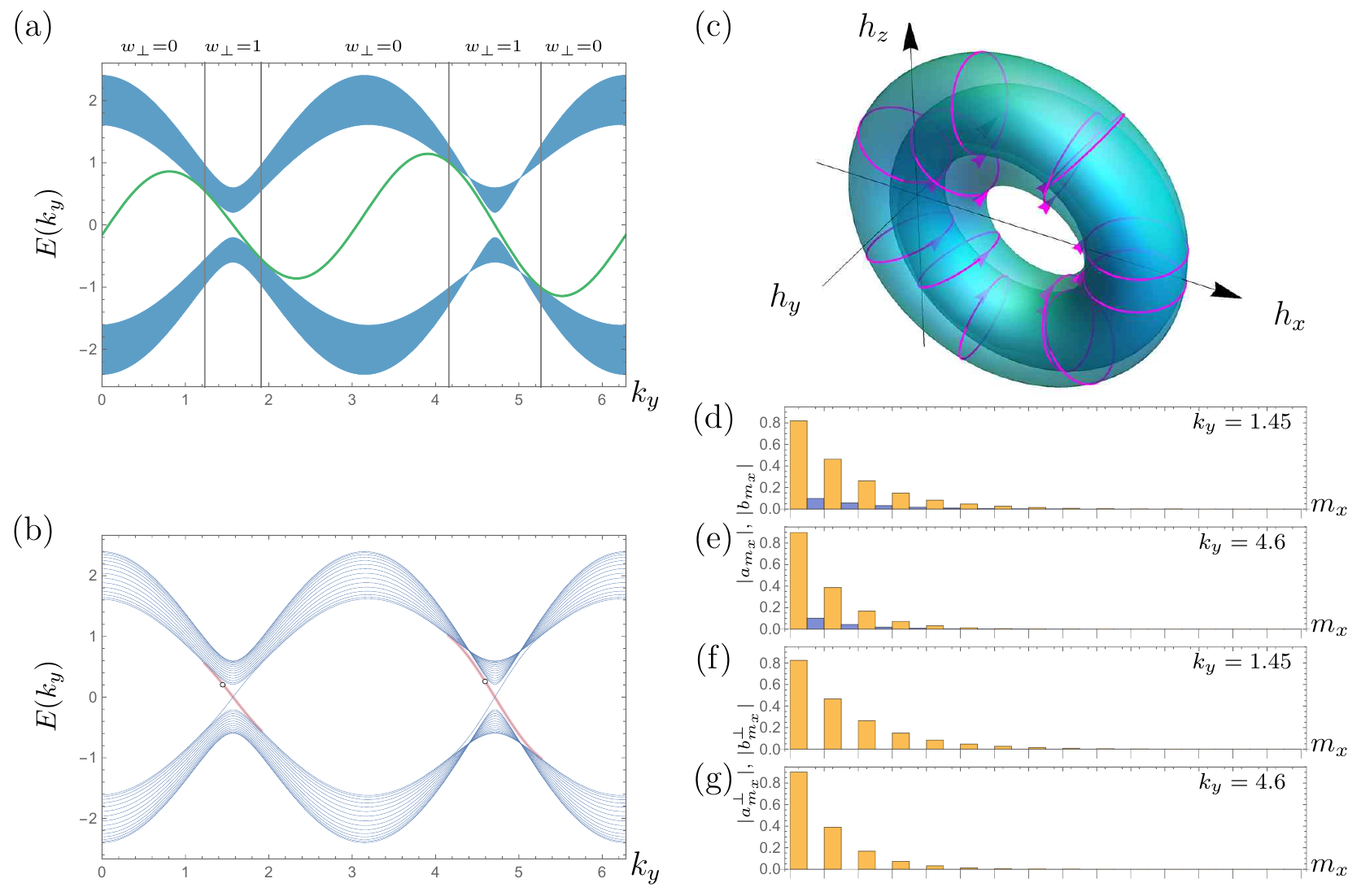}

\caption{ \textbf{Semi-special case I} with \eqnref{semi-special case I}. (a): The anticipated regions of the bulk bands and the trajectory of $h_\bot(k_y)$ (solid curve). The winding number $w_\bot(k_y)$ around the $h'(k_y)$ axis for different portions of $k_y$ is indicated on the top. (b): The energy spectrum of a strip with $N_x=15$. The edge bands appear within the bulk gap, following the trajectories of $\pm h_\bot(k_y)$. The left edge bands are highlighted by thick curves. (c): $\mathbf{h}: \mathbf{k}\mapsto \mathbf{h}(\mathbf{k})$ illustrated as a torus embedded in the $\mathbf{h}$ space. (d) and (e): The edge mode wavefunctions for the two dotted points as indicated in (b). They are localized at the left edge but no longer purely $a$-type or $b$-type. (f) and (g): The same wavefunctions of (d) and (e) shown in terms of $\abs{a^\bot_{m_X}}$ and $\abs{b^\bot_{m_X}}$. They are purely $a^\bot$-type.}\label{fig:semi-special case I}
\end{figure}

\subsection{Generic case I}
We further deform the semi-special case \eqnref{semi-special case I} into a generic case given by
\begin{subequations}\label{generic case I}
\begin{eqnarray}
  h_x(k_x,k_y) &=& 1-0.2 \sin k_y + (1-0.4 \cos k_x) \cos 2k_y,  \\
  h_y(k_x,k_y) &=& 0.4 \sin k_x, \\
  h_z(k_x,k_y) &=& 0.16 \cos k_y + (1-0.4 \cos k_x +0.3 \sin 2k_x) \sin 2k_y.
\end{eqnarray}
\end{subequations}
This case is truly generic in the sense that each constant-$k_y$ loop no longer lies on a planar surface in the $\mathbf{h}$ space as can be seen in (c) and (d) of \figref{fig:generic case I}.
The trajectories of edge bands cannot be anticipated in advance, but the energy spectrum shown in (b) of \figref{fig:generic case I} can be understood as deformed from (b) of \figref{fig:semi-special case I}. The edge modes are not purely $a$-type or $b$-type. Unlike the semi-special case, we can no longer associate the constant $k_y$-loop with a unique direction.


\begin{figure}

\centering
    \includegraphics[width=0.95\textwidth]{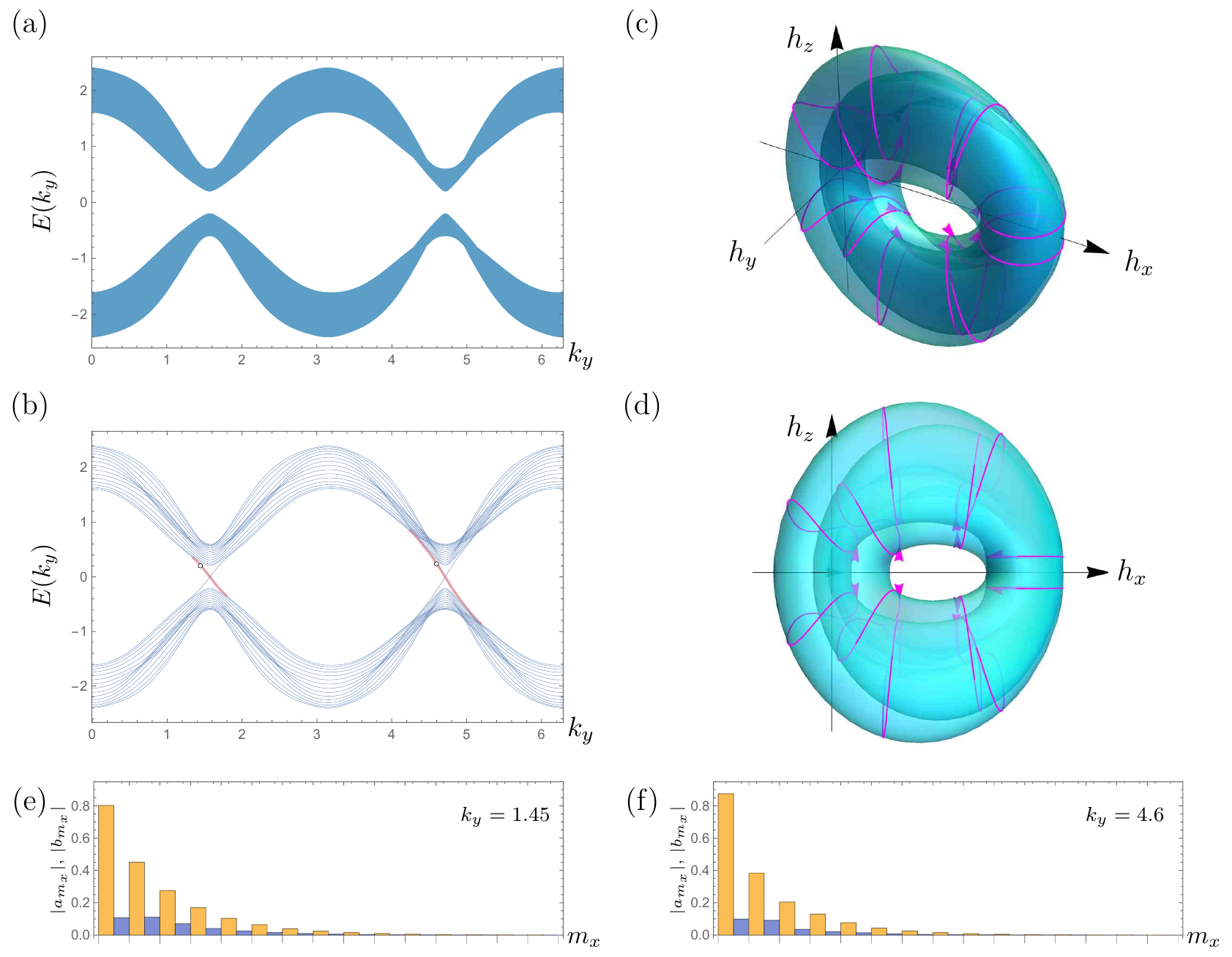}

\caption{\textbf{Generic case I} with \eqnref{generic case I}. (a): The anticipated regions of the bulk bands. (b): The energy spectrum of a strip with $N_x=15$. The edge bands appear within the bulk gap. The left edge bands are highlighted by thick curves. (c): $\mathbf{h}: \mathbf{k}\mapsto \mathbf{h}(\mathbf{k})$ illustrated as a torus embedded in the $\mathbf{h}$ space. (d): The side view of (c). (e) and (f): The edge mode wavefunctions for the two dotted points as indicated in (b).}\label{fig:generic case I}
\end{figure}

\subsection{Special case II}
We also study a special case with a different $\mathbf{h}(\mathbf{k})$ given by
\begin{subequations}\label{special case II}
\begin{eqnarray}
  h_x(k_x,k_y) &=& 0.8+\cos k_y +0.4\cos k_x + 0.5 \cos (2k_x-\pi), \\
  h_y(k_x,k_y) &=& 0.4\sin k_x + 0.5\sin (2k_x-\pi), \\
  h_z(k_x,k_y) &=& h_z(k_y) = 0.3\sin k_y.
\end{eqnarray}
\end{subequations}
As visualized in (c) of \figref{fig:special case II}, each constant-$k_y$ loop is a rhodonea curve, lying on a plane perpendicular to the $h_z$ axis. The rhodonea curves of different values of $k_y$ winds around the $h_z$ axis with the different winding numbers $w=0$, $w=1$, and $w=2$, consequently yielding the Chern number $\mathtt{C}=2$.

The energy spectrum with $N_x=25$ is depicted in (b) of \figref{fig:special case II}. The schematic illustration of \figref{fig:band scheme} is affirmed again: The edge bands follow the trajectories of $\pm h_z(k_y)$ with the multiplicity given by $\abs{w(k_y)}$. The wavefunctions of edge bands are purely $a$-type or $b$-type as depicted in (d)--(f).


\begin{figure}

\centering
    \includegraphics[width=0.95\textwidth]{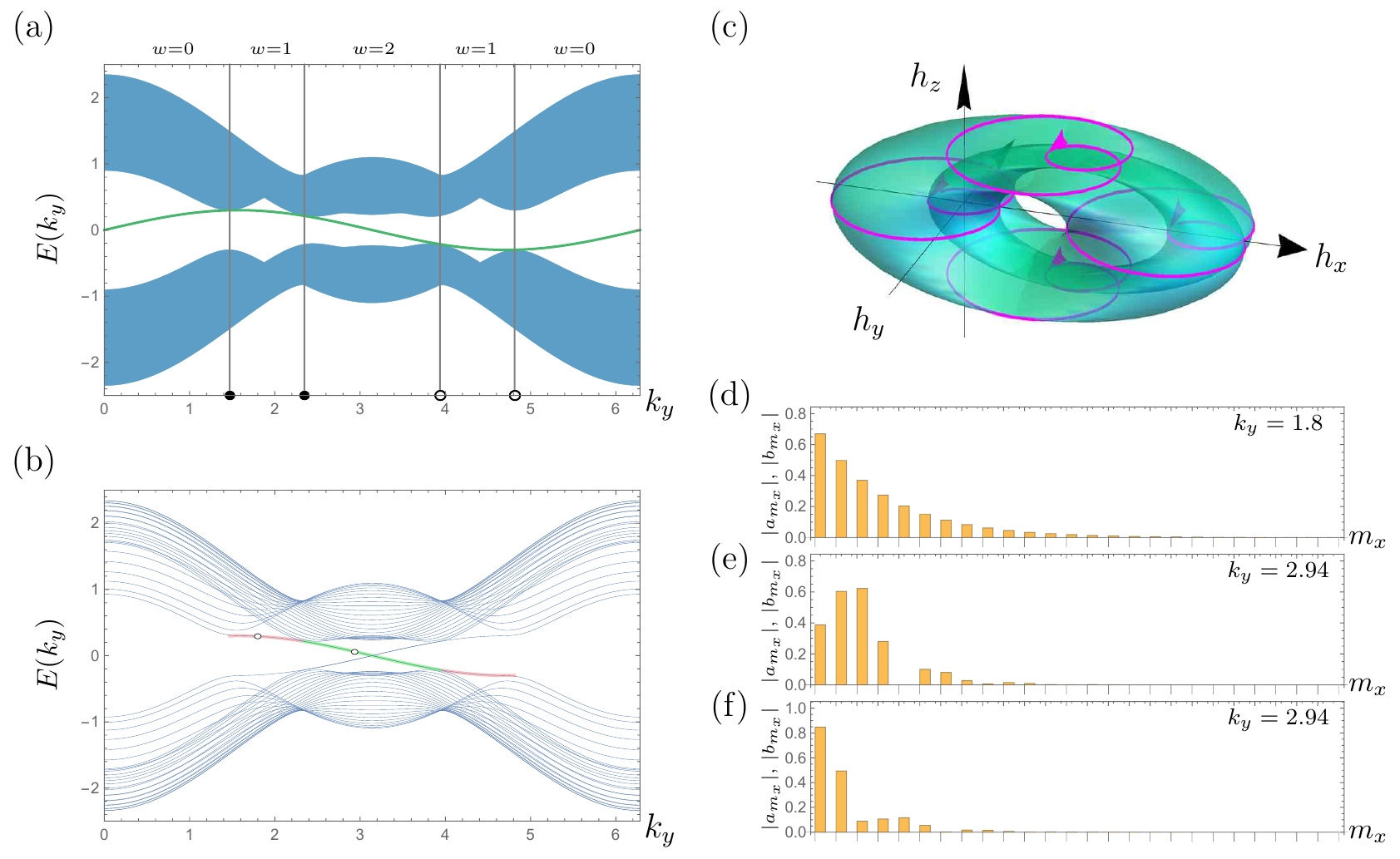}

\caption{\textbf{Special case II} with \eqnref{special case II}. (a): The anticipated regions of the bulk bands and the trajectory of $h_z(k_y)$. The positions of the north and south poles is indicated by solid and hollow dots, respectively. The winding number $w(k_y)$ for different portions of $k_y$ are indicated on the top. (b): The energy spectrum of a strip with $N_x=25$. The edge bands appear within the bulk gap, following the trajectories of $\pm h_z(k_y)$. The left edge bands are highlighted by thick curves. (c): $\mathbf{h}: \mathbf{k}\mapsto \mathbf{h}(\mathbf{k})$ illustrated as a torus embedded in the $\mathbf{h}$ space. Each constant-$k_y$ loop is a rhodonea curve. (d)--(f): The edge mode wavefunctions for the two dotted points as indicated in (b). They are localized at the left edge and purely $a$-type. Note that the edge bands in the $w=2$ portion of $k_y$ (e.g., $k_y=2.94$) are two-fold degenerate --- there are two eigenstates for the same edge band point as shown in (e) and (f).}\label{fig:special case II}
\end{figure}

\subsection{Semi-special case II}
Next, we study a semi-special case given by
\begin{subequations}\label{semi-special case II}
\begin{eqnarray}
  h_x(k_x,k_y) &=& 0.8+\big(1-0.4\cos k_x + 0.5\cos (2k_x-\pi)\big)\cos k_y, \\
  h_y(k_x,k_y) &=& 0.4\sin k_x + 0.5\sin (2k_x-\pi), \\
  h_z(k_x,k_y) &=& 0.3 \big(1-0.4\cos k_x + 0.5\cos (2k_x-\pi)\big)\sin k_y.
\end{eqnarray}
\end{subequations}
As depicted in (f) of \figref{fig:semi-special case II}, \eqnref{semi-special case II} can be viewed as deformed from \eqnref{special case II} in the similar manner that (a) and (b) in \figref{fig:tori} are related to each other.

The energy spectrum of a strip with $N_x=25$ is depicted in (d) and (e) \figref{fig:semi-special case II}. The edge bands follow the trajectories of $\pm h_\bot(k_y)$ with the multiplicity given by $\abs{w_\bot(k_y)}$.
The edge mode wavefunctions are shown in \figref{fig:semi-special case II p2}.
The edge modes are no longer purely $a$-type or $b$-type, but they are purely $a^\bot$-type or $b^\bot$-type.


\begin{figure}

\centering
    \includegraphics[width=0.95\textwidth]{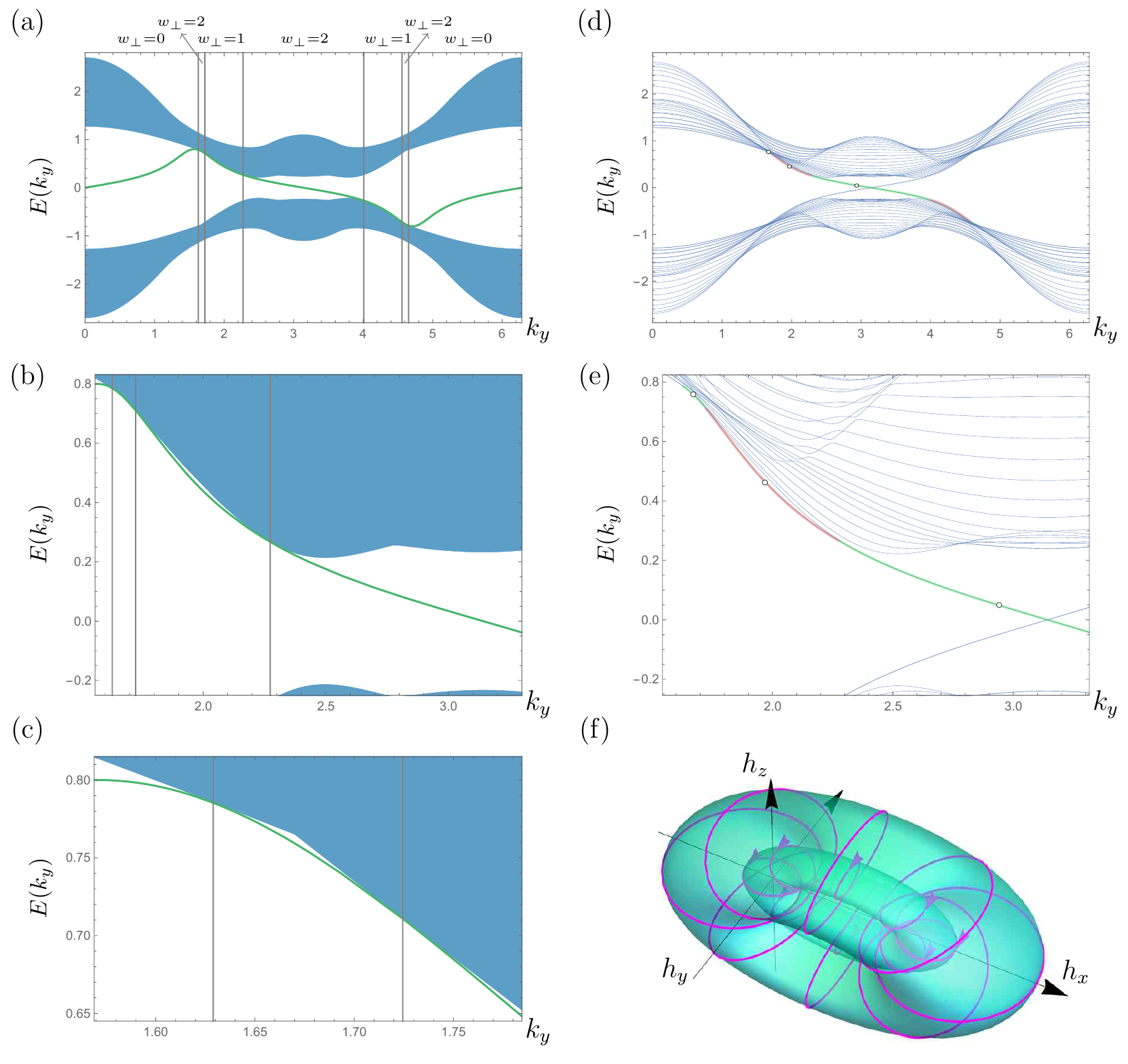}

\caption{\textbf{Semi-special case II} with \eqnref{semi-special case II}. (a): The anticipated regions of the bulk bands and the trajectory of $h_\bot(k_y)$ (solid curve). The winding number $w_\bot(k_y)$ around the $h'_z(k_y)$ axis for different portions of $k_y$ is indicated on the top. (b) and (c): The close-up views of (a). (d): The energy spectrum of a strip with $N_x=25$. The edge bands appear within the bulk gap, following the trajectories of $\pm h_\bot(k_y)$. The left edge bands are highlighted by thick curves. (e) The close-up view of (d). (f): $\mathbf{h}: \mathbf{k}\mapsto \mathbf{h}(\mathbf{k})$ illustrated as a torus embedded in the $\mathbf{h}$ space.}\label{fig:semi-special case II}
\end{figure}

\begin{figure}

\centering
    \includegraphics[width=0.95\textwidth]{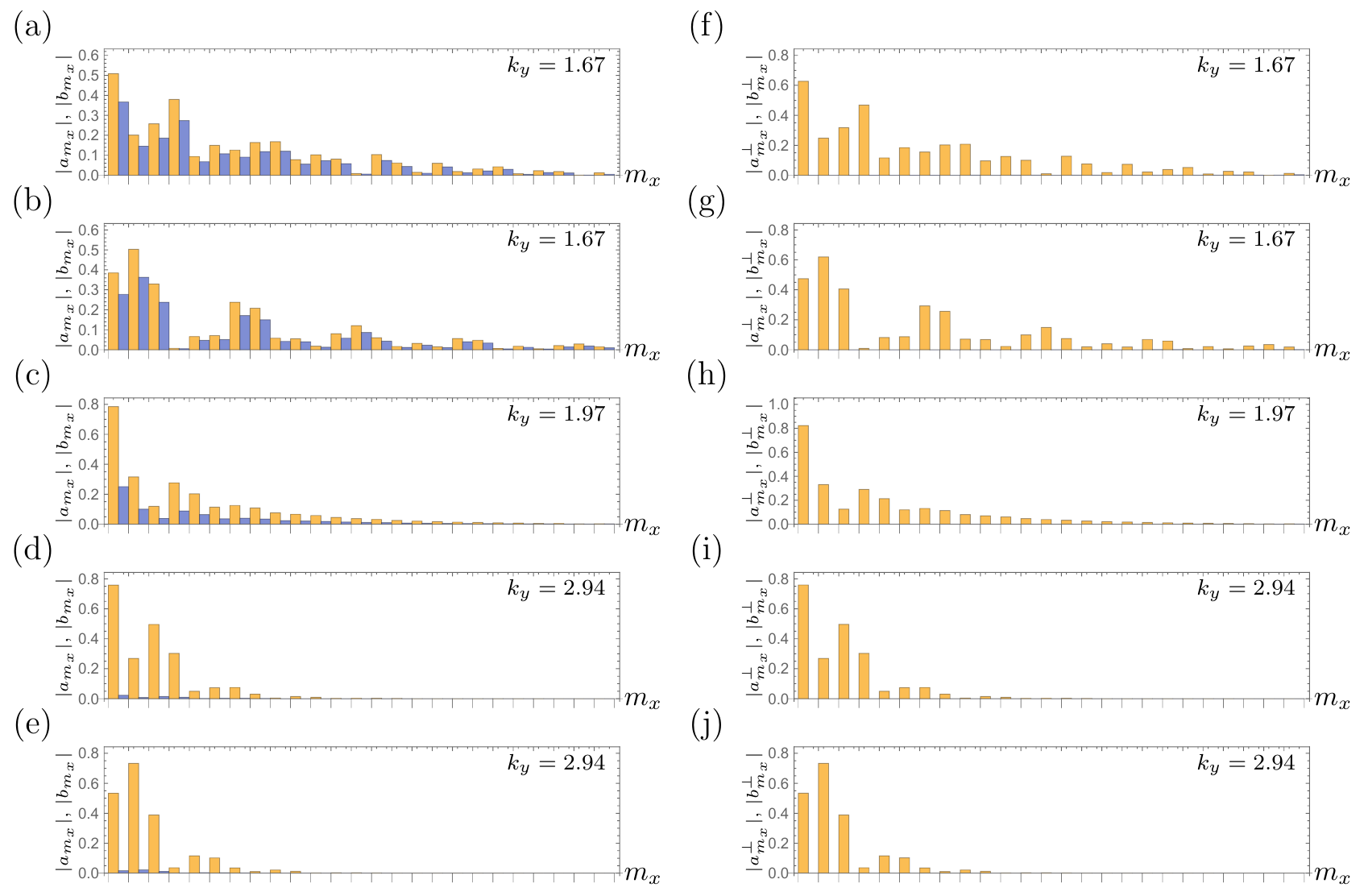}

\caption{\textbf{Semi-special case II} with \eqnref{semi-special case II} (Contiued). (a)--(e): The edge mode wavefunctions for the three dotted points as indicated in (d) and (e) of \figref{fig:semi-special case II}. They are localized at the left edge but not purely $a$-type or $b$-type. Note that the edge bands in the $w_\bot=2$ portions of $k_y$ (e.g., $k_y=1.67$ and $k_y=2.94$) are two-fold degenerate. (f)--(j): The same wave functions of (a)--(e) shown in terms of $\abs{a^\bot_{m_X}}$ and $\abs{b^\bot_{m_X}}$. They are purely $a^\bot$-type.}\label{fig:semi-special case II p2}
\end{figure}

\subsection{Generic case II}
We further deform the semi-special case \eqnref{semi-special case II} into a generic case given by
\begin{subequations}\label{generic case II}
\begin{eqnarray}
  h_x(k_x,k_y) &=& 0.8+\big(1-0.4\cos k_x + 0.5\cos (2k_x-\pi)\big)\cos k_y, \\
  h_y(k_x,k_y) &=& 0.4\sin k_x + 0.5\sin (2k_x-\pi), \\
  h_z(k_x,k_y) &=&  0.3 \big(1-0.4\cos k_x + 0.5\cos (2k_x-\pi)+0.3\sin k_x\big) \sin k_y.
\end{eqnarray}
\end{subequations}
In this case, each constant-$k_y$ loop no longer lies on a planar surface in the $\mathbf{h}$ space as can be seen in (d) and (e) of \figref{fig:generic case II}.
The trajectories of edge bands cannot be anticipated in advance, but the energy spectrum shown in (b) and (c) of \figref{fig:generic case II} can be understood as deformed from (d) and (e) of \figref{fig:semi-special case II}.
Particularly, the two-fold degenerate edge bands in the neighborhood of $k_y=\pi$ as shown in (d) of \figref{fig:semi-special case II} are split into two distinct edge bands as shown in (b) and (c) of \figref{fig:generic case II}, exhibiting the occurrence of $\text{(a)}\rightarrow \text{(b)}$ as depicted in \figref{fig:deformations2}.\footnote{This tells us that the condition for having a degenerate edge band is that the semi-special condition has to be satisfied and that the edge band multiplicity give by $\abs{w_\bot(k_y)}$ is greater than one.}
Furthermore, one of the split edge bands is merged with an upper-to-upper edge band in the manner depicted as $\text{(k)}\rightarrow \text{(l)}$ in \figref{fig:deformations2}. Finally, the two-fold degenerate upper-to-upper edge band appearing in the neighborhood of $k_y=1.67$ as shown in (e) of \figref{fig:semi-special case II} disappears in (b) and (c) of \figref{fig:generic case II}, exhibiting the reverse process as mentioned for (d) of \figref{fig:deformations2}.


\begin{figure}

\centering
    \includegraphics[width=0.95\textwidth]{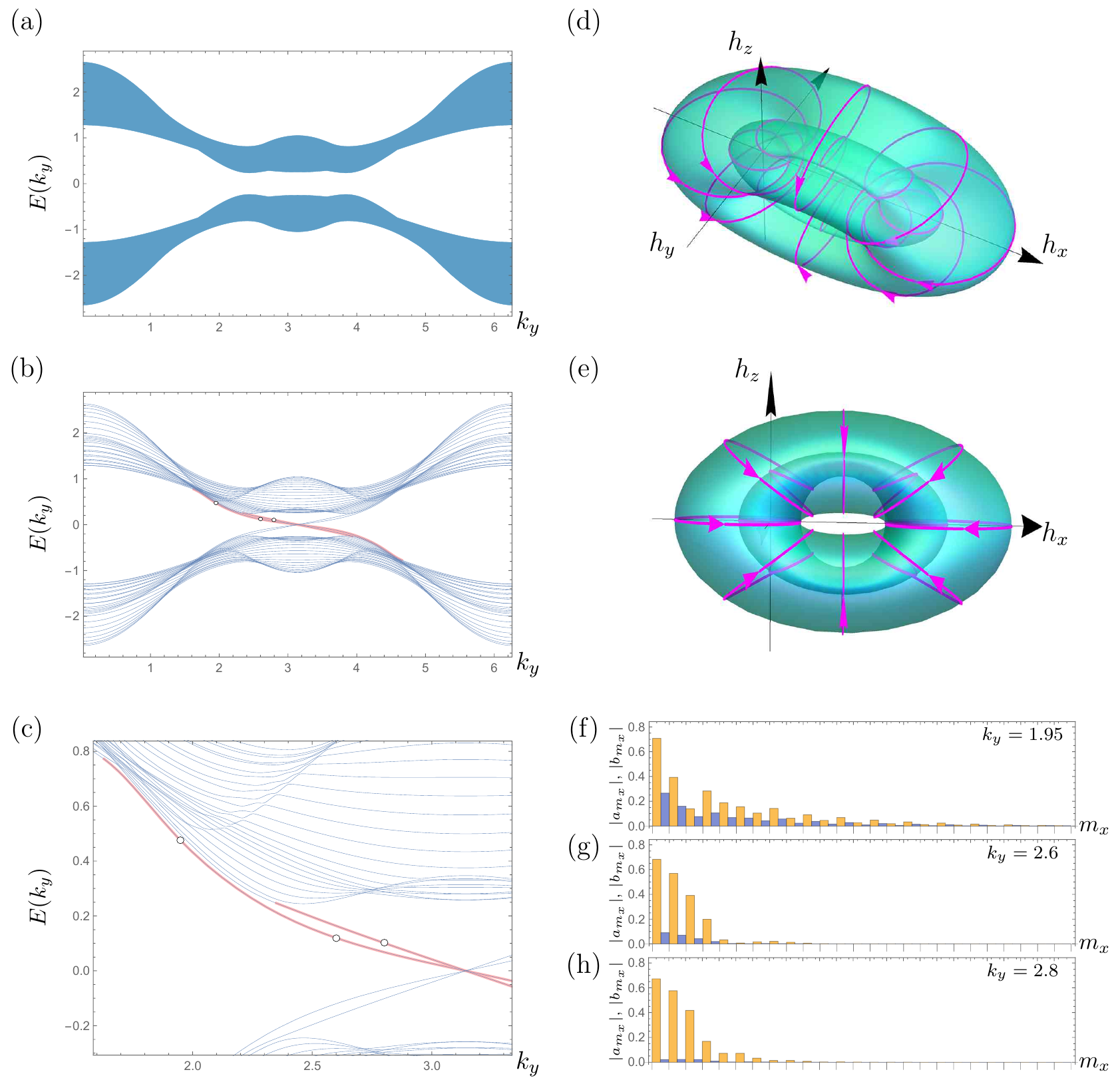}

\caption{\textbf{Generic case II} with \eqnref{generic case II}. (a): The anticipated regions of the bulk bands. (b): The energy spectrum of a strip with $N_x=25$. The edge bands appear within the bulk gap. The left edge bands are highlighted by thick curves. (c): The close-up view of (b). (d): $\mathbf{h}: \mathbf{k}\mapsto \mathbf{h}(\mathbf{k})$ illustrated as a torus embedded in the $\mathbf{h}$ space. (e): The side view of (d). (f)--(h): The edge mode wavefunctions for the three dotted points as indicated in (b) and (c). They are localized at the left edge but not purely $a$-type or $b$-type. Note that the two-fold degeneracy of edge bands in \figref{fig:semi-special case II} is now lifted.}\label{fig:generic case II}
\end{figure}

\subsection{Special case III}
Additionally, we consider a simple example that exhibits standalone edge bands. The function $\mathbf{h}(\mathbf{k})$ is given by
\begin{subequations}\label{special case III}
\begin{eqnarray}
  h_x(k_x,k_y) &=& 0.3+ (1-0.5\cos k_y)\cos k_x,  \\
  h_y(k_x,k_y) &=& (1-0.5\cos k_y)\sin k_x,  \\
  h_z(k_x,k_y) &=& h_z(k_y) = 0.5\sin k_y.
\end{eqnarray}
\end{subequations}
As shown in (c) and (d) of \figref{fig:special case III}, each constant-$k_y$ loop is a regular circle, lying on a plane perpendicular to the $h_z$ axis and winding around the $h_z$ axis once. As the torus of $\mathbf{h}(\mathbf{k})$ does not enclose the origin $\mathbf{h}=0$, the Chern number is $\mathtt{C}=0$. Furthermore, there are no south and north poles, so the trajectory of $h_z(k_y)$ does not touch the bulk band clusters. Consequently, it is expected that the trajectory of $h_z(k_y)$ gives rise to a standalone edge band with the multiplicity given by $\abs{w(k_y)}=1$. The numerical result indeed confirms the expectation as shown in \figref{fig:special case III}.


\begin{figure}
\centering
    \includegraphics[width=0.95\textwidth]{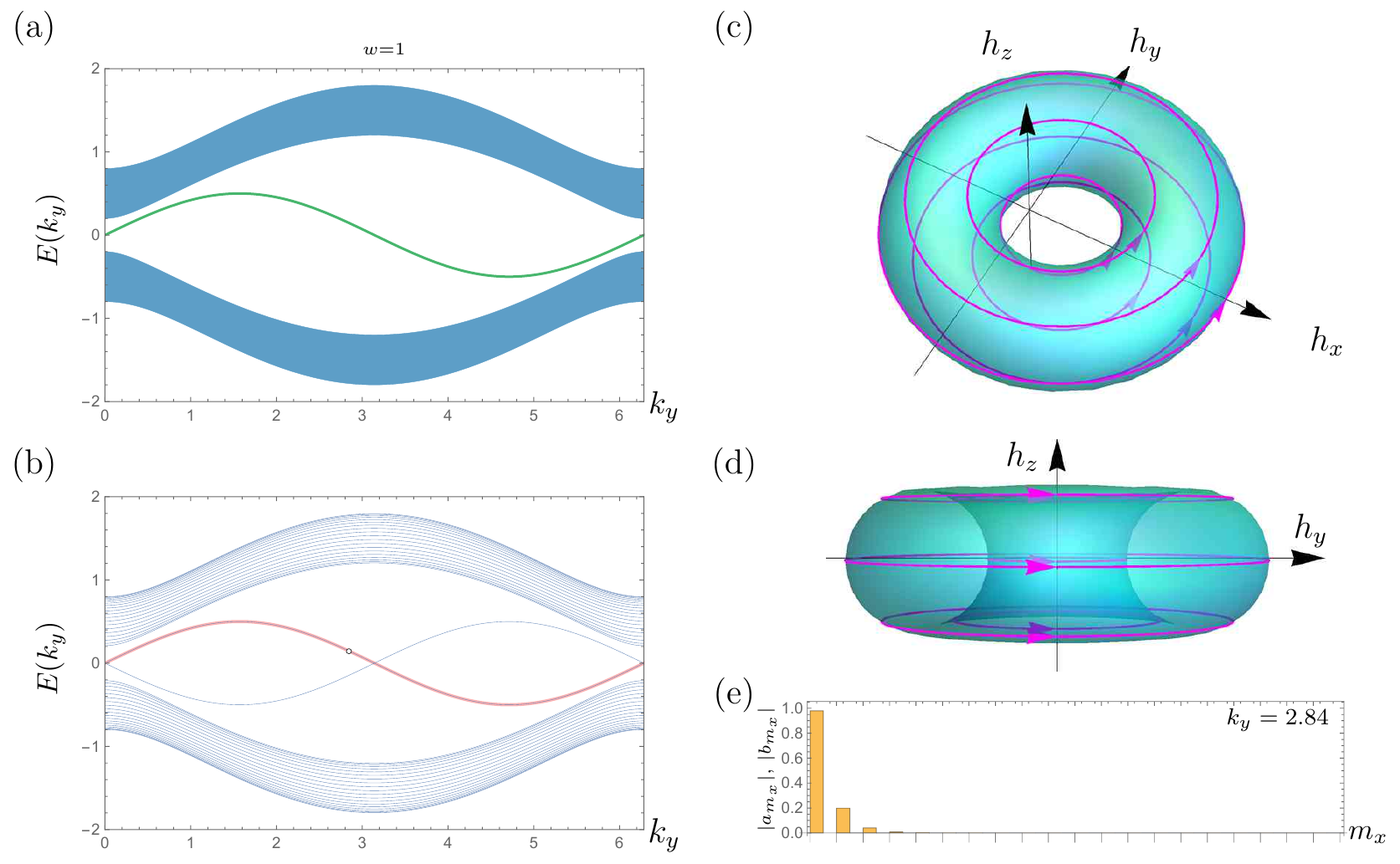}

\caption{\textbf{Special case III} with \eqnref{special case III}. (a): The anticipated regions of the bulk bands (shaded areas) and the trajectory of $h_z(k_y)$ (solid curve). The winding number $w(k_y)$ is $w=1$ everywhere. (b): The energy spectrum of a strip with $N_x=20$. The standalone edge bands appear within the bulk gap. The left edge band is highlighted by a thick curve. (c): $\mathbf{h}: \mathbf{k}\mapsto \mathbf{h}(\mathbf{k})$ illustrated as a torus embedded in the $\mathbf{h}$ space. (d) The side view of (c). (e): The edge mode wavefunction for the dotted point indicated in (b). It is localized at the left edge and purely $a$-type.}\label{fig:special case III}
\end{figure}

\subsection{Special case IV}\label{sec:special case IV}
Next, we consider a simple example with the function $\mathbf{h}(\mathbf{k})$ given by
\begin{subequations}\label{special case IV}
\begin{eqnarray}
  h_x(k_x,k_y) &=& 0.8+  (1-0.5\cos k_y) \cos k_x,\\
  h_y(k_x,k_y) &=&  (1-0.5\cos k_y)\sin k_x, \\
  h_z(k_x,k_y) &=& h_z(k_y) = 1+0.5\sin k_y,
\end{eqnarray}
\end{subequations}
which is deformed from \eqnref{special case III} simply by a translation in the $\mathbf{h}$ space as shown in (c) and (d) of \figref{fig:special case IV}.
The Chern number remains $\mathtt{C}=0$. As the trajectory of $h_z(k_y)$ now touches the upper bulk band cluster, the standalone edge bands in (b) of \figref{fig:special case III} become nonstandalone edge bands in (b) of \figref{fig:special case IV}, exhibiting the occurrence of $\text{(o)}\rightarrow\text{(n)}$ as depicted in \figref{fig:deformations2}.


\begin{figure}
\centering
    \includegraphics[width=0.95\textwidth]{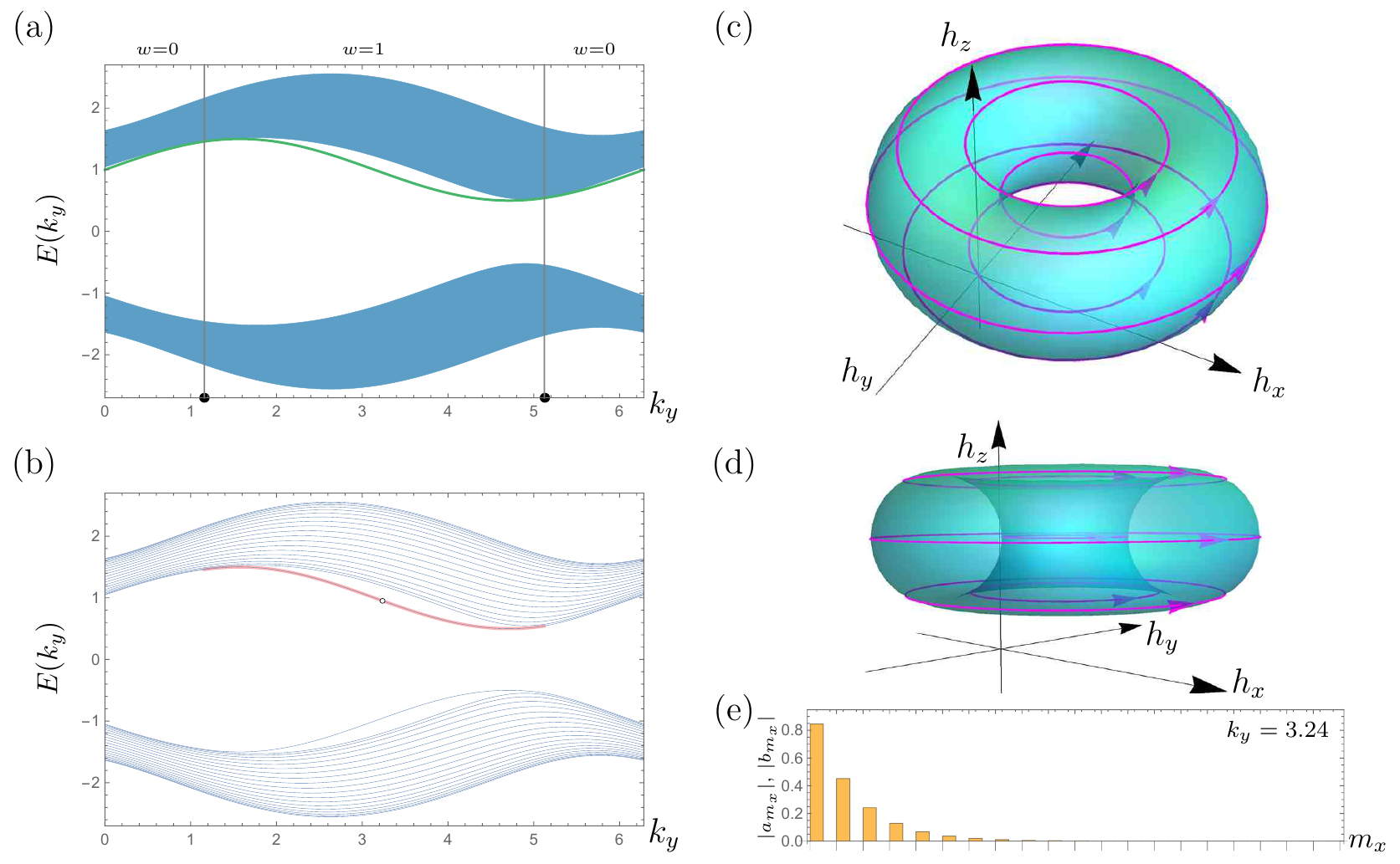}

\caption{\textbf{Special case IV} with \eqnref{special case IV}. (a): The anticipated regions of the bulk bands (shaded areas) and the trajectory of $h_z(k_y)$ (solid curve). The positions of the noth poles are indicated by solid dots. The winding number $w(k_y)$ for different portions of $k_y$ is indicated on the top. (b): The energy spectrum of a strip with $N_x=20$. The edge bands appear within the bulk gap. The left edge band is highlighted by a thick curve. (c): $\mathbf{h}: \mathbf{k}\mapsto \mathbf{h}(\mathbf{k})$ illustrated as a torus embedded in the $\mathbf{h}$ space. (d): The side view of (c). (e): The edge mode wavefunction for the dotted point indicated in (b). It is localized at the left edge and purely $a$-type.}\label{fig:special case IV}
\end{figure}

\subsection{Special case V}\label{sec:special case V}
Finally, we also study a special case with the function $\mathbf{h}(\mathbf{k})$ given by
\begin{subequations}\label{special case V}
\begin{eqnarray}
  h_x(k_x,k_y) &=& 0.5-0.6 \sin k_y + \cos 2k_y +0.5 \cos k_x, \\
  h_y(k_x,k_y) &=& 0.5 \sin k_x, \\
  h_z(k_x,k_y) &=& h_z(k_y) = 0.3\cos k_y + 0.5 \sin 2k_y.
\end{eqnarray}
\end{subequations}
As visualized in (c) of \figref{fig:special case V}, embedding of $\mathbf{h}(\mathbf{k})$ in the $\mathbf{h}$ space is similar to (c) of \figref{fig:special case I} except that now the torus coils around the origin $\mathbf{h}=0$ only once, thus yielding the Chern number $\mathtt{C}=1$.
The energy spectrum with $N_x=15$ is depicted in (b) of \figref{fig:special case V}. The schematic illustration of \figref{fig:band scheme} is affirmed: The edge bands follow the trajectories of $\pm h_z(k_y)$ with the multiplicity given by $\abs{w(k_y)}$. They are purely $a$-type or $b$-type.

This provides an example of a special case that possesses three kinds (upper-to-upper, lower-to-lower, and upper-to-lower) of left edge modes and serves as a good testing ground for adding edge perturbation. In the following two subsections, we will impose edge perturbation upon this case to demonstrate what has been discussed in \secref{sec:uniform edge perturbation}.


\begin{figure}

\centering
    \includegraphics[width=0.95\textwidth]{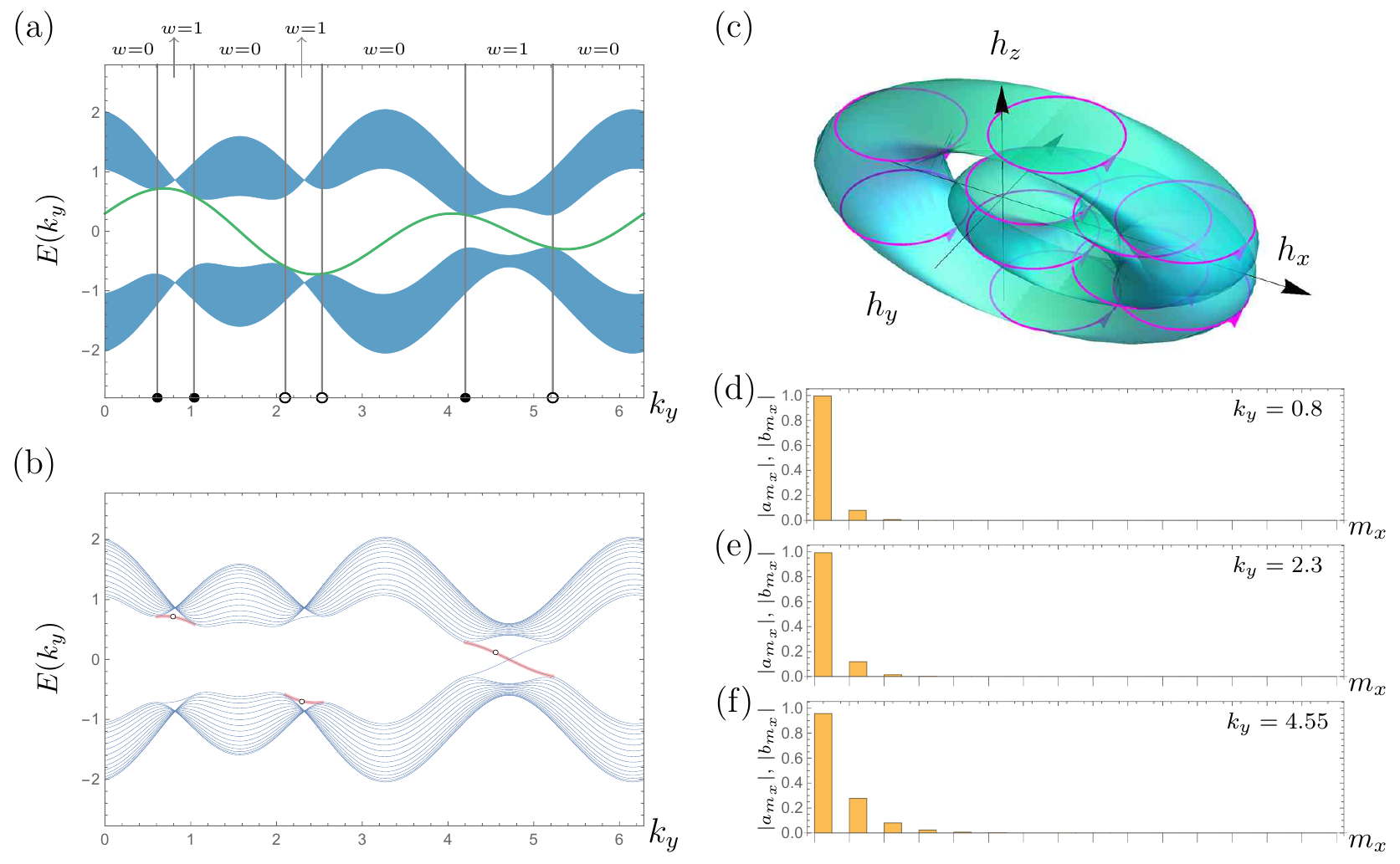}

\caption{\textbf{Special case V} with \eqnref{special case V}. (a): The anticipated regions of the bulk bands and the trajectory of $h_z(k_y)$. The positions of the north and south poles is indicated by solid and hollow dots, respectively. The winding number $w(k_y)$ for different portions of $k_y$ are indicated on the top. (b): The energy spectrum of a strip with $N_x=15$. The edge bands appear within the bulk gap, following the trajectories of $\pm h_z(k_y)$. The left edge bands are highlighted by thick curves. (c): $\mathbf{h}: \mathbf{k}\mapsto \mathbf{h}(\mathbf{k})$ illustrated as a torus embedded in the $\mathbf{h}$ space. (d)--(f): The edge mode wavefunctions for the three dotted points as indicated in (b). They are localized at the left edge and purely $a$-type.}\label{fig:special case V}
\end{figure}

\subsection{Edge perturbation I}\label{sec:edge perturbation I}
In the above examples so far, edge bands all appear within the bulk gap. We have not seen edge bands appear above the upper bulk band cluster, below the lower bulk band cluster, or inside the bulk band clusters (recall \footref{foot:only schematic}). Nor have we seen two left (right) edge bands appear simultaneously as in the configurations of (g) and (h) in \figref{fig:deformations2}. It turns out these possibilities can be produced by adding uniform edge perturbation, which is now studied in the present and next subsections.

For simplicity, we impose edge perturbation only on the left edge region upon the special case given by \eqnref{special case V}.
First, we consider the case that, in \eqnref{eqs''}, $\zeta^0_{m_x,m'_x;n_y}$ and $\zeta^z_{m_x,m'_x;n_y}$ all vanish but some of $\zeta^x_{m_x,m'_x;n_y}$ and $\zeta^y_{m_x,m'_x;n_y}$ are nonzero. Particularly, we study two sets of edge perturbation parameters given respectively by
\begin{subequations}\label{edge perturbation Ia}
\begin{eqnarray}
&&\delta\omega_{1,2}(k_y) = 0.8 \sin{k_y}+0.3\sin{3k_y}+0.8\sin{4k_y} \\
&&\delta\omega_{2,1}(k_y) = 0.8 \sin{k_y}+0.3\sin{3k_y}+0.8\sin{4k_y} ,\\
&&\delta\omega_{m,m'}(k_y)=0\quad \text{otherwise}, \\
&&\delta h^{\pm}_{m,m'}(k_y)  = 0,
\end{eqnarray}
\end{subequations}
and
\begin{subequations}\label{edge perturbation Ib}
\begin{eqnarray}
&&\text{same as \eqnref{edge perturbation Ia} except}\\
&&\delta\omega_{1,1}(k_y) = -0.1,
\end{eqnarray}
\end{subequations}
where $\delta\omega_{m,n}(k_y)$ and $\delta h^\pm_{m,n}(k_y)$ are defined as
\begin{subequations}\label{delta omega and h}
\begin{eqnarray}
\delta\omega_{m,m'}(k_y) &:=& \sum_{n_y=-\bar{n}_y}^{\bar{n}_y} e^{in_yk_y}
\left(\zeta^x_{m,m';n_y}-i\zeta^y_{m,m';n_y}\right), \\
\delta h^\pm_{m,m'}(k_y) &:=& \sum_{n_y=-\bar{n}_y}^{\bar{n}_y} e^{in_yk_y}
\left(\zeta^0_{m,m';n_y}\pm\zeta^z_{m,m';n_y}\right)
\end{eqnarray}
\end{subequations}
in relation to \eqnref{omega and h}.

The results of numerical computation for a strip with $N_x=15$ are presented in \figref{fig:edge perturbation I}: the left column for \eqnref{edge perturbation Ia} and the right column for \eqnref{edge perturbation Ib}. Comparing (a) and (b) of \figref{fig:edge perturbation I} with (b) of \figref{fig:special case V}, we see that the loci of the edge bands following the trajectories of $\pm h_z(k_y)$ remain unchanged, but meanwhile additional left edge bands induced by the edge perturbation also appear in various places --- within the bulk gap, above the upper bulk band cluster, below the lower bulk band cluster, and inside the bulk band clusters.
As predicted in \thmref{theorem 3}, both (a) and (b) are exactly symmetric under $E\rightarrow-E$.
The spectra of (a) and (b) are qualitatively similar to each other except that the deformation of $\text{(h)}\rightarrow \text{(g)}$ as depicted in \figref{fig:deformations2} takes place around $k_y=3\pi/2\approx4.71$ from (a) to (b) in \figref{fig:edge perturbation I}.

We depict edge-mode wavefunctions for a few various points in the spectra. In (c)--(h) of \figref{fig:edge perturbation I}, the wavefunctions of edge modes following the trajectories of $\pm h_z(k_y)$ are shown for the three opposite-energy pairs of dotted points as indicated in (a). They remain purely $a$-type or $b$-type. The left-edge-mode wavesfunctions of (c), (e), and (g) are different from the unperturbed counterparts (d)--(f) in \figref{fig:special case V}.
On the other hand, the right-edge-mode wavesfunctions of (d), (f), and (h) are unaffected by the left edge perturbation (except through minuscule finite size effect) and thus remain exactly dual to (d)--(f) in \figref{fig:special case V} via \eqnref{tilde Psi ky}.

In (i)--(n) of \figref{fig:edge perturbation I}, the wavefunctions of edge-perturbation-induced edge modes are shown for the three opposite-energy pairs of dotted points as indicated in (b). They are all localized at the left edge. In general, the edge-perturbation-induced edge modes are not purely $a$-type or $b$-type as shown in (k)--(n), except at the points where the edge band accidently intercepts the trajectories of $\pm h_z(k_y)$ such as shown in (i) and (j).
Furthermore, in (m) and (n), in addition to a peak localized at the left edge, the wave function also exhibits an almost periodic part over the bulk. This is because (m) and (n) are in the edge bands appearing inside the bulk band clusters and therefore the edge states are degenerate and mixed with bulk states.


\begin{figure}

\centering
    \includegraphics[width=0.95\textwidth]{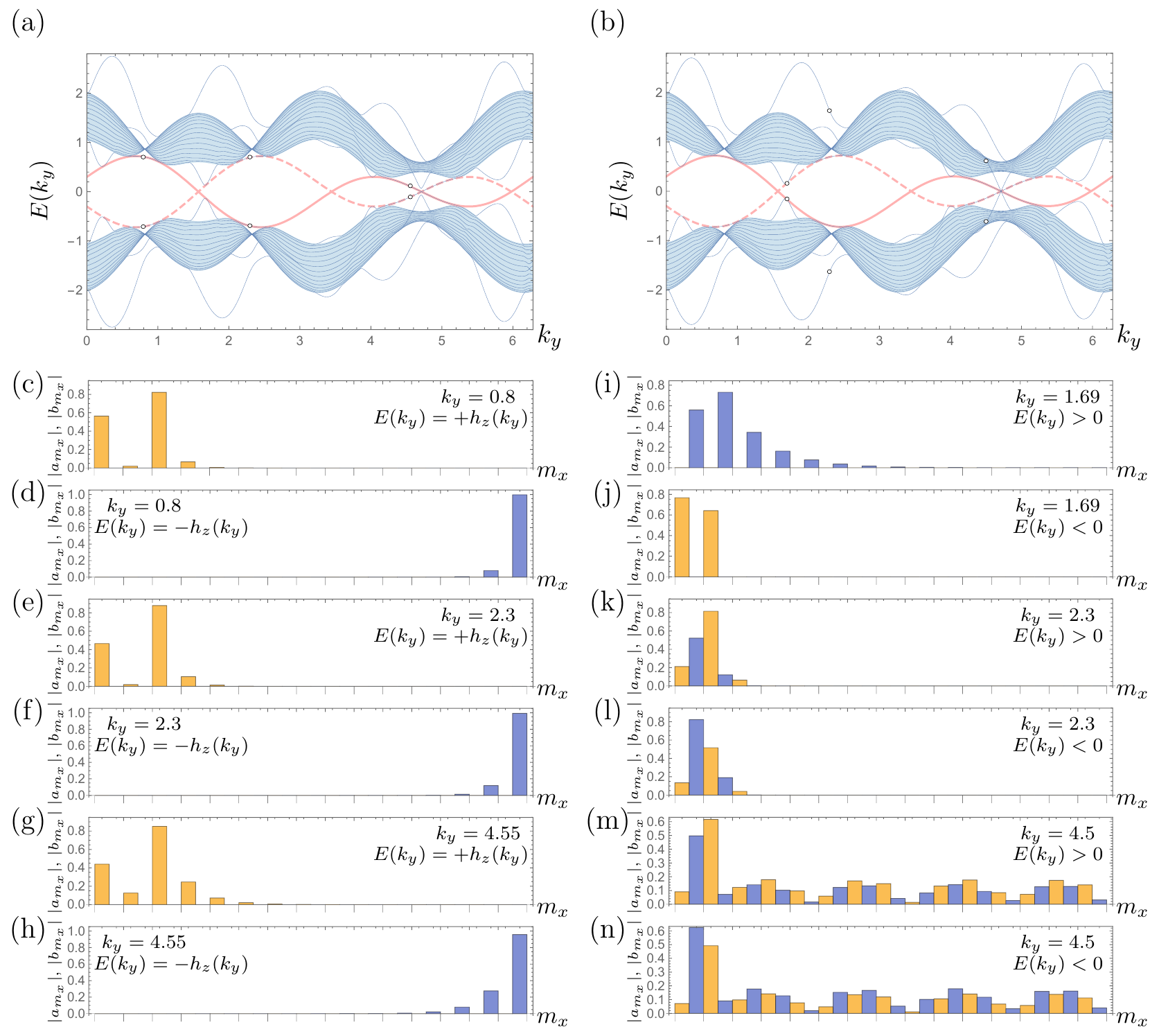}

\caption{\textbf{Edge perturbation I} with \eqnref{edge perturbation Ia} [left column] and with \eqnref{edge perturbation Ib} [right column] upon Special case V. (a) and (b): The energy spectra of a strip with $N_x=15$. The trajectories of $\pm h_z(k_y)$ are shown for reference as the thick-solid line and the thick-dashed line, respectively. (c)--(h): The wavefunctions of edge modes following the trajectories of $\pm h_z(k_y)$ for the three pairs of dotted points as indicated in (a). (i)--(n): The wavefunctions of edge-perturbation-induced edge modes for the three pairs of dotted points as indicated in (b)}\label{fig:edge perturbation I}
\end{figure}

\subsection{Edge perturbation II}\label{sec:edge perturbation II}
Next, we consider the case that, in \eqnref{eqs''}, some of $\zeta^0_{m_x,m'_x;n_y}$ and $\zeta^z_{m_x,m'_x;n_y}$ are nonzero.
Particularly, we study the set of edge perturbation parameters given by
\begin{subequations}\label{edge perturbation II}
\begin{eqnarray}
\delta h^+_{1,1}(k_y) &=&\delta h^{+*}_{1,1}(k_y) = 0.32 \sin{k_y}+0.48\sin{3k_y}+0.32\sin{4k_y},\\
\delta h^-_{1,1}(k_y) &=&\delta h^{-*}_{1,1}(k_y) = -0.32 \sin{k_y}-0.48\sin{3k_y}-0.32\sin{4k_y},\\
\delta h^\pm_{m,m'}(k_y) &=& 0\quad \text{otherwise}, \\
\delta\omega_{m,m'}(k_y) &=& 0.
\end{eqnarray}
\end{subequations}

The result of numerical computation for a strip with $N_x=15$ is presented in \figref{fig:edge perturbation II}. The right edge bands, which follow the trajectory of $-h_z(k_y)$, remain unaltered, whereas the left edge bands are now deviated from the trajectory of $h_z(k_y)$ and a few additional left edge bands also arise. The spectrum shown in (a) is no longer symmetric under $E\rightarrow-E$.

We depict left-edge-mode wavefunctions for a few various points in the spectrum in (b)--(f). They are all mixed in $a_{m_x}$ and $b_{m_x}$, even at the points where the left edge band accidentally intercepts the trajectories of $\pm h_z(k_y)$ such as shown in (e).


\begin{figure}

\centering
    \includegraphics[width=0.95\textwidth]{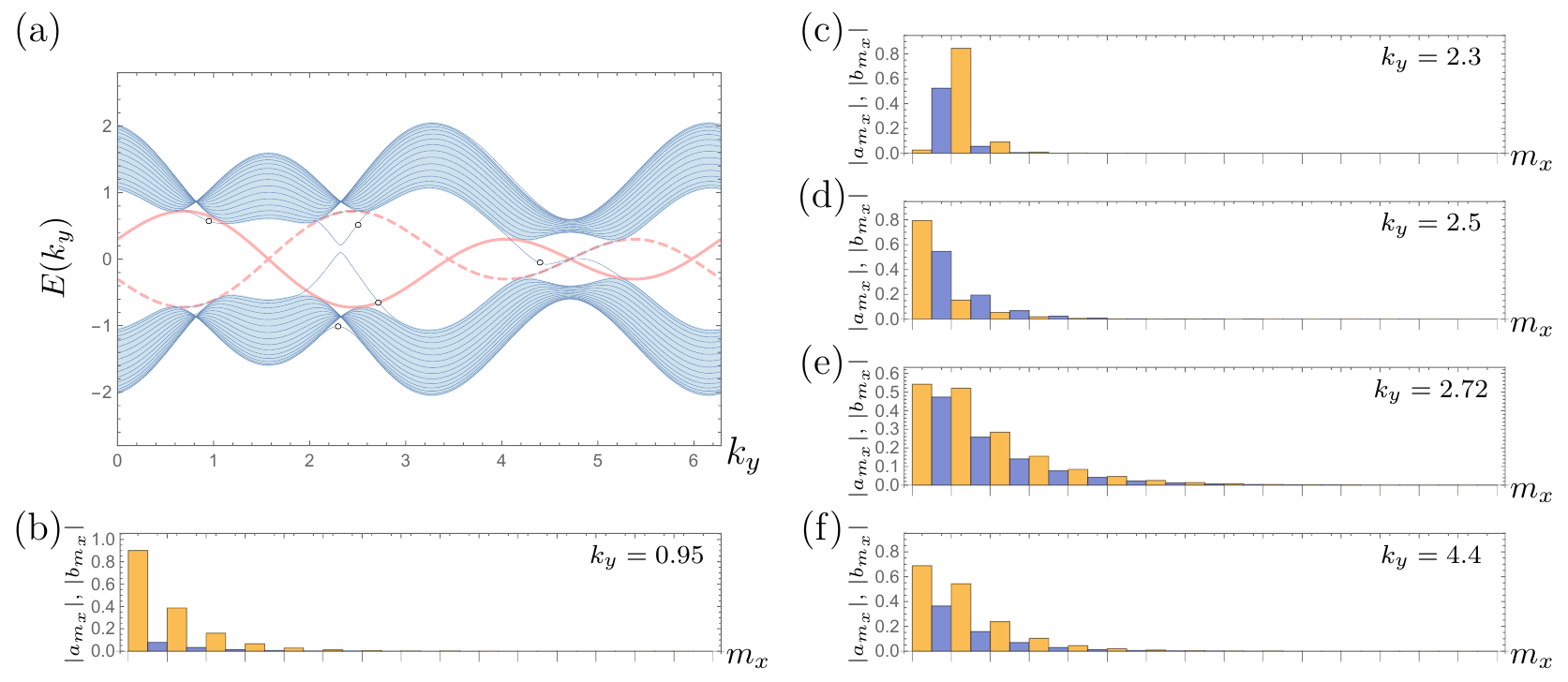}

\caption{\textbf{Edge perturbation II} with \eqnref{edge perturbation II} upon Special case V. (a): The energy spectrum of a strip with $N_x=15$. The trajectories of $\pm h_z(k_y)$ are shown as the thick-solid line and the thick-dashed line, respectively. (b)--(f): The left-edge-mode wavefunctions for the dotted points as indicated in (a).}\label{fig:edge perturbation II}
\end{figure}

\section{Summary and discussion}\label{sec:summary}
With the inclusion of arbitrary long-range hopping and (pseudo)spin-orbit coupling amplitudes, we construct a generic model for \emph{any} two-dimensional two-band Chern insulators as formulated in \eqnref{H(kx,ky)} with \eqnref{h(k)}. This provides a simple framework to investigate \emph{arbitrary} adiabatic deformations upon the systems of any \emph{arbitrary} Chern numbers. Without appealing to advanced techniques beyond the standard methods of solving linear difference equations and applying Cauchy's integral formula, we obtain a detailed description of the bulk-boundary correspondence on a strip, as stated in \thmref{theorem 1} and \thmref{theorem 2}, and a rigorous proof of it --- first for special cases (i.e., $h_z(\mathbf{k})=h_z(k_y)$), then with the inclusion of arbitrary uniform edge perturbation, and finally extended to generic cases.

We have proved the bulk-boundary correspondence only in the \emph{weak} form for a strip, but not yet the \emph{strong} form for a large sample in an arbitrary two-dimensional shape. The strong form posits that, on a large finite sample of a Chern insulator with a clean bulk but an arbitrary edge perimeter, the number of edge modes propagating along the perimeter (counterclockwise counted as positive, and clockwise as negative) is equal to the Chern number of the bulk. Once the weak form has been proved, the strong form can be implied by imagining a local edge portion of the large sample deformed into a straight strip and by the reasoning of unitarity. We refer readers to Section 6.3 of \cite{asboth2016short} for more details of the proof connecting the weak form to the strong form.

Our elementary approach not only is more transparent about the underlying physics of the bulk-boundary correspondence but also reveals various intriguing nontopological features of Chern insulators recapped in the following.
\begin{enumerate}[(i)]

\item It is rigorously shown in \secref{sec:eigenvalue problem} that while the bulk states depend on the width $N_x$ of the strip (recall \footref{foot:tomography}), the edge states are independent of $N_x$, except for negligible finite-size effect. This is an important feature of edge states often overlooked or taken for granted.

\item As long as the semi-special condition (i.e., the constant-$k_y$ loop lies on a plane in the $\mathbf{h}$ space) is satisfied (even only in a local open neighborhood of $k_y$), the trajectories of $\pm h_\bot(k_y)$ against $k_y$ give rise to edge bands with the multiplicity given by $\abs{w_\bot(k_y)}$, as discussed in \secref{sec:semi-special cases}.
    This observation is extremely useful, because it enables us to largely anticipate the loci of edge bands directly from $\mathbf{h}(\mathbf{k})$ without performing any full-fledged numerical computation for the energy spectrum, and therefore we can design at will a model with various desired features as demonstrated in many examples in \secref{sec:examples}.

\item Consequently, we also obtain the condition for having degenerate edge bands (i.e., multiple edge bands following the same trajectory in an interval of $k_y$): the map $\mathbf{h}(\mathbf{k})$ satisfies the semi-special condition in the interval and $\abs{w_\bot(k_y)}>1$.

\item We obtain a precise description of ``spin-momentum locking'' on a strip for semi-special cases: the (pseudo)spin is either parallel or antiparallel to the direction of $\hat{\mathbf{n}}(k_y)$ (i.e., the normal unit vector of the plane where the constant-$k_y$ loop lies) in edge states that follow the trajectories of $\pm h_\bot(k_y)$. However, it should be remarked that, contrary to popular opinion, the spin-momentum locking is not a topological feature in the strict sense, as it makes sense and is robust only under deformations within the confines of the semi-special condition.

\item Not only the bulk-boundary correspondence is shown to be robust against arbitrary uniform edge perturbation, but a finer differentiation between different kinds of edge perturbation is also revealed. The generic form of uniform edge perturbation imposed upon a special case is described in \eqnref{eqs''}. In case that $\zeta^0_{m_x,m'_x;n_y}$ and $\zeta^z_{m_x,m'_x;n_y}$ all vanish but some of $\zeta^x_{m_x,m'_x;n_y}$ and $\zeta^y_{m_x,m'_x;n_y}$ acting on the left (right) edge are nonzero, the left (right) edge bands following the trajectories of $\pm h_z(k_y)$ do not change their loci and remain purely $a$-type or $b$-type, although the corresponding wavefunctions are altered.
    On the other hand, in case that some of $\zeta^0_{m_x,m'_x;n_y}$ and $\zeta^z_{m_x,m'_x;n_y}$ acting on the left (right) edge are nonzero, the left (right) edge bands are deviated from the trajectories of $\pm h(k_y)$ and no longer purely $a$-type or $b$-type. The edge bands following the trajectories of $\pm h(k_y)$ are robust against $\zeta^x_{m_x,m'_x;n_y}$ and $\zeta^y_{m_x,m'_x;n_y}$, but sensitive to $\zeta^0_{m_x,m'_x;n_y}$ and $\zeta^z_{m_x,m'_x;n_y}$.
    Whether $\zeta^0_{m_x,m'_x;n_y}$ and $\zeta^z_{m_x,m'_x;n_y}$ are zero or not, the inclusion of edge perturbation in general also gives rise to more edge bands, which in general are not purely $a$-type or $b$-type.

\item If $h_0=0$, the strip Hamiltonian $\hat{H}_{N_x}(k_y)$ exhibits the $h_0=0$ symmetry as elaborated in \secref{sec:h0=0 symmetry'}. The $h_0=0$ symmetry relates an energy eigenstate $\ket{\Psi}$ of $E=E_0$ to a counterpart energy eigenstate $\ket{\tilde{\Psi}}$ of $E=-E_0$ via \eqnref{tilde Psi ky}, which can be understood as inherited from \eqnref{tilde a b} for $\hat{H}_\mathrm{bulk}$.
    Particularly, the $h_0=0$ symmetry associates a left (right) edge mode with a right (left) edge mode of the opposite energy.
    When edge perturbation is introduced, however, the $h_0=0$ symmetry is broken.
    Nevertheless, if edge perturbation is imposed upon a special case as formulated in \eqnref{eqs''} in a particular way that all of $\zeta^0_{m_x,m'_x;n_y}$ and $\zeta^z_{m_x,m'_x;n_y}$ are zero and only some of $\zeta^x_{m_x,m'_x;n_y}$ and $\zeta^y_{m_x,m'_x;n_y}$ are nonzero, the energy eigenvalues of the strip still appear in pairs with opposite signs as predicted in \thmref{theorem 3}. In the absence of edge perturbation, this pairing symmetry is identical to the $h_0=0$ symmetry; in the presence of edge perturbation, however, this pairing symmetry should not be confused with the $h_0=0$ symmetry.
    For the edge modes following the trajectories of $\pm h_z(k_y)$, the symmetry of \thmref{theorem 3} associates a \emph{left} edge mode of $E=\pm h_z$ with a \emph{right} edge mode of $E=\mp h_z$, but the corresponding wavefunctions are no longer related with each other via \eqnref{tilde Psi ky}.
    For edge-perturbation-induced edge modes, the symmetry of \thmref{theorem 3} associates a \emph{left} (right) edge mode of $E=E_0$ with a different \emph{left} (right) edge mode of $E=-E_0$.

\end{enumerate}

Finally, it is noteworthy that our proof of the bulk-boundary correspondence for the two-dimensional two-band model is essentially based on a dimension-reduction scheme that recasts the Chern number in terms of winding numbers as given in \eqnref{Chern number 2} and consequently enables us to employ the same techniques devised for the one-dimensional generalized SSH model in our previous work \cite{chen2017elementary}. This suggests that, via a proper dimension-reduction procedure, our elementary approach to the problem of bulk-boundary correspondence might be applicable to other topological systems with richer structure or in higher dimensions.

\begin{acknowledgments}
BHC was supported in part by the Ministry of Science and Technology, Taiwan under the Grant MOST 104-2112-M-003-003-MY3; DWC was supported under MOST 106-2112-M-110-010 and MOST 107-2112-M-110-003.
\end{acknowledgments}

\appendix

\section{Eqs.~\eqnref{eqs'} and \eqnref{eqs''} in a matrix form}\label{appendix}
Solving the coupled difference equation \eqnref{eqs'} and its modified version with uniform edge perturbation as given in \eqnref{eqs''} can be viewed as solving the eigenvalue problem of a $2N_x\times2N_x$ matrix corresponding to $\hat{H}_{N_x}(k_y)$.
For convenience, define
\begin{subequations}\label{omega and h}
\begin{eqnarray}
\omega_{m,m'}(k_y) &:=& \sum_{n_y=-\bar{n}_y}^{\bar{n}_y} e^{in_yk_y}\left(
\left(w^x_{m'-m,n_y}-iw^y_{m'-m,n_y}\right)
+\left(\zeta^x_{m,m';n_y}-i\zeta^y_{m,m';n_y}\right)\right), \\
h^\pm_{m,m'}(k_y) &:=& \pm h_z(k_y)\delta_{m,m'}+\sum_{n_y=-\bar{n}_y}^{\bar{n}_y} e^{in_yk_y}
\left(\zeta^0_{m,m';n_y}\pm\zeta^z_{m,m';n_y}\right),
\end{eqnarray}
\end{subequations}
where in fact $w_{m'-m,n_y}\neq0$ only if $\abs{m'-m}\leq\bar{n}_x$ (because $\bar{n}_x$ gives the upper bound for the distance of the long-range interaction), $\zeta^a_{m,m';n_y}\neq0$ only if $m,m'\leq\bar{n}_x$ or $m,m'\geq N_x-\bar{n}_x+1$ (because $\zeta^a_{m,m';n_y}$ are edge perturbation parameters), and consequently $\omega_{m,m'},h^\pm_{m,m'}\neq0$ only if $\abs{m-m'}\leq\bar{n}_x$.
By \eqnref{h(k)b} and \eqnref{zeta}, it follows
\begin{subequations}\label{omega and h 2}
\begin{eqnarray}
\omega^*_{m',m}(k_y) &:=& \sum_{n_y=-\bar{n}_y}^{\bar{n}_y} e^{in_yk_y}\left(
\left(w^x_{m'-m,n_y}+iw^y_{m'-m,n_y}\right)
+\left(\zeta^x_{m,m';n_y}+i\zeta^y_{m,m';n_y}\right)\right), \\
h^{\pm*}_{m',m}(k_y) &:=& \pm h_z(k_y)\delta_{m,m'}+\sum_{n_y=-\bar{n}_y}^{\bar{n}_y} e^{in_yk_y}
\left(\zeta^0_{m,m';n_y}\pm\zeta^z_{m,m';n_y}\right) = h^\pm_{m,m'}.
\end{eqnarray}
\end{subequations}
In the basis $\{\ket{1\uparrow},\ket{1,\downarrow},\dots,\ket{N\uparrow},\ket{N\downarrow}\}$, the matrix of $\hat{H}_N(k_y)$ takes the form
\begin{equation}
\left(
  \begin{array}{cccccc}
    \begin{array}{cc}
      h^+_{1,1} & \omega_{1,1} \\
      \omega^*_{1,1} & h^-_{1,1}
    \end{array}
    &
    \begin{array}{cc}
      h^+_{1,2} & \omega_{1,2} \\
      \omega^*_{2,1} & h^-_{1,2}
    \end{array} & \cdots &  & \cdots
    &
    \begin{array}{cc}
     h^+_{1,N} & \omega_{1,N} \\
     \omega^*_{N,1} & h^-_{1,N}
    \end{array}\\
    \begin{array}{cc}
      h^+_{2,1} & \omega_{2,1} \\
      \omega^*_{1,2} & h^-_{2,1}
    \end{array}
    &
    \begin{array}{cc}
     h^+_{2,2} & \omega_{2,2} \\
     \omega^*_{2,2} & h^-_{2,2}
    \end{array}
    &  &  & & \vdots\\
    \vdots &  & \ddots & & & \\
     &  &  &  & & \vdots\\
    \vdots &  &  &
    &
    \begin{array}{cc}
       h^+_{N-1,N-1} & \omega_{N-1,N-1} \\
       \omega^*_{N-1,N-1} & h^-_{N-1,N-1}
    \end{array}
    &
    \begin{array}{cc}
       h^+_{N-1,N} & \omega_{N-1,N} \\
       \omega^*_{N,N-1} & h^-_{N-1,N}
    \end{array}\\
    \begin{array}{cc}
       h^+_{N,1} & \omega_{N,1} \\
       \omega^*_{1,N} & h^-_{N,1}
    \end{array}
    & \cdots &  & \cdots
    &
    \begin{array}{cc}
       h^+_{N,N-1} & \omega_{N,N-1} \\
       \omega^*_{N-1,N} & h^-_{N,N-1}
    \end{array}
    &
    \begin{array}{cc}
    h^+_{N,N} & \omega_{N,N} \\
    \omega^*_{N,N} & h^-_{N,N}
    \end{array}
  \end{array}
\right)
\end{equation}
Shuffling the basis order into $\{\ket{1\uparrow},\dots,\ket{N\uparrow},\ket{1,\downarrow},\dots,\ket{N\downarrow}\}$, we can represent the matrix of $\hat{H}_N(k_y)$ in a more succinct block form
\begin{equation}
\left(
\begin{array}{cc}
  H^+ & \Omega \\
  \Omega^\dag & H^-
\end{array}
\right),
\end{equation}
where $\Omega$ and $H^\pm$ are $N\times N$ matrices with $\Omega_{ij}=\omega_{ij}$ and $H^\pm_{i,j}=h^\pm_{i,j}$.

In case all $\zeta^0_{m_x,m'_x;n_y}$ and $\zeta^z_{m_x,m'_x;n_y}$ are zero, the matrix of $\hat{H}_N(k_y)$ takes a special form
\begin{equation}\label{special matrix}
\left(
\begin{array}{cc}
  h_z\mathbbm{1}_{N\times N} & \Omega \\
  \Omega^\dag & -h_z\mathbbm{1}_{N\times N}
\end{array}
\right),
\end{equation}
which exhibits a particular symmetry as addressed in the following theorem.
\begin{theorem}\label{theorem 3}
If a matrix $H_N$ takes the form of \eqnref{special matrix}, $\det(H_N-\lambda\mathbbm{1}_{2N\times 2N})$ is invariant under $\lambda\rightarrow-\lambda$. In other words, the eigenvalues of $\hat{H}_N$ always appear in pairs with opposite signs if $\zeta^0_{m_x,m'_x;n_y}=0$ and $\zeta^z_{m_x,m'_x;n_y}=0$.
\end{theorem}
\begin{proof}
By Schur's determinant identity, we have
\begin{eqnarray}
&&\det(H_N-\lambda\mathbbm{1}_{2N\times 2N})
\equiv
\left|
  \begin{array}{cc}
  A:=(h_z-\lambda)\mathbbm{1}_{N\times N} & B:=\Omega \\
  C:=\Omega^\dag & D:=(-h_z-\lambda)\mathbbm{1}_{N\times N}
  \end{array}
\right| \nonumber\\
&=& \abs{D}\abs{A-BD^{-1}C}
=\Big|(-h_z-\lambda)\mathbbm{1}_{N\times N}\Big|\,
\abs{(h_z-\lambda)\mathbbm{1}_{N\times N}-\Omega\frac{\mathbbm{1}_{N\times N}}{(-h_z-\lambda)}\Omega^\dag} \nonumber\\
&=&(-h_z-\lambda)^N (h_z-\lambda)^N \abs{\mathbbm{1}_{N\times N}-\frac{\Omega\Omega^\dag}{(h_z-\lambda)(-h_z-\lambda)}},
\end{eqnarray}
which is obviously invariant under $\lambda\rightarrow-\lambda$.
\end{proof}

If there is no edge perturbation at all (i.e., all $\zeta^a_{m_x,m_x';n_y}$ are zero), the symmetry of \thmref{theorem 3} is nothing but the $h_0=0$ symmetry as discussed in \secref{sec:h0=0 symmetry'}, which relates an eigenstate $\ket{\Psi}$ of $\lambda=E$ to the counterpart $\ket{\tilde{\Psi}}$ of $\lambda=-E$ via \eqnref{tilde Psi ky}.
When edge perturbation is introduced, however, the $h_0=0$ symmetry is broken and the symmetry of \thmref{theorem 3} should not be confused with the $h_0=0$ symmetry. An eigenstate of $\lambda=E$ and its counterpart of $\lambda=-E$ is no longer related via the simple relation \eqnref{tilde Psi ky}, but via a complicated relation depending on $\omega_{m,m'}$.


\bibliography{mybib}
\bibliographystyle{ieeetr}

\end{document}